\documentclass[
aps,
prd,
preprint,
titlepage,
floatfix,
superscriptaddress,
nofootinbib,
longbibliography
]{revtex4-2}

\usepackage{amsmath}
\usepackage{amssymb}
\usepackage{amsthm}
\usepackage{bm}
\usepackage{mathtools}
\usepackage{graphicx}
\usepackage{booktabs}
\usepackage{xcolor}
\usepackage{microtype}

\DeclareMathOperator{\Tr}{Tr}
\DeclareMathOperator{\rank}{rank}
\DeclareMathOperator{\ran}{ran}
\DeclareMathOperator{\diag}{diag}

\newcommand{\QFI}{\mathcal{Q}}
\newcommand{\CFI}{\mathcal{F}}
\newcommand{\dcp}{\delta_{\mathrm{CP}}}
\newcommand{\dd}{\mathrm{d}}

\theoremstyle{plain}
\newtheorem{theorem}{Theorem}[section]

\newtheorem{proposition}[theorem]{Proposition}
\newtheorem{corollary}[theorem]{Corollary}

\theoremstyle{remark}

\usepackage[hidelinks]{hyperref}

\hypersetup{
  pdftitle={
    Operational quantum estimation theory for neutrino oscillations:
    identifiability, attainability, and spectral precision bounds
  },
  pdfauthor={Jianlong Lu},
  pdfsubject={
    Quantum estimation theory, multiparameter identifiability,
    attainability, and detector-level inference for neutrino oscillations
  },
  pdfkeywords={
    quantum estimation theory,
    quantum Fisher information,
    neutrino oscillations,
    multiparameter estimation,
    identifiability,
    Holevo bound,
    spectral information,
    DUNE
  },
  pdfdisplaydoctitle=true
}

\begin{document}

\title{Operational quantum estimation theory for neutrino oscillations: identifiability, attainability, and spectral precision bounds}

\author{Jianlong Lu}
\email{jianlong@nus.edu.sg}
\affiliation{
Department of Mathematics, National University of Singapore,
Singapore 119076
}

\date{\today}

\begin{abstract}
Quantum Fisher information (QFI) bounds state-encoded precision before
measurement choice but does not establish identifiability, joint
attainability, or detector sensitivity.  We develop an operational framework
for three-flavor oscillations that separates state, measurement, and
reconstructed-event information.  At fixed baseline, energy, source flavor,
and matter profile, the propagated state is a pure qutrit, so its
six-coordinate QFI has rank at most four.  Resolved broadband components can
restore rank because the aggregate kernel equals the intersection of their
active kernels.  We analyze joint attainability using quantum curvature,
exact pure-state Holevo costs, and numerical primal--dual brackets, and
propagate information through flavor projection,
detector response, Poisson sampling, and nuisance profiling.  An independent
implementation of the public DUNE GLoBES configuration reproduces the
reference spectra to relative error $3.87\times10^{-16}$ and the profiled
likelihood curvature to $3.26\times10^{-6}$.  Under the declared scaling and
tolerance, its 264-bin event information has numerical rank six at all 176
documented physics points, whereas any four retained rule totals have rank at
most four.  The weakest record has effective rank five at the declared
practical threshold.  A conditional likelihood pilot also shows finite-grid
overcoverage and dependence on the auxiliary-measurement ensemble.  The
framework identifies when local quantum or Fisher bounds do not support
global experimental-sensitivity claims.
\end{abstract}

\makeatletter
\@booleanfalse\preprintsty@sw
\makeatother
\maketitle
\makeatletter
\@booleantrue\preprintsty@sw
\makeatother

\section{Introduction}
\label{sec:introduction}

Neutrino flavor conversion is a coherent interference phenomenon and one of
the principal sources of evidence that neutrinos are massive.  In the standard
three-flavor description, propagation is governed by two independent
mass-squared splittings, three mixing angles, and one Dirac CP phase, with
matter effects introducing an additional dependence on the propagation
environment
\cite{Pontecorvo1958,MakiNakagawaSakata1962,Wolfenstein1978,
MikheyevSmirnov1985,ParticleDataGroup2026}.
Determining these parameters precisely is central to tests of the three-flavor
paradigm and to searches for leptonic CP violation, the neutrino mass ordering,
and deviations from standard propagation.  This inference problem is
intrinsically multiparameter: appearance and disappearance probabilities
depend on correlated combinations of the oscillation coordinates, the phase
$\dcp$ is periodic, $\theta_{23}$ admits an octant ambiguity, and the mass
ordering is a discrete hypothesis.  Current global analyses consequently
exhibit non-Gaussian and ordering-dependent parameter profiles
\cite{NuFIT2024}, while future broadband experiments are designed to resolve
these structures using the joint energy dependence of neutrino and
antineutrino appearance and disappearance samples
\cite{DUNETDR2020,DUNEPhysics2020}.

Conventional sensitivity analyses begin with predicted event spectra and a
statistical model for reconstructed detector observations.  Their local
curvature is described by a classical Fisher information matrix, while
nonlocal confidence regions are obtained from a likelihood or a
Monte Carlo-calibrated likelihood-ratio construction.  Quantum estimation
theory asks a complementary question: how much local information about a
parameter is encoded in a specified family of quantum states before a
particular measurement is imposed?  For every fixed positive-operator-valued
measure (POVM), its classical Fisher information matrix is bounded above by
the symmetric-logarithmic-derivative quantum Fisher information matrix.  More
precisely, along any one-dimensional tangent direction $\bm{v}$, the supremum
of the corresponding scalar classical Fisher information over admissible
measurements is
$\bm{v}^{\mathsf T}\QFI\bm{v}$.  In a multiparameter problem, however, the
measurements optimizing different tangent directions need not coincide.
Consequently, the QFI defines a natural statistical metric and a collection
of directional precision bounds, but its inverse need not be a jointly
attainable covariance matrix
\cite{BraunsteinCaves1994,Paris2009,Holevo2011,Matsumoto2002}.
It is therefore a useful benchmark only after the state family,
parameterization, estimands, number and allocation of state preparations, and
allowed measurements have been fixed.  It is not, by itself, a detector
sensitivity or a globally valid confidence-region calculation.

For neutrino experiments, this distinction leads naturally to a hierarchy of
statistical models,
\begin{equation}
  \rho_{\boldsymbol{\lambda}}(L,E)
  \xrightarrow{\{M_y\}}
  p_y(E|\boldsymbol{\lambda})
  \xrightarrow{\Phi,\,R_{b|y,E},\,\epsilon,\,\boldsymbol{\nu}}
  \mu_b(\boldsymbol{\lambda},\boldsymbol{\nu})
  \xrightarrow{\mathrm{Poisson}+\mathrm{constraints}}
  \mathcal{L}(\boldsymbol{\lambda},\boldsymbol{\nu}),
  \label{eq:introduction_operational_chain}
\end{equation}
where $\rho_{\boldsymbol{\lambda}}$ is the propagated neutrino state,
$\{M_y\}$ is a physical measurement such as flavor projection, $\Phi$ denotes
the incident spectrum and exposure, $R_{b|y,E}$ is the detector response,
$\epsilon$ denotes efficiencies, $\boldsymbol{\nu}$ collects nuisance
parameters, and $\mu_b$ is the expected reconstructed count in bin $b$.
The state QFI, the classical Fisher information associated with the fixed
measurement $\{M_y\}$, the Poisson event Fisher information, and the
nuisance-profiled event information answer different operational questions.
Conflating these layers can turn a valid quantum bound into an invalid claim
about experimentally achievable precision.

Quantum-estimation methods were applied to neutrino oscillations in a
two-flavor setting by Nogueira \emph{et al.}, who compared the information
encoded in oscillating neutrino states with that accessible through flavor
measurements \cite{NogueiraEtAl2017}.  More recent work has extended this
program to three-flavor oscillations, CP-phase estimation, coherent and
incoherent state models, reconstructed accelerator-event spectra, and
multiparameter degeneracy diagnostics
\cite{IgnotiEtAl2025,FrugiueleEtAl2026,ChundawatLi2026,
ChundawatEtAl2026,YadavDegeneracy2026,HuangEtAl2026}.
Frugiuele \emph{et al.} compared state QFI with the classical Fisher
information of ideal flavor measurements for accelerator and reactor states
in vacuum, demonstrating that flavor projection can be close to optimal for
some mixing angles while extracting only a limited fraction of the
information encoded about $\dcp$
\cite{FrugiueleEtAl2026}.  Chundawat \emph{et al.} compared intrinsic quantum
information about $\dcp$ with information in reconstructed T2K and NO$\nu$A
event spectra \cite{ChundawatEtAl2026}.  Yadav \emph{et al.} studied
the three-coordinate submodel
$(\theta_{23},\dcp,\Delta m^2_{31})$ and used fidelity and the QFI matrix to
distinguish quantum states associated with probability-degenerate parameter
points \cite{YadavDegeneracy2026}.  Huang \emph{et al.} considered all six
standard oscillation coordinates, including the role of
parameter-dependent basis transformations and off-diagonal QFI correlations
\cite{HuangEtAl2026}.  Together, these works establish quantum estimation as a
useful language for neutrino parameter encoding.  They also sharpen the need
to distinguish local quantum geometry from simultaneous attainability,
spectral identifiability, and detector-level inference.

The first unresolved issue is identifiability.  At a fixed preparation and
propagation setting, a normalized pure three-flavor state defines a ray in
$\mathbb{CP}^{2}$, whose real dimension is four.  In vacuum, the propagation
setting may be parameterized by $L/E$; in matter, the baseline, energy, and
matter profile must in general be specified separately.  Consequently, the
image of any six-coordinate pure-state oscillation model at one such setting
has dimension at most four, and its QFI matrix has rank at most four.  Its
ordinary inverse therefore cannot represent six simultaneous precision
bounds, irrespective of the numerical size of its diagonal elements.  The
correct local objects are the identifiable tangent space, estimable parameter
functions, and the information metric restricted to that space.  Broadband
observations can change this conclusion because different energies may
possess different null directions.  The information aggregated over resolved
spectral components can acquire a larger rank even though every individual
fixed-energy pure-state QFI remains rank deficient.  Whether this restoration
survives depends additionally on whether the energy label is retained,
coarsened, or traced out.

The second issue is multiparameter attainability.  In a one-parameter regular
model, the symmetric-logarithmic-derivative quantum Cramér--Rao bound is
asymptotically attainable under standard local assumptions.  In a
multiparameter model, the measurements that optimize different coordinates
need not be compatible.  Matrix invertibility therefore does not imply the
existence of a common measurement attaining the inverse-QFI covariance
simultaneously.  The appropriate analysis requires the quantum geometric
curvature or weak-commutativity condition and, when incompatibility remains,
the weighted Holevo Cramér--Rao bound
\cite{Holevo2011,Matsumoto2002,Ragy2016,RagyErratum2019,AlbarelliEtAl2020,
DemkowiczEtAl2020}.
Moreover, scalar precision costs require a declared weight matrix: the
numerical comparison of parameters measured in radians and
$\mathrm{eV}^{2}$ is otherwise coordinate dependent.

A third issue concerns representation covariance.  A passive change of basis
cannot alter the information encoded by a physical state.  If the basis
transformation itself depends on the estimated parameters, however, ordinary
coordinate derivatives acquire a connection term.  Omitting that term changes
the quantum geometric tensor and can produce a spurious basis dependence.
Related care is required for antineutrino propagation.  CP conjugation acts as
a parameter-space pullback, including the Jacobian associated with
$\dcp\mapsto-\dcp$, while matter propagation also reverses the sign of the
matter potential.  Diagonal tensor components alone do not reveal the
resulting mixed-index sign transformations.

In this work, we address these issues in a single operational framework.  Our
main contributions are as follows.

\begin{enumerate}
  \item We formulate neutrino parameter estimation as an explicit
  state--measurement--event hierarchy.  State QFI, fixed-measurement flavor
  information, reconstructed-event Fisher information, Gaussian
  auxiliary-constraint information, and nuisance-profiled likelihood
  curvature are defined separately and connected only through declared
  physical maps.

  \item We prove representation covariance of the projected quantum geometric
  tensor under parameter-dependent passive basis transformations when the
  required connection is included.  We also derive the full tensor pullback
  under CP conjugation and distinguish this coordinate transformation from
  the physical reversal of the matter potential.

  \item We prove the pure-state rank bound
  $\rank\QFI\leq 2(d-1)$ and apply it to the six-coordinate pure qutrit model.
  For broadband information matrices
  $\QFI_{\mathrm{broad}}=\sum_k N_k\QFI_k$ with positive weights, we establish
  \begin{equation}
    \ker \QFI_{\mathrm{broad}}
    = \bigcap_k \ker \QFI_k ,
    \label{eq:introduction_kernel_intersection}
  \end{equation}
  which gives an exact criterion for spectral rank restoration.  We then
  distinguish a resolved classical energy label from an incoherent state
  obtained by tracing over that label and quantify the associated information
  loss.

  \item We evaluate multiparameter attainability on the identifiable tangent
  space.  This includes compatibility diagnostics, exact pure-state Holevo
  costs where available, and finite-dimensional semidefinite programs with
  tolerance-qualified numerical primal--dual brackets for mixed spectral
  states.  These calculations separate a small inverse-SLD cost from a
  precision cost that is actually attainable under the stated asymptotic
  measurement model.

  \item We independently reproduce and reassess the fixed-energy
  six-coordinate calculation of Ref.~\cite{HuangEtAl2026}, and then carry the
  resulting framework to reconstructed detector spectra.  For the
  checksum-locked public DUNE GLoBES configuration
  \cite{HuberEtAl2005,KoppEtAl2007,DUNEPublicConfig2021}, our independent
  forward model reproduces all 640 GLoBES-exported rule/bin signal and
  background entries with relative $\ell_2$ error $3.87\times10^{-16}$.  Its
  nuisance-profiled likelihood curvature agrees with the independent GLoBES
  calculation with relative $\ell_2$ error $3.26\times10^{-6}$.  The joint
  264-bin event information has numerical rank six at the declared relative
  $10^{-10}$ threshold at all 176 points of the documented physics envelope,
  which forms the physics subset of a 226-record robustness ledger, whereas
  replacing the spectra by four rule totals gives numerical rank four.  The
  weakest spectral direction is nevertheless effectively rank deficient at a
  declared effective-information threshold of $10^{-6}$ after the coordinate
  scaling, demonstrating that numerical rank and useful precision are not
  equivalent.
\end{enumerate}

The novelty of the present work is therefore not the broad observation that
QFI can differ from flavor-measurement information, nor the use of Fisher
information for a single oscillation coordinate.  It is the integrated
connection between representation-covariant quantum geometry, exact
identifiability, attainable multiparameter costs, spectral-label processing,
and a validated detector likelihood.  This connection also clarifies the
scope of every numerical result: theorem-level statements and exact analytic
results are kept separate from tolerance-qualified numerical computations,
public-configuration benchmarks, and diagnostic likelihood pilots.

The experimental study deliberately uses the simplified public configuration
accompanying the DUNE simulation note.  That configuration was released for
phenomenological studies and produces sensitivities similar, but not
identical, to the official Technical Design Report analysis
\cite{DUNEPublicConfig2021}.  It contains named normalization uncertainties
rather than the full official covariance and near-detector constraint model.
Accordingly, all DUNE quantities reported here are local
public-configuration benchmarks.  We do not interpret the inverse event Fisher
matrix as a globally valid confidence region or quote an official ordering or
CP sensitivity.  An exact conditional likelihood bank, paired
pseudo-experiment pilot, and expanded four-coordinate Asimov profile are
instead used to diagnose finite-grid resolution, dependence on the
auxiliary-measurement ensemble, and failures of automatic Wilks thresholds.
The expanded profile remains conditional on fixed solar coordinates.  Such
caution is required for periodic parameters, discrete orderings, boundaries,
and constrained nuisances, as emphasized independently by the NO$\nu$A
profile-construction study \cite{NOvAProfiledFC2022}.

The remainder of the paper is organized as follows.  Section~\ref{sec:model}
defines the operational statistical model and fixes the oscillation,
antineutrino, matter, and parameter conventions.  Section~\ref{sec:geometry}
develops the projected quantum geometric tensor, representation covariance,
and CP-conjugation laws.  Section~\ref{sec:identifiability} proves the
fixed-energy rank and broadband kernel-intersection results and analyzes loss
of the spectral label.  Section~\ref{sec:attainability} treats
multiparameter compatibility, exact pure-state Holevo results, and numerical
primal--dual Holevo brackets.  Section~\ref{sec:detector_likelihood} derives
the flavor, detector, Poisson, and nuisance-profiled information layers.
Section~\ref{sec:reassessment} presents the independent reproduction and
geometric interpretation of the fixed-energy calculation.  Section~\ref{sec:dune} gives
the public-DUNE validation, information analysis, robustness study, and
conditional likelihood diagnostics.  Finally,
Secs.~\ref{sec:discussion} and~\ref{sec:conclusions} discuss the distinction
between mathematical and practical identifiability, state the limitations of
the present experimental model, and summarize the results.

\section{Operational statistical model and conventions}
\label{sec:model}

Quantum Fisher information is meaningful only relative to a specified state
family, set of estimands, and class of admissible measurements.  Experimental
precision requires still more structure: source spectra, interaction rates,
detector response, backgrounds, nuisance parameters, and a sampling model.
We therefore define the complete inference problem before introducing its
quantum geometry.  This separation prevents quantities associated with
different statistical experiments from being compared as if they were
interchangeable.

\subsection{Estimands, controls, and the operational model}

The oscillation-parameter vector is
\begin{equation}
  \bm{\lambda}
  =
  \left(
    \theta_{12},
    \theta_{13},
    \theta_{23},
    \dcp,
    \Delta m^2_{21},
    \Delta m^2_{31}
  \right)^{\mathsf T}.
  \label{eq:parameter-vector}
\end{equation}
Angles and phases are expressed in radians and the mass-squared differences
in $\mathrm{eV}^2$.  We take $\Delta m^2_{21}>0$; the sign of
$\Delta m^2_{31}$ distinguishes normal and inverted ordering.  The CP phase
is periodic, $\dcp\equiv\dcp+2\pi$, so Eq.~\eqref{eq:parameter-vector}
defines a local coordinate chart rather than a globally Euclidean parameter
space.

We distinguish three kinds of quantities throughout:

\begin{enumerate}
  \item the estimands $\bm{\lambda}$ whose local precision is being studied;
  \item fixed controls $\bm{c}$, such as baseline, true energy, source flavor,
        horn polarity, matter profile, and exposure; and
  \item nuisance coordinates $\bm{\nu}$ describing uncertain flux,
        cross-section, response, calibration, background, and normalization
        parameters.
\end{enumerate}

For a fixed experimental specification, the operational model may be written
schematically as
\begin{equation}
  \mathfrak{M}
  =
  \left(
    \{\rho_{\alpha,0}\},
    \mathcal{E}_{\bm{c},\bm{\lambda}},
    \{M_y\},
    \mathcal{X},
    R,
    \epsilon,
    B,
    \mathcal{L}_{\mathrm{aux}}
  \right).
  \label{eq:operational-model}
\end{equation}
Here $\rho_{\alpha,0}$ is the source state of flavor $\alpha$,
$\mathcal{E}_{\bm{c},\bm{\lambda}}$ is the propagation channel,
$\{M_y\}$ is a declared quantum measurement, $\mathcal{X}$ collects exposure,
flux, cross-section, and true-bin factors, $R$ is the classical migration
kernel, $\epsilon$ denotes efficiencies, $B$ the background model, and
$\mathcal{L}_{\mathrm{aux}}$ the likelihood for auxiliary constraints on
$\bm{\nu}$.

Equation~\eqref{eq:operational-model} defines a sequence
\begin{equation}
  \rho_{\alpha,0}
  \xrightarrow{\ \mathcal{E}_{\bm{c},\bm{\lambda}}\ }
  \rho_{\alpha k}
  \xrightarrow{\ \{M_y\}\ }
  p_{y|\alpha k}
  \xrightarrow{\ \mathcal{X},R,\epsilon,B\ }
  \mu_b
  \xrightarrow{\ \mathrm{Poisson}\ }
  n_b ,
  \label{eq:operational-chain}
\end{equation}
where $k$ labels true-energy or other preparation bins, $y$ labels quantum
measurement outcomes, $b$ is a compound reconstructed-event index, $\mu_b$
is the predicted event mean, and $n_b$ is the observed count.  Each arrow
answers a different statistical question and may discard information.

A passive representation change, including a parameter-dependent choice of
basis, is not an additional arrow in Eq.~\eqref{eq:operational-chain}.  It must
transform the states, measurements, and their derivatives consistently and
cannot change an operational information tensor.  By contrast, physically
applying a parameter-dependent unitary while keeping the measurement fixed
defines a new encoding channel and hence a different model
$\mathfrak{M}$.  The representation-covariance proof is given in
Sec.~\ref{sec:geometry}.

\subsection{Three-flavor propagation conventions}

We use the standard three-flavor mixing convention
\cite{Pontecorvo1958,MakiNakagawaSakata1962}
\begin{equation}
  U
  =
  R_{23}U_{13}(\dcp)R_{12},
  \label{eq:pmns-factorization}
\end{equation}
or explicitly
\begin{equation}
  U =
  \begin{pmatrix}
    c_{12}c_{13}
      & s_{12}c_{13}
      & s_{13}e^{-i\dcp}
      \\[2pt]
    -s_{12}c_{23}-c_{12}s_{23}s_{13}e^{i\dcp}
      & c_{12}c_{23}-s_{12}s_{23}s_{13}e^{i\dcp}
      & s_{23}c_{13}
      \\[2pt]
    s_{12}s_{23}-c_{12}c_{23}s_{13}e^{i\dcp}
      & -c_{12}s_{23}-s_{12}c_{23}s_{13}e^{i\dcp}
      & c_{23}c_{13}
  \end{pmatrix},
  \label{eq:pmns-matrix}
\end{equation}
where $s_{ij}=\sin\theta_{ij}$ and $c_{ij}=\cos\theta_{ij}$.
An irrelevant common mass has been removed by setting
\begin{equation}
  D_m
  =
  \diag
  \left(0,\Delta m^2_{21},\Delta m^2_{31}\right).
  \label{eq:mass-convention}
\end{equation}

In the flavor basis, neutrino propagation through matter satisfies
\begin{equation}
  i\frac{\dd}{\dd x}|\psi(x)\rangle
  =
  H_\nu(E,x;\bm{\lambda})|\psi(x)\rangle ,
  \label{eq:schrodinger-equation}
\end{equation}
with
\begin{align}
  H_\nu(E,x;\bm{\lambda})
  &=
  \frac{1}{2E}
  \left[
    U D_m U^\dagger
    +
    \diag\!\left(a(E,x),0,0\right)
  \right],
  \label{eq:neutrino-hamiltonian}
  \\
  a(E,x)
  &=
  2E V_{\mathrm{CC}}(x)
  =
  2\sqrt{2}\,G_{\mathrm F}N_e(x)E .
  \label{eq:matter-potential}
\end{align}
This is the standard coherent forward-scattering potential
\cite{Wolfenstein1978,MikheyevSmirnov1985}.  For density
$\rho_{\mathrm m}$ in $\mathrm{g\,cm^{-3}}$, electron fraction $Y_e$, and
energy in $\mathrm{GeV}$, the effective mass-squared contribution is
approximately
\begin{equation}
  a
  =
  1.526\times10^{-4}\,
  Y_e
  \left(\frac{\rho_{\mathrm m}}{\mathrm{g\,cm^{-3}}}\right)
  \left(\frac{E}{\mathrm{GeV}}\right)
  \mathrm{eV}^2 .
  \label{eq:numerical-matter-potential}
\end{equation}

Antineutrino propagation applies both required transformations,
\begin{equation}
  H_{\bar{\nu}}(E,x;\bm{\lambda})
  =
  \frac{1}{2E}
  \left[
    U^*D_mU^{\mathsf T}
    -
    \diag\!\left(a(E,x),0,0\right)
  \right].
  \label{eq:antineutrino-hamiltonian}
\end{equation}
Thus CP conjugation is not implemented by complex conjugating $U$ alone in
matter.

For constant density, the propagated state from source flavor $\alpha$ is
\begin{equation}
  |\psi_{\alpha k}(\bm{\lambda})\rangle
  =
  S_k(\bm{\lambda})|\nu_\alpha\rangle,
  \qquad
  S_k
  =
  \exp\!\left[-iH(E_k;\bm{\lambda})L\right].
  \label{eq:propagated-state}
\end{equation}
Equation~\eqref{eq:propagated-state} is written in natural units,
$\hbar=c=1$.  Whenever $L$ is supplied in kilometres, it is converted
consistently to inverse-energy units before the matrix exponential is
evaluated.

In vacuum, when $L$ is in $\mathrm{km}$, $E$ in $\mathrm{GeV}$, and
$\Delta m^2$ in $\mathrm{eV}^2$, the relative amplitude phases are written
as
\begin{equation}
  \exp\!\left[
    -2i\kappa
    \frac{\Delta m^2 L}{E}
  \right],
  \qquad
  \kappa=1.267
  \label{eq:numerical-vacuum-phase}
\end{equation}
in the controlled reproduction calculations.  The public-experiment
calculation instead uses the conversion constants of the independently
validated GLoBES implementation
\cite{HuberEtAl2005,KoppEtAl2007,DUNEPublicConfig2021}; conversion prescriptions are not mixed
within an analysis.

\subsection{State and fixed-measurement information}

Let $\rho(\bm{\lambda})$ be a differentiable state model.  The symmetric
logarithmic derivatives $L_i$ satisfy
\begin{equation}
  \partial_i\rho
  =
  \frac{1}{2}\left(L_i\rho+\rho L_i\right),
  \qquad
  \partial_i\equiv\frac{\partial}{\partial\lambda_i}.
  \label{eq:sld-definition}
\end{equation}
The corresponding SLD quantum Fisher information matrix
\cite{Helstrom1976,Holevo2011} is
\begin{equation}
  \QFI_{ij}
  =
  \frac{1}{2}
  \Tr\!\left[
    \rho\left(L_iL_j+L_jL_i\right)
  \right].
  \label{eq:sld-qfi}
\end{equation}
For a normalized pure state $\rho=|\psi\rangle\langle\psi|$, this becomes
\begin{equation}
  \QFI_{ij}
  =
  4\,\operatorname{Re}
  \left[
    \langle\partial_i\psi|
    \left(\mathbb{I}-|\psi\rangle\langle\psi|\right)
    |\partial_j\psi\rangle
  \right].
  \label{eq:pure-state-qfi}
\end{equation}
The projector removes the unobservable phase direction.  The associated
imaginary tensor, which controls local multiparameter compatibility, is
introduced in Sec.~\ref{sec:geometry}.

For a fixed POVM $\{M_y\}$,
\begin{equation}
  p_y(\bm{\lambda})
  =
  \Tr[\rho(\bm{\lambda})M_y],
  \qquad
  \sum_y p_y=1,
  \label{eq:povm-probabilities}
\end{equation}
and its classical Fisher information is
\begin{equation}
  \CFI^{(M)}_{ij}
  =
  \sum_{y:p_y>0}
  \frac{
    \partial_i p_y\,\partial_j p_y
  }{p_y}.
  \label{eq:measurement-cfi}
\end{equation}
At a zero-probability point, terms in
Eq.~\eqref{eq:measurement-cfi} must not be discarded solely because both the
numerator and denominator vanish.  If the model admits a finite,
path-independent interior limit, that limit defines the corresponding local
contribution.  If the limit is direction dependent or divergent, the point is
nonregular and must instead be analyzed using directional limits,
higher-order expansions, or a nonlocal likelihood construction.

For any fixed POVM,
\begin{equation}
  \CFI^{(M)}
  \preceq
  \QFI ,
  \label{eq:qfi-cfi-order}
\end{equation}
in the positive-semidefinite order
\cite{BraunsteinCaves1994,Holevo2011}.  This inequality means that
$\bm{v}^{\mathsf T}\CFI^{(M)}\bm{v}
\leq\bm{v}^{\mathsf T}\QFI\bm{v}$ for every local tangent direction
$\bm{v}$.  Equality is measurement- and model-dependent and does not imply
simultaneous attainability of the inverse SLD matrix in a multiparameter
problem.

The ideal flavor measurement is
\begin{equation}
  M_\beta
  =
  |\nu_\beta\rangle\langle\nu_\beta|,
  \qquad
  p_{\beta|\alpha k}
  =
  |\langle\nu_\beta|\psi_{\alpha k}\rangle|^2 .
  \label{eq:flavor-measurement}
\end{equation}
If $N_k$ parameter-independent probes are prepared in resolved bin $k$, the
exposure-weighted state and flavor information matrices are
\begin{equation}
  \QFI_{\mathrm{state}}
  =
  \sum_k N_k\QFI^{(k)},
  \qquad
  \CFI_{\mathrm{flavor}}
  =
  \sum_k N_k\CFI_{\mathrm{flavor}}^{(k)} .
  \label{eq:resolved-information}
\end{equation}
Equation~\eqref{eq:resolved-information} describes an experiment in which the
preparation label $k$ is retained.  It also conditions on the source model.
If the production rate or the bin weights depend on an estimand, their rate
information must be included separately rather than being silently absorbed
into the conditional state QFI.

\subsection{Forward-folded event likelihood}

A general reconstructed-event mean can be expressed as
\begin{equation}
  \mu_b(\bm{\lambda},\bm{\nu})
  =
  \sum_{r,\alpha,\beta,k}
  \mathcal{X}_{r\alpha\beta k}(\bm{\nu})\,
  \epsilon_{r\beta k}(\bm{\nu})\,
  R_{b|r\beta k}(\bm{\nu})\,
  P_{\alpha\rightarrow\beta}
    (E_k;\bm{\lambda},\bm{\nu})
  +
  B_b(\bm{\nu}).
  \label{eq:event-mean}
\end{equation}
The index $r$ collects running mode, interaction channel, and analysis rule.
The factor $\mathcal{X}$ contains exposure, flux, cross-section, bin width,
and any fixed channel normalization.  Conditional on an event being
reconstructed and selected, the response is normalized as
\begin{equation}
  R_{b|r\beta k}\geq0,
  \qquad
  \sum_b R_{b|r\beta k}=1,
  \label{eq:response-normalization}
\end{equation}
while inefficiency is represented explicitly by $\epsilon$.  Equivalent
conventions using a column-substochastic response are allowed provided that
the choice is declared and used consistently.

For independent Poisson counts, the log likelihood is
\begin{equation}
  \ell_{\mathrm{data}}
  =
  \sum_b
  \left[
    n_b\ln\mu_b-\mu_b-\ln\Gamma(n_b+1)
  \right].
  \label{eq:poisson-log-likelihood}
\end{equation}
Introduce the joint coordinate vector
\begin{equation}
  \bm{\zeta}
  =
  \begin{pmatrix}
    \bm{\lambda}\\
    \bm{\nu}
  \end{pmatrix}.
\end{equation}
Using $q$ as the compound reconstructed-event index, the data Fisher matrix is
\begin{equation}
  F^{\mathrm{data}}_{ac}
  =
  \sum_q
  \frac{1}{\mu_q}
  \frac{\partial\mu_q}{\partial\zeta_a}
  \frac{\partial\mu_q}{\partial\zeta_c}.
  \label{eq:poisson-fisher}
\end{equation}

Auxiliary measurements or Gaussian constraints contribute
\begin{equation}
  F^{\mathrm{tot}}
  =
  F^{\mathrm{data}}
  +
  F^{\mathrm{aux}}.
  \label{eq:total-fisher}
\end{equation}
For nuisance pulls standardized to one constraint standard deviation,
$F^{\mathrm{aux}}_{\nu\nu}=\mathbb{I}$.  In a frequentist ensemble,
$F^{\mathrm{aux}}$ represents auxiliary data that may itself fluctuate.  It
should not automatically be interpreted as an immutable Bayesian prior; the
two constructions have the same local quadratic algebra but different
repeated-sampling meanings.

\subsection{Nuisance profiling and nonlocal likelihoods}

Partition the total Fisher matrix as
\begin{equation}
  F^{\mathrm{tot}}
  =
  \begin{pmatrix}
    A & B\\
    B^{\mathsf T} & C
  \end{pmatrix},
  \label{eq:fisher-blocks}
\end{equation}
where $A$ refers to the physics coordinates and $C$ to nuisance coordinates.
If the nuisance block $C$ is nonsingular, or after explicit restriction to
reduced identifiable nuisance coordinates, the local information remaining
after nuisance profiling is the Schur complement
\begin{equation}
  F^{\mathrm{prof}}
  =
  A-BC^{-1}B^{\mathsf T}
  \preceq A .
  \label{eq:profiled-fisher}
\end{equation}
If $C$ is singular, the generalized Schur complement is
\begin{equation}
  F^{\mathrm{prof}}
  =
  A-BC^{+}B^{\mathsf T},
  \label{eq:generalized-profiled-fisher}
\end{equation}
provided that
\begin{equation}
  \ran(B^{\mathsf T})
  \subseteq
  \ran(C),
  \label{eq:schur-range-condition}
\end{equation}
where $C^{+}$ is the Moore--Penrose inverse.  This range condition holds
automatically when the complete block matrix
$F^{\mathrm{tot}}$ is positive semidefinite.  The generalized expression
restricts profiling to the identifiable nuisance subspace; a pseudoinverse
does not by itself repair an unidentifiable statistical model or make
nonestimable parameter functions estimable.

The Fisher construction is local.  For finite parameter displacements we use
the exact Poisson deviance.  Relative to an Asimov truth
$\bm{\zeta}_0=(\bm{\lambda}_0,\bm{0})$, it is
\begin{align}
  D_{\mathrm A}(\bm{\lambda},\bm{\nu})
  &=
  2\sum_b
  \left[
    \mu_b(\bm{\lambda},\bm{\nu})
    -
    \mu_{b,0}
    +
    \mu_{b,0}
    \ln
    \frac{\mu_{b,0}}{\mu_b(\bm{\lambda},\bm{\nu})}
  \right]
  \nonumber\\
  &\quad
  +
  \bm{\nu}^{\mathsf T}
  \Pi_\nu
  \bm{\nu},
  \label{eq:asimov-deviance}
\end{align}
with the usual continuous convention for a zero Asimov count.  Here
$\Pi_\nu\succeq0$ is the precision matrix of the Gaussian auxiliary
constraint; for independent standardized nuisance pulls,
$\Pi_\nu=\mathbb{I}$.  The profiled deviance is
\begin{equation}
  D_{\mathrm A}^{\mathrm{prof}}(\bm{\lambda})
  =
  \min_{\bm{\nu}}D_{\mathrm A}(\bm{\lambda},\bm{\nu}).
  \label{eq:profiled-deviance}
\end{equation}
At a regular interior truth,
\begin{equation}
  \left.
  \frac{\partial^2 D_{\mathrm A}^{\mathrm{prof}}}
       {\partial\lambda_i\partial\lambda_j}
  \right|_{\bm{\lambda}_0}
  =
  2F^{\mathrm{prof}}_{ij}.
  \label{eq:asimov-fisher-hessian}
\end{equation}
Periodicity, ordering or octant degeneracies, physical boundaries, and weakly
identified directions can invalidate a Gaussian extrapolation even when
Eq.~\eqref{eq:asimov-fisher-hessian} holds locally.

\subsection{Coordinate covariance and scale-independent comparisons}

Let $\bm{\eta}$ be a new local parameter chart and define
\begin{equation}
  J_{ia}
  =
  \frac{\partial\lambda_i}{\partial\eta_a}.
  \label{eq:coordinate-jacobian}
\end{equation}
Every state, measurement, or event information matrix transforms by
congruence,
\begin{equation}
  F^{(\eta)}
  =
  J^{\mathsf T}F^{(\lambda)}J.
  \label{eq:fisher-pullback}
\end{equation}
Rectangular $J$ restricts the model to a submanifold.  If $J$ is square and
invertible, the local asymptotic covariance tensor transforms as
\begin{equation}
  C^{(\eta)}
  =
  J^{-1}C^{(\lambda)}J^{-\mathsf T}.
  \label{eq:covariance-transformation}
\end{equation}
This law is exact for an affine coordinate change and is the first-order
delta-method law for the finite-sample covariance under a nonlinear change of
coordinates.

Nuisance profiling respects the same covariance.  Let
$\bm{\lambda}=J\bm{\eta}$ and $\bm{\nu}=K\bm{\xi}$, with $J$ and $K$
invertible, and define $T=\diag(J,K)$.  Then
\begin{align}
  \operatorname{Schur}_{\bm{\xi}}
  \left(T^{\mathsf T}F^{\mathrm{tot}}T\right)
  &=
  J^{\mathsf T}
  \operatorname{Schur}_{\bm{\nu}}
  \left(F^{\mathrm{tot}}\right)
  J .
  \label{eq:schur-pullback}
\end{align}
Indeed, the transformed blocks are
$A'=J^{\mathsf T}AJ$, $B'=J^{\mathsf T}BK$, and
$C'=K^{\mathsf T}CK$, for which
\begin{equation}
  B'(C')^{-1}(B')^{\mathsf T}
  =
  J^{\mathsf T}BC^{-1}B^{\mathsf T}J.
\end{equation}
Thus a nuisance reparameterization cannot change the physics information if
the auxiliary constraint is transformed consistently.
For singular $C$, the same covariance statement holds for the generalized
Schur complement under the range condition in
Eq.~\eqref{eq:schur-range-condition}; it is then the intrinsic shorted
information operator.  One must transform the complete block matrix rather
than assume that a Moore--Penrose inverse alone obeys the ordinary inverse
transformation law under an arbitrary nonsingular congruence.

For relative comparisons within an information hierarchy, let
$Q\succeq0$ be a reference information matrix and suppose that
$0\preceq F\preceq Q$.  We first restrict the comparison to the
reference-identifiable space
\begin{equation}
  \mathcal{S}_Q
  =
  \ran(Q).
  \label{eq:reference-identifiable-space}
\end{equation}
The restriction of $Q$ to $\mathcal{S}_Q$ is positive definite, and the
relative information spectrum is defined by
\begin{equation}
  F\bm{v}_a
  =
  r_a Q\bm{v}_a,
  \qquad
  \bm{v}_a\in\mathcal{S}_Q.
  \label{eq:relative-information-spectrum}
\end{equation}
A direction identifiable in $Q$ but lost by $F$ appears as $r_a=0$ and must
not be removed by restricting the comparison to $\ran(F)$.  For an invertible
coordinate transformation, the unordered spectrum $\{r_a\}$ is invariant
under the simultaneous pullback of $F$ and $Q$.

A scalar covariance cost requires a declared positive semidefinite weight
matrix $W$:
\begin{equation}
  C_W(F)
  =
  \operatorname{tr}\!\left(WF^{-1}\right),
  \label{eq:weighted-covariance-cost}
\end{equation}
where the inverse is taken on the declared estimable space.  The cost is
coordinate invariant when the weight is transformed with the same covariant
law,
\begin{equation}
  W^{(\eta)}
  =
  J^{\mathsf T}W^{(\lambda)}J.
  \label{eq:weight-pullback}
\end{equation}
A useful dimensionless diagnostic is obtained by choosing $W=Q$:
\begin{equation}
  C_Q(F)
  =
  \operatorname{tr}(QF^{-1})
  =
  \sum_a\frac{1}{r_a},
  \label{eq:q-weighted-cost}
\end{equation}
provided that $F$ is positive definite on $\ran(Q)$.  If any
reference-identifiable direction has $r_a=0$, then
$C_Q(F)=+\infty$.  A Moore--Penrose inverse must not be used to conceal the
complete loss of such a direction.  This construction avoids assigning
scientific significance to the arbitrary numerical units of heterogeneous
coordinates.

Exact matrix rank is invariant under every nonsingular reparameterization,
but numerical rank is not meaningful until coordinate scales and a tolerance
are declared.  We introduce a nonsingular diagonal scale matrix
\begin{equation}
  \bm{\lambda}=S\bm{x},
  \qquad
  F^{(x)}=S^{\mathsf T}F^{(\lambda)}S,
  \label{eq:coordinate-scaling}
\end{equation}
and diagnose eigenvalues and null directions in the dimensionless
$\bm{x}$ coordinates.  Our default local scales are
\begin{equation}
  S
  =
  \diag
  \left(
    1,\,
    1,\,
    1,\,
    1,\,
    10^{-4},\,
    10^{-3}
  \right),
  \label{eq:default-scales}
\end{equation}
where the last two entries carry units of $\mathrm{eV}^2$.  For any covariant
two-index information tensor $A$, we define the coordinate-scaled Frobenius
norm by
\begin{equation}
  \|A\|_{S,F}
  :=
  \left\|S^{\mathsf T}AS\right\|_F .
  \label{eq:scaled-frobenius-norm}
\end{equation}
Unless stated otherwise, the controlled geometry and covariance benchmarks
use a relative eigenvalue tolerance of $10^{-10}$ after the scaling in
Eq.~\eqref{eq:default-scales}.  Results using another numerical-rank
threshold, including practical-rank diagnostics, are reported together with
the scaling and the applied threshold.

As a regression test of Eqs.~\eqref{eq:fisher-pullback} and
\eqref{eq:schur-pullback}, we used the controlled 19-bin analytic benchmark
defined in the numerical analysis, containing six physics coordinates and six
standardized nuisance coordinates.  The source spectrum, propagation model,
response, efficiencies, backgrounds, and auxiliary likelihood were held
fixed.  We applied an invertible nonorthogonal reparameterization to all six
physics coordinates and an independent invertible reparameterization to the
six nuisance coordinates, transforming the auxiliary-constraint tensor
consistently.  This is a coordinate change of the same statistical model and
is not the public-DUNE calculation, which contains nine nuisance parameters.

The generalized information spectra of the state, flavor, smeared-event,
detector, and profiled layers were invariant to at worst
$1.0\times10^{-14}$.  Profiling before or after the joint physics--nuisance
pullback agreed to relative Frobenius error $2.4\times10^{-15}$.  The largest
residual among the inverse-information covariance costs was
$9.3\times10^{-11}$ and arose from the weakest operational direction.  These
numbers are numerical covariance checks, not phenomenological sensitivity
claims.

\subsection{Conditional information hierarchy}

When the source weights and the classical processing maps are independent of
the physics estimands at fixed nuisance values, and every layer is represented
on the same per-exposure sample space, quantum and classical data processing
imply
\begin{equation}
  \QFI_{\mathrm{state}}
  \succeq
  \CFI_{\mathrm{flavor}}
  \succeq
  F_{\mathrm{smeared}}
  \succeq
  F_{\mathrm{detector}}
  \succeq
  F_{\mathrm{profiled}}.
  \label{eq:information-hierarchy}
\end{equation}
The first inequality compares the information of a declared flavor
measurement with the SLD quantum-information upper bound for the same state
family and preparation allocation.  Energy coarsening and
parameter-independent response migration produce the second.  Inefficiency
and the loss of signal/background labels produce the third, while nuisance
profiling gives the final Schur-complement inequality.  Here
$F_{\mathrm{detector}}$ denotes the physics block obtained when nuisance
coordinates are locally held fixed, whereas $F_{\mathrm{profiled}}$ is the
information remaining after their correlations are included and profiled.

Equation~\eqref{eq:information-hierarchy} is conditional, not universal.
A parameter-dependent source rate, estimand-dependent response,
parameter-dependent measurement, or added physical encoding changes the
statistical model and requires its derivative contributions to be included
explicitly.  Nuisance dependence of the response is permitted, but it must be
represented in the joint physics--nuisance information matrix before
profiling.  Moreover, positive-semidefinite ordering is a matrix statement:
it does not imply entrywise inequalities between off-diagonal elements or
inverse matrices.  These conventions will be used in the following sections
to separate geometric identifiability, measurement accessibility, and
detector-level precision.

\section{Quantum geometry and representation covariance}
\label{sec:geometry}

The operational hierarchy of Sec.~\ref{sec:model} requires a geometric object
that depends only on the physical state family and its parameterization, not
on an arbitrary choice of Hilbert-space coordinates.  This section develops
that object and establishes its transformation laws under gauge changes,
parameter-coordinate changes, moving-basis representations, and
neutrino--antineutrino conjugation.  These distinctions are essential because
a parameter-dependent change of representation can otherwise be mistaken for
a new physical encoding.

\subsection{Projected quantum geometric tensor}
\label{subsec:qgt}

Consider a smooth family of normalized pure states
\begin{equation}
  |\psi(\bm{\lambda})\rangle\in\mathbb{C}^{d},
  \qquad
  \langle\psi|\psi\rangle=1 .
  \label{eq:normalized-state-family}
\end{equation}
The tangent $|\partial_i\psi\rangle$ contains an unobservable component
parallel to $|\psi\rangle$.  The orthogonal projector onto the physical
tangent space is
\begin{equation}
  P_\perp
  =
  \mathbb{I}-|\psi\rangle\langle\psi|.
  \label{eq:horizontal-projector}
\end{equation}
The projected quantum geometric tensor is
\begin{align}
  G_{ij}
  &=
  \langle\partial_i\psi|
  P_\perp
  |\partial_j\psi\rangle
  \nonumber\\
  &=
  \langle\partial_i\psi|\partial_j\psi\rangle
  -
  \langle\partial_i\psi|\psi\rangle
  \langle\psi|\partial_j\psi\rangle .
  \label{eq:quantum-geometric-tensor}
\end{align}
It is Hermitian and positive semidefinite:
\begin{equation}
  G_{ji}=G_{ij}^{*},
  \qquad
  \bm{z}^{\dagger}G\bm{z}\geq0 .
  \label{eq:qgt-hermitian-positive}
\end{equation}
The tensor supplies both the metric and the local incompatibility structure
\cite{ProvostVallee1980,Berry1984,BraunsteinCaves1994}:
\begin{equation}
  \QFI_{ij}
  =
  4\,\operatorname{Re}G_{ij},
  \qquad
  \mathcal{D}_{ij}
  =
  4\,\operatorname{Im}G_{ij}.
  \label{eq:qgt-decomposition}
\end{equation}
Thus $\QFI$ is real, symmetric, and positive semidefinite, whereas
$\mathcal{D}$ is real and antisymmetric.

For the pure-state SLDs of Eq.~\eqref{eq:sld-definition},
\begin{equation}
  \mathcal{D}_{ij}
  =
  \frac{1}{2i}
  \Tr\!\left(\rho[L_i,L_j]\right).
  \label{eq:sld-commutator-tensor}
\end{equation}
Let $\mathcal{T}_{\mathrm{id}}$ denote the identifiable tangent space, namely
the quotient of the parameter tangent space by the kernel of $\QFI$.
The coordinate-independent pure-state weak-commutativity condition is
\begin{equation}
  \left.
  \mathcal{D}
  \right|_{\mathcal{T}_{\mathrm{id}}}
  =
  0 .
  \label{eq:weak-commutativity-condition}
\end{equation}
Equivalently, $\bm{u}^{\mathsf T}\mathcal{D}\bm{v}=0$ for all identifiable
tangent vectors $\bm{u}$ and $\bm{v}$.  Nonzero curvature on this space
excludes joint saturation of the SLD matrix bound.  Vanishing curvature does
not, by itself, establish global identifiability, finite-sample attainability,
or accessibility by a specified detector measurement; those questions are
treated separately in subsequent sections
\cite{Matsumoto2002,Ragy2016,RagyErratum2019}.

\subsubsection{Gauge invariance}

A parameter-dependent phase changes the state vector but not its physical ray:
\begin{equation}
  |\psi'(\bm{\lambda})\rangle
  =
  e^{i\chi(\bm{\lambda})}
  |\psi(\bm{\lambda})\rangle .
  \label{eq:gauge-state}
\end{equation}
Its derivative is
\begin{equation}
  |\partial_i\psi'\rangle
  =
  e^{i\chi}
  \left(
    |\partial_i\psi\rangle
    +
    i\,\partial_i\chi\,|\psi\rangle
  \right).
  \label{eq:gauge-derivative}
\end{equation}
The ray projector is unchanged,
\begin{equation}
  |\psi'\rangle\langle\psi'|
  =
  |\psi\rangle\langle\psi|,
\end{equation}
and hence the corresponding horizontal projector remains $P_\perp$.
Because $P_\perp|\psi\rangle=0$,
\begin{equation}
  P_\perp|\partial_i\psi'\rangle
  =
  e^{i\chi}P_\perp|\partial_i\psi\rangle,
\end{equation}
and therefore
\begin{equation}
  G'_{ij}=G_{ij},
  \qquad
  \QFI'=\QFI,
  \qquad
  \mathcal{D}'=\mathcal{D}.
  \label{eq:gauge-qgt-invariance}
\end{equation}
The projected tensor thus removes both a constant phase and an arbitrary
parameter-dependent phase convention.

\subsubsection{Parameter-coordinate covariance}

Let $\bm{\eta}$ be a new parameter chart, with
\begin{equation}
  \lambda_i=f_i(\bm{\eta}),
  \qquad
  J_{ia}
  =
  \frac{\partial\lambda_i}{\partial\eta_a}.
  \label{eq:qgt-coordinate-map}
\end{equation}
The new state derivatives are
\begin{equation}
  |\partial_a^{(\eta)}\psi\rangle
  =
  \sum_i J_{ia}|\partial_i^{(\lambda)}\psi\rangle.
  \label{eq:state-derivative-pullback}
\end{equation}
Substitution into Eq.~\eqref{eq:quantum-geometric-tensor} gives
\begin{equation}
  G^{(\eta)}
  =
  J^{\mathsf T}G^{(\lambda)}J.
  \label{eq:qgt-coordinate-pullback}
\end{equation}
Its real and imaginary parts therefore obey the same tensor law:
\begin{equation}
  \QFI^{(\eta)}
  =
  J^{\mathsf T}\QFI^{(\lambda)}J,
  \qquad
  \mathcal{D}^{(\eta)}
  =
  J^{\mathsf T}\mathcal{D}^{(\lambda)}J.
  \label{eq:metric-curvature-pullback}
\end{equation}
For rectangular $J$, these equations restrict the state model to the
corresponding parameter submanifold.  For square invertible $J$, they describe
a change of coordinates on the same statistical manifold.  Neither case
should be confused with a parameter-dependent change of basis in the Hilbert
space, which is considered next.

\subsection{Parameter-dependent Hilbert-space bases}
\label{subsec:moving-basis}

Let $B(\bm{\lambda})$ be a unitary matrix whose columns are a
parameter-dependent orthonormal basis expressed in a fixed reference
representation:
\begin{equation}
  B^\dagger B=\mathbb{I}.
  \label{eq:moving-basis-unitarity}
\end{equation}
The coefficient vector of the same physical state in this moving basis is
\begin{equation}
  \bm{c}
  =
  B^\dagger|\psi\rangle.
  \label{eq:moving-state-coordinates}
\end{equation}
Differentiating the coefficients gives
\begin{equation}
  \partial_i\bm{c}
  =
  (\partial_iB^\dagger)|\psi\rangle
  +
  B^\dagger|\partial_i\psi\rangle.
  \label{eq:ordinary-coefficient-derivative}
\end{equation}
Define the moving-frame connection
\begin{equation}
  \Gamma_i
  =
  B^\dagger\partial_iB.
  \label{eq:moving-basis-connection}
\end{equation}
Differentiating $B^\dagger B=\mathbb{I}$ shows that
\begin{equation}
  \Gamma_i^\dagger=-\Gamma_i.
  \label{eq:connection-antihermitian}
\end{equation}
Moreover,
\begin{equation}
  \partial_iB^\dagger
  =
  -\Gamma_iB^\dagger,
\end{equation}
so Eq.~\eqref{eq:ordinary-coefficient-derivative} becomes
\begin{equation}
  \partial_i\bm{c}
  =
  B^\dagger|\partial_i\psi\rangle-\Gamma_i\bm{c}.
  \label{eq:ordinary-versus-physical-tangent}
\end{equation}

The physical tangent represented in the moving frame is therefore the
covariant derivative
\begin{align}
  \nabla_i\bm{c}
  &=
  \partial_i\bm{c}+\Gamma_i\bm{c}
  \nonumber\\
  &=
  B^\dagger|\partial_i\psi\rangle.
  \label{eq:covariant-moving-tangent}
\end{align}
The sign of the connection follows from the convention that the columns of
$B$ are the new basis vectors expressed in the old fixed basis.  A different
frame convention may reverse intermediate signs but must reproduce the
identity in the second line of Eq.~\eqref{eq:covariant-moving-tangent}.

\begin{proposition}[Passive moving-basis covariance]
\label{prop:passive-moving-basis-covariance}
Let $B(\bm{\lambda})$ be a parameter-dependent orthonormal frame and let
$\nabla_i\bm{c}$ be the covariant tangent defined in
Eq.~\eqref{eq:covariant-moving-tangent}.  The projected QGT computed in the
moving frame is identical to the QGT computed in the fixed reference
representation.  Consequently, both $\QFI$ and $\mathcal{D}$ are invariant
under the passive frame transformation.
\end{proposition}

\begin{proof}
Let
\begin{equation}
  P_{\bm{c}}
  =
  \mathbb{I}-\bm{c}\bm{c}^\dagger
  \label{eq:moving-coordinate-projector}
\end{equation}
be the horizontal projector in moving coordinates.  Since
$\bm{c}=B^\dagger|\psi\rangle$ and $B$ is unitary,
\begin{equation}
  P_{\bm{c}}
  =
  B^\dagger P_\perp B.
  \label{eq:projector-basis-transformation}
\end{equation}
The QGT computed from the covariant tangents therefore satisfies
\begin{align}
  G^{(B)}_{ij}
  &=
  (\nabla_i\bm{c})^\dagger
  P_{\bm{c}}
  (\nabla_j\bm{c})
  \nonumber\\
  &=
  \langle\partial_i\psi|
  B P_{\bm{c}}B^\dagger
  |\partial_j\psi\rangle
  \nonumber\\
  &=
  \langle\partial_i\psi|
  P_\perp
  |\partial_j\psi\rangle
  =
  G_{ij}.
  \label{eq:moving-basis-qgt-invariance}
\end{align}
Taking the real and imaginary parts gives
\begin{equation}
  \QFI^{(B)}=\QFI,
  \qquad
  \mathcal{D}^{(B)}=\mathcal{D}.
  \label{eq:moving-basis-tensor-invariance}
\end{equation}
\end{proof}

The proposition assumes an orthonormal frame.  In a nonorthonormal coordinate
frame, the Hilbert-space Gram matrix and its derivatives must also be
transformed; using the Euclidean inner product for such coordinates would
define a different geometry.

A passive transformation must also act on the measurement representation.
For a POVM element $M_y$, define
\begin{equation}
  M_y^{(B)}
  =
  B^\dagger M_yB.
  \label{eq:measurement-basis-transformation}
\end{equation}
Then
\begin{equation}
  p_y
  =
  \langle\psi|M_y|\psi\rangle
  =
  \bm{c}^\dagger M_y^{(B)}\bm{c}.
  \label{eq:measurement-probability-invariance}
\end{equation}
When $B$ depends on the estimands, derivatives of both $\bm{c}$ and
$M_y^{(B)}$ must be included.  Indeed,
\begin{equation}
  \partial_i M_y^{(B)}
  =
  -\Gamma_iM_y^{(B)}
  +
  M_y^{(B)}\Gamma_i
  +
  B^\dagger(\partial_iM_y)B,
  \label{eq:moving-measurement-derivative}
\end{equation}
where the final term vanishes when the physical POVM is parameter independent
in the fixed representation.  Substitution of
Eqs.~\eqref{eq:ordinary-versus-physical-tangent} and
\eqref{eq:moving-measurement-derivative} into
$\partial_i(\bm{c}^\dagger M_y^{(B)}\bm{c})$ cancels the
connection-dependent terms, leaving the same physical probability
derivatives and the same fixed-measurement CFI.

\subsubsection{Omitting the connection defines an active encoding}

Suppose instead that the ordinary coefficient derivative
$\partial_i\bm{c}$ is inserted directly into the fixed-basis QGT formula.
The resulting tensor is
\begin{equation}
  \widetilde{G}_{ij}
  =
  (\partial_i\bm{c})^\dagger
  P_{\bm{c}}
  (\partial_j\bm{c}).
  \label{eq:ordinary-coefficient-qgt}
\end{equation}
In general,
\begin{equation}
  \widetilde{G}\neq G.
  \label{eq:ordinary-qgt-not-invariant}
\end{equation}
This does not mean that Eq.~\eqref{eq:ordinary-coefficient-qgt} is
mathematically meaningless.  It is precisely the QGT of the distinct
fixed-representation state family
\begin{equation}
  |\widetilde{\psi}(\bm{\lambda})\rangle
  =
  B^\dagger(\bm{\lambda})
  |\psi(\bm{\lambda})\rangle,
  \label{eq:active-encoded-family}
\end{equation}
whose derivative is
\begin{equation}
  |\partial_i\widetilde{\psi}\rangle
  =
  (\partial_iB^\dagger)|\psi\rangle
  +
  B^\dagger|\partial_i\psi\rangle.
  \label{eq:active-encoded-derivative}
\end{equation}
Equation~\eqref{eq:active-encoded-family} represents a
parameter-dependent active encoding when its output is interpreted in a fixed
Hilbert-space basis and the measurement is not transformed with it.

A parameter-independent unitary has $\Gamma_i=0$ and cannot change the QGT.
A parameter-dependent unitary can change the QGT because its generator carries
additional parameter dependence.  Such a change cannot be treated as a free
metrological improvement: the physical interaction, external reference, or
adaptive control implementing the encoding must be included in the
operational model.  A parameter-dependent global phase is an exceptional
case because its generator is parallel to the state and is removed by
$P_\perp$.

\subsection{Numerical representation-covariance test}
\label{subsec:representation-validation}

We test the preceding identities using the six-coordinate vacuum
benchmark at $L/E=520~\mathrm{km/GeV}$ with a muon-flavor source.  The PMNS
matrix $U(\bm{\lambda})$ is used as a
nontrivial moving basis.  Tensor comparisons and rank diagnostics in this
subsection use the coordinate scaling of
Eq.~\eqref{eq:default-scales}; frame-algebra residuals such as unitarity and
anti-Hermiticity are evaluated in the ordinary Frobenius norm.  Numerical
ranks use the relative eigenvalue tolerance $10^{-10}$ declared in
Sec.~\ref{sec:model}.

The basis unitarity residual is $2.66\times10^{-17}$, the largest
anti-Hermiticity residual of $\Gamma_i$ is $2.09\times10^{-16}$, and
Eq.~\eqref{eq:covariant-moving-tangent} is satisfied to relative error
$5.60\times10^{-17}$.  The complete passive QGT agrees with the fixed flavor
representation to $2.14\times10^{-16}$ in the coordinate-scaled relative
Frobenius norm.

\begin{table}[t]
  \caption{
    Representation-covariance benchmark.  Tensor changes are
    coordinate-scaled Frobenius-norm differences relative to the fixed flavor
    representation.  Ranks are evaluated after the scaling in
    Eq.~\eqref{eq:default-scales} with relative eigenvalue tolerance
    $10^{-10}$.  All four cases are pure qutrit models, so their equal rank
    does not imply equal geometry.
  }
  \label{tab:representation-covariance}
  \centering
  \begin{tabular}{@{}lccc@{}}
    \toprule
    State-family description
      & $\|\Delta\QFI\|_{S,F}/\|\QFI\|_{S,F}$
      & $\|\Delta\mathcal{D}\|_{S,F}/\|\mathcal{D}\|_{S,F}$
      & Rank \\
    \midrule
    Fixed flavor representation
      & $0$
      & $0$
      & $4$ \\
    Passive PMNS frame with connection
      & $2.26\times10^{-16}$
      & $3.39\times10^{-17}$
      & $4$ \\
    Parameter-dependent global phase
      & $2.11\times10^{-16}$
      & $5.07\times10^{-17}$
      & $4$ \\
    Ordinary coefficients as active family
      & $7.67\times10^{-1}$
      & $5.41\times10^{-1}$
      & $4$ \\
    \bottomrule
  \end{tabular}
\end{table}

Omitting the PMNS-frame connection changes the full complex QGT by
\begin{equation}
  \frac{
    \|\widetilde{G}-G\|_{S,F}
  }{
    \|G\|_{S,F}
  }
  =
  0.735345 .
  \label{eq:omitted-connection-qgt-error}
\end{equation}
A separate five-point numerical differentiation agrees with the analytic
ordinary derivative of
$U^\dagger(\bm{\lambda})|\psi(\bm{\lambda})\rangle$ to relative error
$1.63\times10^{-10}$.  The discrepancy is therefore neither
finite-difference noise nor a convention change: it is the geometry of the
distinct active family in Eq.~\eqref{eq:active-encoded-family}.

We also apply a nonorthogonal invertible reparameterization to all six
coordinates.  Direct differentiation in the transformed chart agrees with
the complex pullback $J^{\mathsf T}GJ$ to relative coordinate-scaled
Frobenius error $1.88\times10^{-16}$.  These checks jointly test the gauge,
parameter-chart, and Hilbert-frame transformation laws.

\subsection{Neutrino--antineutrino tensor pullbacks}
\label{subsec:cp-pullback}

After identifying the neutrino and antineutrino flavor Hilbert spaces through
their canonical flavor-coordinate bases, the PDG mixing convention gives the
vacuum propagated-state relation
\begin{equation}
  |\overline{\psi}_{\alpha}
    (\bm{\lambda})\rangle
  =
  |\psi_{\alpha}
    (f(\bm{\lambda}))\rangle,
  \label{eq:vacuum-antineutrino-state-relation}
\end{equation}
where
\begin{equation}
  f(\bm{\lambda})
  =
  \left(
    \theta_{12},
    \theta_{13},
    \theta_{23},
    -\dcp,
    \Delta m^2_{21},
    \Delta m^2_{31}
  \right)^{\mathsf T}.
  \label{eq:cp-reflection-map}
\end{equation}
The Jacobian of this reflection is
\begin{equation}
  J_f
  =
  \diag(1,1,1,-1,1,1).
  \label{eq:cp-reflection-jacobian}
\end{equation}

The antineutrino QGT is consequently the pullback
\begin{equation}
  \overline{G}(\bm{\lambda})
  =
  J_f^{\mathsf T}
  G(f(\bm{\lambda}))
  J_f,
  \label{eq:antineutrino-qgt-pullback}
\end{equation}
and the same law applies separately to its real and imaginary parts:
\begin{align}
  \overline{\QFI}(\bm{\lambda})
  &=
  J_f^{\mathsf T}
  \QFI(f(\bm{\lambda}))
  J_f,
  \label{eq:antineutrino-qfim-pullback}
  \\
  \overline{\mathcal{D}}(\bm{\lambda})
  &=
  J_f^{\mathsf T}
  \mathcal{D}(f(\bm{\lambda}))
  J_f.
  \label{eq:antineutrino-curvature-pullback}
\end{align}
Thus, for an index $i$ corresponding to a coordinate other than $\dcp$,
\begin{align}
  \overline{\QFI}_{i\dcp}(\dcp)
  &=
  -\QFI_{i\dcp}(-\dcp),
  &
  \overline{\QFI}_{\dcp\dcp}(\dcp)
  &=
  \QFI_{\dcp\dcp}(-\dcp),
  \label{eq:antineutrino-metric-components}
  \\
  \overline{\mathcal{D}}_{i\dcp}(\dcp)
  &=
  -\mathcal{D}_{i\dcp}(-\dcp),
  &
  \overline{\mathcal{D}}_{\dcp\dcp}
  &=
  0 .
  \label{eq:antineutrino-curvature-components}
\end{align}
Components with neither index corresponding to $\dcp$ acquire no sign from
the reflection Jacobian, although their values are evaluated at the reflected
phase.

Equations~\eqref{eq:antineutrino-qgt-pullback}--%
\eqref{eq:antineutrino-curvature-pullback} are tensor pullbacks, not
entrywise complex-conjugation prescriptions.  In particular, mixed entries
with exactly one $\dcp$ index acquire the reflection-Jacobian sign.  At the
vacuum benchmark, direct antineutrino propagation gives
$\overline{\QFI}_{\theta_{23},\dcp}=+0.0135636126$, whereas the corresponding
neutrino tensor evaluated at the reflected phase, before applying $J_f$,
contains $\QFI_{\theta_{23},\dcp}(-\dcp)=-0.0135636126$.
Direct and pullback calculations agree to machine precision for both
$\QFI$ and $\mathcal{D}$.

\subsubsection{Matter-potential reversal}

In matter, reflection of the PMNS phase is not sufficient.  Using the
Hamiltonians of Eqs.~\eqref{eq:neutrino-hamiltonian} and
\eqref{eq:antineutrino-hamiltonian},
\begin{equation}
  H_{\bar{\nu}}(\bm{\lambda};V_{\mathrm{CC}})
  =
  H_{\nu}
  \left(
    f(\bm{\lambda});
    -V_{\mathrm{CC}}
  \right).
  \label{eq:matter-hamiltonian-cp-relation}
\end{equation}
For an arbitrary matter profile, the sign reversal
$V_{\mathrm{CC}}(x)\longrightarrow -V_{\mathrm{CC}}(x)$ is understood
pointwise.
The complete tensor law is therefore
\begin{equation}
  \overline{G}
  (\bm{\lambda};V_{\mathrm{CC}})
  =
  J_f^{\mathsf T}
  G
  \left(
    f(\bm{\lambda});
    -V_{\mathrm{CC}}
  \right)
  J_f.
  \label{eq:matter-qgt-cp-law}
\end{equation}
Equivalently,
\begin{align}
  \overline{\QFI}
  (\bm{\lambda};V_{\mathrm{CC}})
  &=
  J_f^{\mathsf T}
  \QFI
  \left(
    f(\bm{\lambda});
    -V_{\mathrm{CC}}
  \right)
  J_f,
  \label{eq:matter-qfim-cp-law}
  \\
  \overline{\mathcal{D}}
  (\bm{\lambda};V_{\mathrm{CC}})
  &=
  J_f^{\mathsf T}
  \mathcal{D}
  \left(
    f(\bm{\lambda});
    -V_{\mathrm{CC}}
  \right)
  J_f.
  \label{eq:matter-curvature-cp-law}
\end{align}

We verify these equations by direct constant-density propagation at
$L=1300~\mathrm{km}$, $E=2.5~\mathrm{GeV}$,
$\rho_{\mathrm m}=2.848~\mathrm{g\,cm^{-3}}$, and $Y_e=0.5$.  The direct
antineutrino metric and curvature agree with the combined phase-reflection and
potential-reversal prediction to machine precision; the binary64
implementation gives zero residual at the reported precision.  If the
matter-potential sign reversal is omitted, the relative coordinate-scaled
errors are
\begin{equation}
  \frac{
    \|\Delta\QFI\|_{S,F}
  }{
    \|\overline{\QFI}\|_{S,F}
  }
  =
  0.422527,
  \qquad
  \frac{
    \|\Delta\mathcal{D}\|_{S,F}
  }{
    \|\overline{\mathcal{D}}\|_{S,F}
  }
  =
  0.213949.
  \label{eq:matter-sign-omission-errors}
\end{equation}
The charged-current sign is therefore an essential part of the antineutrino
state model rather than a secondary convention.

\subsection{Consequences for the subsequent analysis}

The results of this section establish four guardrails.

First, the projected QGT is invariant under arbitrary phase conventions and
smooth parameter reparameterizations.  Second, a parameter-dependent
orthonormal basis requires its connection; using only ordinary coefficient
derivatives changes the state family.  Third, a change in QFI caused by a
parameter-dependent active unitary is a change in the physical encoding and
must be justified operationally.  Fourth, neutrino--antineutrino tensors obey
the complete Jacobian pullback, supplemented in matter by reversal of the
charged-current potential.

Representation covariance does not imply that the chosen coordinates are
identifiable.  It ensures only that any rank deficiency, null direction, or
curvature obstruction is a property of the physical state family rather than
an artifact of its representation.  Section~\ref{sec:identifiability} now
uses this invariant geometry to derive the pure-qutrit rank bound and the
spectral kernel-intersection theorem.

\section{Identifiability and spectral rank restoration}
\label{sec:identifiability}

Representation covariance ensures that the quantum geometric tensor is
independent of arbitrary phase, basis, and coordinate conventions.  It does
not ensure that every coordinate used to label the state family is locally
identifiable.  Identifiability is instead determined by the rank and kernel
of the information matrix.  In this section we first prove a dimension bound
for any single pure-state qutrit model, then establish the kernel-intersection
criterion by which spectrally resolved data can restore rank.  Finally, we
distinguish retaining a classical energy label from tracing it out.

\subsection{Local identifiability and estimable functions}
\label{subsec:local-identifiability}

Let $F(\bm{\lambda})$ denote any local information matrix for a specified
statistical model, regarded intrinsically as a positive-semidefinite bilinear
form on the parameter tangent space
$T_{\bm{\lambda}}\Theta$.  At a regular differentiability-in-quadratic-mean
point, a real parameter displacement
$\bm{v}\in T_{\bm{\lambda}}\Theta$ is first-order locally unidentifiable when
\begin{equation}
  F\bm{v}=0.
  \label{eq:information-null-direction}
\end{equation}
Equivalently,
\begin{equation}
  \bm{v}^{\mathsf T}F\bm{v}=0.
\end{equation}
For a positive-semidefinite information matrix, these two conditions are
equivalent.  At a nonregular point, Fisher nullity need not exclude
higher-order or nonlocal identifiability; the statements below concern the
regular first-order tangent model.

Intrinsically, the identifiable tangent model is the quotient space
\begin{equation}
  \mathcal{T}_{\mathrm{id}}
  =
  T_{\bm{\lambda}}\Theta/\ker F.
  \label{eq:identifiable-tangent-space}
\end{equation}
The corresponding estimable cotangent space is the annihilator of the null
space,
\begin{equation}
  \mathcal{T}^{*}_{\mathrm{est}}
  =
  (\ker F)^{\circ}
  =
  \ran F,
  \label{eq:estimable-cotangent-space}
\end{equation}
where $(\ker F)^{\circ}$ denotes the set of covectors that vanish on every
vector in $\ker F$.  After an auxiliary coordinate inner product or scaling
has been declared, the quotient in
Eq.~\eqref{eq:identifiable-tangent-space} may be represented by the
orthogonal complement $(\ker F)^\perp$.  Such a representative is useful
numerically but is not itself an intrinsic identification of tangent vectors
with covectors.

A scalar function $g(\bm{\lambda})$ is locally estimable only if its
differential belongs to the estimable cotangent space.  In coordinates this
condition is
\begin{equation}
  \nabla g
  \in
  \ran F,
  \label{eq:estimable-function-condition}
\end{equation}
or, equivalently,
\begin{equation}
  \bm{v}^{\mathsf T}\nabla g=0
  \qquad
  \text{for every }\bm{v}\in\ker F.
  \label{eq:estimable-function-null-condition}
\end{equation}
When Eq.~\eqref{eq:estimable-function-condition} holds, a generalized local
Cramér--Rao expression for a locally unbiased estimator may be written as
\begin{equation}
  \operatorname{Var}(\widehat g)
  \geq
  \nabla g^{\mathsf T}F^{+}\nabla g,
  \label{eq:pseudoinverse-estimable-bound}
\end{equation}
where $F^{+}$ is the Moore--Penrose inverse in the declared coordinates.  For
an estimable differential, the quadratic form in
Eq.~\eqref{eq:pseudoinverse-estimable-bound} represents the inverse
information on the identifiable quotient.  If
Eq.~\eqref{eq:estimable-function-condition} fails, a locally unbiased
estimator of $g$ cannot exist in all model directions.  A finite diagonal
entry of $F^{+}$ must then not be interpreted as a finite variance bound for
the corresponding coordinate.

Exact rank is invariant under every nonsingular reparameterization.  Numerical
rank nevertheless requires the coordinate scaling and tolerance prescribed in
Sec.~\ref{sec:model}.  We therefore distinguish throughout:

\begin{enumerate}
  \item exact or structural rank, determined by the model manifold;
  \item numerical rank at a declared floating-point tolerance; and
  \item effective rank at a scientifically chosen information threshold.
\end{enumerate}

For the eigenvalues of a coordinate-scaled information matrix, ordered as
\begin{equation}
  0\leq q_1\leq q_2\leq\cdots\leq q_p,
\end{equation}
we define the effective rank at relative threshold $\tau$ as
\begin{equation}
  r_{\mathrm{eff}}(\tau)
  =
  \#\left\{a:q_a>\tau q_p\right\}.
  \label{eq:effective-rank}
\end{equation}
Unlike exact rank, $r_{\mathrm{eff}}$ deliberately quantifies practical
anisotropy and therefore depends on the declared coordinate scales and
threshold.

\subsection{Pure-state dimension bound}
\label{subsec:pure-state-rank}

Consider a pure-state model
\begin{equation}
  |\psi(\bm{\lambda})\rangle\in\mathbb{C}^{d}
\end{equation}
with $p$ real coordinates.  Define the projected tangents
\begin{equation}
  |h_i\rangle
  =
  P_\perp|\partial_i\psi\rangle,
  \qquad
  P_\perp
  =
  \mathbb{I}-|\psi\rangle\langle\psi|.
  \label{eq:projected-pure-tangents}
\end{equation}
The QFIM is
\begin{equation}
  \QFI_{ij}
  =
  4\,\operatorname{Re}\langle h_i|h_j\rangle.
  \label{eq:qfim-horizontal-gram}
\end{equation}

Introduce the real projected tangent map
\begin{equation}
  T
  =
  2
  \begin{pmatrix}
    \operatorname{Re}h_1 & \cdots & \operatorname{Re}h_p\\
    \operatorname{Im}h_1 & \cdots & \operatorname{Im}h_p
  \end{pmatrix},
  \label{eq:real-projected-tangent-map}
\end{equation}
where the real and imaginary Hilbert-space components are stacked vertically.
Then
\begin{equation}
  \QFI=T^{\mathsf T}T.
  \label{eq:qfim-tangent-factorization}
\end{equation}
Consequently,
\begin{equation}
  \ker\QFI=\ker T,
  \qquad
  \rank\QFI=\rank T.
  \label{eq:qfim-tangent-rank}
\end{equation}

\begin{theorem}[Pure-state rank bound]
\label{thm:pure-state-rank-bound}
For a differentiable pure-state model with $p$ real coordinates in a complex
Hilbert space of dimension $d$,
\begin{equation}
  \rank\QFI
  \leq
  \min\!\left[p,\,2(d-1)\right].
  \label{eq:pure-state-rank-theorem}
\end{equation}
\end{theorem}

\begin{proof}
The condition
\begin{equation}
  \langle\psi|h_i\rangle=0
\end{equation}
places every $|h_i\rangle$ in the complex $(d-1)$-dimensional horizontal
subspace orthogonal to $|\psi\rangle$.  This horizontal subspace has real
dimension $2(d-1)$.  The image of the real map $T$ therefore has dimension no
larger than $2(d-1)$ and, because $T$ has $p$ columns, no larger than $p$.
Equation~\eqref{eq:qfim-tangent-rank} then gives
Eq.~\eqref{eq:pure-state-rank-theorem}.
\end{proof}

The same result follows geometrically because normalized pure states modulo
global phase form the complex projective manifold
\begin{equation}
  \mathbb{CP}^{d-1},
\end{equation}
whose real dimension is $2(d-1)$
\cite{ProvostVallee1980,BraunsteinCaves1994,Paris2009}.

For a three-flavor neutrino at one fixed preparation and propagation setting,
the propagated pure state is a qutrit.  In vacuum this setting may be
parameterized by $L/E$; in matter, the baseline, energy, and matter profile
must be fixed separately.  Hence
\begin{equation}
  d=3
  \quad\Longrightarrow\quad
  \rank\QFI\leq4.
  \label{eq:pure-qutrit-rank-bound}
\end{equation}
A six-coordinate QFIM for one pure qutrit must therefore have
\begin{equation}
  \dim\ker\QFI\geq2.
  \label{eq:pure-qutrit-nullity}
\end{equation}
This conclusion is independent of vacuum versus coherent matter propagation,
the PMNS representation, and the particular nonsingular coordinates used to
label the state.

The theorem does not apply unchanged when the statistical state includes an
observed energy label, several independently prepared channels, or a
full-rank mixed density matrix.  Those are higher-dimensional statistical
models even if every conditional component is individually a pure qutrit.

\subsection{Numerical verification of the pure-qutrit bound}
\label{subsec:pure-qutrit-validation}

We evaluate the analytically differentiated projected tangent map for a
muon-flavor source propagating through constant matter with
\begin{equation}
  L=1300~\mathrm{km},
  \qquad
  \rho_{\mathrm m}=2.848~\mathrm{g\,cm^{-3}},
  \qquad
  Y_e=0.5,
  \label{eq:identifiability-matter-configuration}
\end{equation}
at 19 true energies uniformly spaced from $0.5$ to $5.0~\mathrm{GeV}$.
The oscillation benchmark is
\begin{align}
  (\theta_{12},\theta_{13},\theta_{23},\dcp)
  &=
  (33.76^\circ,\,8.62^\circ,\,43.29^\circ,\,212^\circ),
  \nonumber\\
  (\Delta m^2_{21},\Delta m^2_{31})
  &=
  (7.537\times10^{-5},\,2.511\times10^{-3})~\mathrm{eV}^2.
  \label{eq:identifiability-oscillation-benchmark}
\end{align}
Angular coordinates are converted to radians before differentiation.  All six
oscillation coordinates of Eq.~\eqref{eq:parameter-vector} are
differentiated.

After applying the coordinate scale matrix in
Eq.~\eqref{eq:default-scales}, every tangent map $T_kS$ has numerical rank
four at relative singular-value tolerance $10^{-10}$.  Moreover, the
factorization
\begin{equation}
  \QFI_k=T_k^{\mathsf T}T_k
  \label{eq:numerical-tangent-factorization}
\end{equation}
holds with maximum relative coordinate-scaled Frobenius residual
\begin{equation}
  \max_k
  \frac{
    \left\|
      \QFI_k-T_k^{\mathsf T}T_k
    \right\|_{S,F}
  }{
    \|\QFI_k\|_{S,F}
  }
  =
  9.28\times10^{-16}.
  \label{eq:tangent-factorization-residual}
\end{equation}
Theorem~\ref{thm:pure-state-rank-bound} guarantees at least two exact null
directions for every six-coordinate pure-qutrit matrix.  The numerical
rank-four result indicates that this structural upper bound is saturated at
each of the 19 sampled energies, with no additional null direction resolved
at the declared tolerance.

\subsection{Spectral kernel-intersection theorem}
\label{subsec:kernel-intersection}

Suppose that an energy or channel label $k$ is observed and that independent
positive expected counts $N_k$ contribute information matrices $\QFI_k$.
The aggregate state QFI is
\begin{equation}
  \QFI_{\mathrm{tot}}
  =
  \sum_{k:N_k>0}N_k\QFI_k.
  \label{eq:spectral-qfi-sum}
\end{equation}

\begin{theorem}[Positive-weight kernel intersection]
\label{thm:positive-weight-kernel-intersection}
For positive-semidefinite matrices $\QFI_k$ and strictly positive active
weights,
\begin{equation}
  \ker\QFI_{\mathrm{tot}}
  =
  \bigcap_{k:N_k>0}\ker\QFI_k.
  \label{eq:kernel-intersection-theorem}
\end{equation}
\end{theorem}

\begin{proof}
If $\bm{v}$ lies in every active conditional kernel, then
$\QFI_{\mathrm{tot}}\bm{v}=0$.  Conversely, suppose
$\bm{v}\in\ker\QFI_{\mathrm{tot}}$.  Positivity gives
\begin{equation}
  0
  =
  \bm{v}^{\mathsf T}\QFI_{\mathrm{tot}}\bm{v}
  =
  \sum_{k:N_k>0}
  N_k\,
  \bm{v}^{\mathsf T}\QFI_k\bm{v},
  \label{eq:kernel-proof-sum}
\end{equation}
where every term is nonnegative.  Each term must therefore vanish.  For a
positive-semidefinite matrix,
\begin{equation}
  \bm{v}^{\mathsf T}\QFI_k\bm{v}=0
  \quad\Longleftrightarrow\quad
  \QFI_k\bm{v}=0.
\end{equation}
Thus $\bm{v}$ belongs to every active conditional kernel.
\end{proof}

The condition $N_k>0$ is essential.  A zero-weight bin contributes no
information and must be excluded from the intersection.  Negative weights
have no interpretation as independent Fisher information and invalidate the
positivity proof.

The theorem also follows directly from the tangent factorization.  Define
\begin{equation}
  T_{\mathrm{stack}}
  =
  \begin{pmatrix}
    \sqrt{N_1}\,T_1\\
    \sqrt{N_2}\,T_2\\
    \vdots
  \end{pmatrix},
  \label{eq:stacked-tangent-map}
\end{equation}
where only bins with $N_k>0$ are included.  Then
\begin{equation}
  \QFI_{\mathrm{tot}}
  =
  T_{\mathrm{stack}}^{\mathsf T}T_{\mathrm{stack}},
  \qquad
  \ker T_{\mathrm{stack}}
  =
  \bigcap_{k:N_k>0}\ker T_k.
  \label{eq:stacked-tangent-kernel}
\end{equation}
Spectral rank restoration therefore occurs exactly when the active
conditional null spaces cease to share a common direction.

\subsection{Energy-pair complementarity}
\label{subsec:energy-complementarity}

Each single-energy model in our 19-point grid has a two-dimensional numerical
kernel in the six-dimensional scaled parameter space at relative tolerance
$10^{-10}$.  Specifically, the kernel bases and principal angles below are
computed in the dimensionless coordinates $\bm{x}$ defined by
$\bm{\lambda}=S\bm{x}$.  Although the existence and dimension of an exact
kernel intersection are invariant under nonsingular scaling, the numerical
values of nonzero principal angles depend on the declared coordinate metric.

For two energies $E_a$ and $E_b$, let
\begin{equation}
  0\leq\phi_1\leq\phi_2\leq\frac{\pi}{2}
  \label{eq:kernel-principal-angles}
\end{equation}
be the principal angles between their two-dimensional numerical null spaces.
An exactly zero principal angle would signal a shared null direction.  A small
but numerically nonzero angle instead diagnoses approximate alignment at the
declared resolution.

We scan all $\binom{19}{2}=171$ distinct energy pairs using normalized equal
weights,
\begin{equation}
  \QFI_{ab}
  =
  \frac{1}{2}
  \left(
    \QFI_a+\QFI_b
  \right).
  \label{eq:equally-weighted-pair-qfi}
\end{equation}
Every pair has a numerically trivial kernel intersection and numerical rank
six at relative tolerance $10^{-10}$.  The aggregate-kernel dimension obtained
from $\QFI_{ab}$ agrees with the independently stacked
conditional-kernel calculation for every pair, with no dimension mismatch.
The stacked-tangent factorization is satisfied to a maximum relative
coordinate-scaled residual of $1.21\times10^{-15}$.

Numerical rank restoration varies greatly in strength.
Table~\ref{tab:energy-pair-complementarity} compares the most E-optimal pair
on the scanned grid with the weakest pair.

\begin{table}[t]
  \caption{
    Extremal normalized, equally weighted energy pairs among the 171 scanned
    pairs.  ``Best'' refers only to the declared 19-point grid and maximizes
    the smallest eigenvalue of
    $S^{\mathsf T}\QFI_{ab}S$.  The condition number is computed from all six
    scaled eigenvalues.  Numerical rank uses relative tolerance $10^{-10}$.
  }
  \label{tab:energy-pair-complementarity}
  \centering
  \begin{tabular}{@{}lcccc@{}}
    \toprule
    Pair [GeV]
      & $(\phi_1,\phi_2)$
      & $q_{\min}$
      & $\kappa$
      & $r_{\mathrm{eff}}(10^{-6})$ \\
    \midrule
    $(0.50,\,2.75)$
      & $(5.306^\circ,\,39.670^\circ)$
      & $1.4161\times10^{-2}$
      & $1.59\times10^{3}$
      & $6$ \\
    $(4.75,\,5.00)$
      & $(0.0320^\circ,\,0.2290^\circ)$
      & $1.0021\times10^{-7}$
      & $4.75\times10^{7}$
      & $4$ \\
    \bottomrule
  \end{tabular}
\end{table}

The weakest pair is numerically rank six at the declared tolerance, and
neither computed principal angle vanishes at the reported numerical
resolution.  This floating-point result is not, by itself, a symbolic proof
that the exact angles are nonzero.  Nevertheless, both kernels are clearly
nearly aligned at the declared scale, producing two extremely weak restored
directions.  Numerical full rank therefore answers whether a direction is
resolved at a stated tolerance, not whether it carries useful information at
a finite information scale.

Across all pairs, the correlation between $\phi_1$ and
$\log_{10}q_{\min}$ is $0.518$.  Null-space rotation is an important but not
exclusive determinant of spectral complementarity: information strength in
the non-null directions also affects the smallest aggregate eigenvalue.

The effective-rank distribution is summarized in
Table~\ref{tab:pair-effective-ranks}.

\begin{table}[t]
  \caption{
    Effective ranks of the 171 energy pairs that have numerical rank six at
    relative tolerance $10^{-10}$.  These counts use the scaled eigenvalues
    and the definition in Eq.~\eqref{eq:effective-rank}.
  }
  \label{tab:pair-effective-ranks}
  \centering
  \begin{tabular}{@{}lrrr@{}}
    \toprule
    Threshold & Rank 4 & Rank 5 & Rank 6 \\
    \midrule
    $\tau=10^{-6}$ & $6$  & $25$ & $140$ \\
    $\tau=10^{-4}$ & $63$ & $30$ & $78$ \\
    \bottomrule
  \end{tabular}
\end{table}

Define the normalized uniformly weighted 19-energy aggregate by
\begin{equation}
  \QFI_{\mathrm{19\,bin}}
  =
  \frac{1}{19}
  \sum_{k=1}^{19}\QFI_k.
  \label{eq:uniform-19-bin-qfi}
\end{equation}
Its scaled eigenvalues are
\begin{equation}
  \begin{split}
  \operatorname{eig}
  \left(
    S^{\mathsf T}\QFI_{\mathrm{19\,bin}}S
  \right)
  ={}&
  \bigl(
    0.004189,\,
    0.023149,\,
    0.165448,\,
    5.40775,\,
    6.92602,\,
    8.78010
  \bigr).
  \end{split}
  \label{eq:uniform-spectrum-eigenvalues}
\end{equation}
It has condition number $\kappa=2.10\times10^{3}$ and remains effective rank
six even at $\tau=10^{-4}$.  Broadband coverage can
therefore do more than resolve common null directions numerically:
sufficiently complementary energies can also stabilize the restored tangent
model.

\subsection{Retaining versus tracing the energy label}
\label{subsec:energy-label}

Spectral aggregation depends on whether the energy-bin label is retained as
part of the observation.  Let $K$ be a parameter-independent classical label
with probabilities $w_k$ and conditional states $\rho_k$.  The resolved
classical--quantum state is
\begin{equation}
  \rho_{KQ}
  =
  \sum_k
  w_k
  |k\rangle\langle k|
  \otimes
  \rho_k.
  \label{eq:resolved-classical-quantum-state}
\end{equation}
Because the label sectors are orthogonal and the weights are fixed controls,
\begin{equation}
  \QFI(\rho_{KQ})
  =
  \sum_k w_k\QFI(\rho_k).
  \label{eq:resolved-label-qfi}
\end{equation}
If the weights depend on the estimands, the complete QFI becomes
\begin{equation}
  \QFI_{ij}(\rho_{KQ})
  =
  \sum_{k:w_k>0}
  \frac{
    \partial_iw_k\,\partial_jw_k
  }{w_k}
  +
  \sum_{k:w_k>0}
  w_k\QFI_{ij}(\rho_k).
  \label{eq:parameter-dependent-weight-information}
\end{equation}
The first term is the classical source-weight information.  The formula is an
interior, support-regular identity: a component whose weight vanishes requires
the separate boundary analysis described in
Sec.~\ref{subsubsec:zero_probability_boundaries}.  Such a parameter-dependent
preparation model is distinct from the fixed-weight model used below.

If the label is discarded, the remaining state is
\begin{equation}
  \rho_Q
  =
  \Tr_K\rho_{KQ}
  =
  \sum_k w_k\rho_k.
  \label{eq:energy-traced-state}
\end{equation}
Partial trace is a parameter-independent quantum channel, so QFI monotonicity
gives
\begin{equation}
  \QFI(\rho_Q)
  \preceq
  \QFI(\rho_{KQ}).
  \label{eq:energy-trace-qfi-loss}
\end{equation}
The unresolved state is generally mixed even when every $\rho_k$ is pure.
A full-rank qutrit density matrix belongs to an eight-real-dimensional
manifold, so its QFIM may have rank six without contradicting the
four-dimensional pure-qutrit bound.

The model in Eq.~\eqref{eq:resolved-classical-quantum-state} represents a
classically resolved incoherent spectrum.  It is not a coherent energy
wavepacket.  A source with experimentally relevant coherence between energy
components would require a different Hilbert-space model and a declared
measurement capable of accessing that coherence.

\subsubsection{Flavor measurement with and without the label}

For a flavor measurement, retaining energy produces joint probabilities
\begin{equation}
  p_{k\beta}
  =
  w_k p_{\beta|k}.
  \label{eq:resolved-flavor-probabilities}
\end{equation}
With parameter-independent weights, the corresponding CFI is
\begin{equation}
  \CFI_{\mathrm{resolved}}
  =
  \sum_k w_k\CFI^{(k)}_{\mathrm{flavor}}.
  \label{eq:resolved-flavor-cfi}
\end{equation}
If energy is discarded before inference, the probabilities become
\begin{equation}
  p_\beta
  =
  \sum_k w_kp_{\beta|k},
  \label{eq:unresolved-flavor-probabilities}
\end{equation}
and classical data processing gives
\begin{equation}
  \CFI_{\mathrm{unresolved}}
  \preceq
  \CFI_{\mathrm{resolved}}.
  \label{eq:flavor-label-data-processing}
\end{equation}

A categorical distribution with $m$ outcomes has only $m-1$ independent
probabilities.  Therefore
\begin{equation}
  \rank\CFI_{\mathrm{unresolved}}
  \leq
  m-1.
  \label{eq:categorical-rank-bound}
\end{equation}
For an unresolved three-flavor measurement,
\begin{equation}
  \rank\CFI_{\mathrm{unresolved}}
  \leq2,
  \label{eq:unresolved-flavor-rank-bound}
\end{equation}
regardless of how many oscillation coordinates enter the three probabilities.

\subsection{Numerical energy-label comparison}
\label{subsec:energy-label-results}

Using the same baseline, matter configuration, source flavor, and oscillation
benchmark as in
Eqs.~\eqref{eq:identifiability-matter-configuration} and
\eqref{eq:identifiability-oscillation-benchmark}, we compare the five energies
\begin{equation}
  E=(1.5,\,2.0,\,2.5,\,3.0,\,3.5)~\mathrm{GeV}
\end{equation}
with fixed normalized weights
\begin{equation}
  w=(0.1,\,0.2,\,0.4,\,0.2,\,0.1).
\end{equation}
The energy-traced density matrix has eigenvalues
\begin{equation}
  \operatorname{eig}(\rho_Q)
  =
  \left(
    4.8772\times10^{-5},\,
    0.139994,\,
    0.859958
  \right),
  \label{eq:mixed-state-density-eigenvalues}
\end{equation}
and is numerically full rank with a minimum eigenvalue well separated from
binary64 roundoff.

\begin{table}[t]
  \caption{
    Numerical and effective information ranks in the controlled five-energy
    model.  Numerical rank uses relative tolerance $10^{-10}$, while
    effective rank uses $\tau=10^{-6}$; both are evaluated after the
    coordinate scaling in Eq.~\eqref{eq:default-scales}.
  }
  \label{tab:energy-label-ranks}
  \centering
  \begin{tabular}{@{}lccc@{}}
    \toprule
    Information layer
      & Numerical rank
      & Effective rank
      & Smallest scaled eigenvalue \\
    \midrule
    Energy-resolved state QFI
      & $6$ & $6$ & $3.64\times10^{-4}$ \\
    Energy-unresolved state QFI
      & $6$ & $6$ & $3.27\times10^{-4}$ \\
    Energy-resolved flavor CFI
      & $6$ & $5$ & $1.07\times10^{-8}$ \\
    Energy-unresolved flavor CFI
      & $2$ & $2$ & numerical zero \\
    \bottomrule
  \end{tabular}
\end{table}

The generalized eigenvalues of the unresolved mixed-state QFI relative to the
resolved state QFI are
\begin{equation}
  \left(
    0.419934,\,
    0.879166,\,
    0.935465,\,
    0.991895,\,
    0.999296,\,
    0.999974
  \right).
  \label{eq:energy-label-retention-spectrum}
\end{equation}
Thus tracing the label preserves nearly all state information in several
generalized directions but removes approximately $58.0\%$ of the information
in the most affected direction.

By the data-processing result in
Eq.~\eqref{eq:energy-trace-qfi-loss}, the difference
\begin{equation}
  \QFI_{\mathrm{resolved}}
  -
  \QFI_{\mathrm{unresolved}}
\end{equation}
is positive semidefinite.  Numerically, its minimum scaled eigenvalue is
\begin{equation}
  2.92\times10^{-6}.
  \label{eq:state-label-psd-margin}
\end{equation}
Likewise,
\begin{equation}
  \CFI_{\mathrm{resolved}}
  -
  \CFI_{\mathrm{unresolved}}
  \succeq0
\end{equation}
by classical data processing, and its computed minimum scaled eigenvalue is
\begin{equation}
  1.05\times10^{-8}.
  \label{eq:flavor-label-psd-margin}
\end{equation}

The resolved flavor CFI has numerical rank six at relative tolerance
$10^{-10}$, but its smallest eigenvalue is sufficiently weak that it has
effective rank five at $\tau=10^{-6}$ and a condition number of approximately
\begin{equation}
  8.8\times10^{8}.
\end{equation}
This is a concrete example in which spectral labels restore numerical rank at
a fine tolerance without producing robust simultaneous precision.

\subsection{Identifiability conclusions}

The results of this section establish the following hierarchy of statements:

\begin{enumerate}
  \item A single pure three-flavor state at a fixed preparation and
        propagation setting has QFIM rank at most four, irrespective of the
        number of coordinates used to label it.

  \item Spectrally resolved information restores a direction exactly when that
        direction is absent from the intersection of the active conditional
        kernels.

  \item A trivial common kernel establishes formal first-order local
        identifiability of the regular tangent model but does not control the
        magnitude of the restored eigenvalues.  A numerical calculation
        establishes this conclusion only at its declared scaling and
        tolerance.

  \item Retaining an energy label and tracing it out are different statistical
        models related by parameter-independent data processing.

  \item A mixed qutrit can have rank-six QFI, whereas an unresolved
        three-outcome flavor measurement has CFI rank at most two.

  \item A pseudoinverse constrains only functions whose differentials belong
        to the estimable cotangent space.  It does not manufacture information
        in exact null directions.
\end{enumerate}

Identifiability is necessary but not sufficient for simultaneous precision.
Even a full-rank aggregate QFIM may possess nonzero curvature and therefore
fail the multiparameter compatibility condition.
Section~\ref{sec:attainability} next restricts the model to its identifiable
tangent space and evaluates attainable multiparameter bounds.

\section{Multiparameter attainability}
\label{sec:attainability}

The rank analysis of Sec.~\ref{sec:identifiability} determines which local
parameter combinations are estimable, but it does not determine whether the
inverse SLD quantum Fisher information matrix can be attained simultaneously.
That second question is genuinely multiparameter: measurements that are
optimal for distinct tangent directions need not be compatible.  In this
section we therefore restrict every calculation to an identifiable tangent
model before testing compatibility or evaluating an attainable bound.

\subsection{SLD compatibility and weighted precision costs}
\label{subsec:sld_compatibility}

Let $L_i$ denote the symmetric logarithmic derivative (SLD),
\begin{equation}
  \partial_i\rho
  =
  \frac{1}{2}\left(\rho L_i+L_i\rho\right).
  \label{eq:sld_definition_attainability}
\end{equation}
With the conventions introduced above, the SLD QFIM and the SLD commutator
tensor are
\begin{align}
  \QFI_{ij}
  &=
  \operatorname{Re}\Tr\!\left(\rho L_iL_j\right),
  \label{eq:qfim_sld_attainability}
  \\
  \mathcal{D}_{ij}
  &=
  \operatorname{Im}\Tr\!\left(\rho L_iL_j\right)
   =
  \frac{1}{2i}\Tr\!\left(\rho[L_i,L_j]\right).
  \label{eq:curvature_sld_attainability}
\end{align}
We refer to $\mathcal{D}$ as the curvature tensor in our convention.  Other
definitions of the mean Uhlmann curvature may differ from
Eq.~\eqref{eq:curvature_sld_attainability} by an overall sign or numerical
factor.  For a pure state, the present convention gives
\begin{equation}
  \mathcal{D}_{ij}
  =
  4\,\operatorname{Im}\langle h_i|h_j\rangle,
  \qquad
  |h_i\rangle
  =
  P_\perp|\partial_i\psi\rangle,
  \label{eq:pure_state_curvature_horizontal}
\end{equation}
where the horizontal tangents were defined in
Eq.~\eqref{eq:projected-pure-tangents}.

All compatibility statements below are made after restriction to the
identifiable tangent model.  The coordinate-independent weak-commutativity
condition is
\begin{equation}
  \left.
  \mathcal{D}
  \right|_{\mathcal{T}_{\mathrm{id}}}
  =
  0.
  \label{eq:weak_commutativity}
\end{equation}
Equivalently,
$\bm{u}^{\mathsf T}\mathcal{D}\bm{v}=0$ for every pair of identifiable
tangent vectors $\bm{u}$ and $\bm{v}$.  This condition is necessary for joint
saturation of the SLD matrix bound and, for regular pure-state models, is also
sufficient for attaining that bound locally
\cite{Matsumoto2002,Ragy2016,RagyErratum2019}.  Weak commutativity must not be confused with
identifiability: an invertible QFIM may be incompatible, while a singular
QFIM may contain a lower-dimensional compatible tangent model.

For an estimator with covariance $V$, introduce a positive-definite weight
matrix $W$ on the identifiable tangent model and the scalar loss
\begin{equation}
  \mathcal{C}(W,V)=\operatorname{tr}(WV),
  \label{eq:weighted_covariance_cost}
\end{equation}
where $\operatorname{tr}$ acts on parameter space and $\Tr$ acts on the
quantum Hilbert space.  The scalar SLD bound is
\begin{equation}
  C_{\mathrm{S}}(W)
  =
  \operatorname{tr}\!\left(W\QFI^{-1}\right).
  \label{eq:sld_scalar_cost}
\end{equation}
Although $C_{\mathrm{S}}$ is always a lower bound on the locally unbiased
cost of an identifiable model, it need not be attainable.

A scalar multiparameter statement is incomplete until $W$ has been specified.
Under a nonsingular reparameterization, $W$ transforms covariantly with
$\QFI$, while $V$ transforms contravariantly, leaving
$\operatorname{tr}(WV)$ invariant.  We use
\begin{equation}
  W=\QFI
  \label{eq:metric_weight}
\end{equation}
for intrinsic geometric comparisons.  For a $p$-dimensional identifiable
tangent, this convention gives
\begin{equation}
  C_{\mathrm{S}}(\QFI)=p.
  \label{eq:metric_weight_sld_cost}
\end{equation}
It does not replace a scientifically chosen loss function in a
phenomenological forecast.

\subsection{Holevo bound and its hierarchy}
\label{subsec:holevo_definition}

Let $X_i$ be Hermitian influence operators satisfying the
local-unbiasedness conditions
\begin{align}
  \Tr(\rho X_i)&=0,
  \label{eq:holevo_zero_mean}
  \\
  \Tr\!\left[(\partial_j\rho)X_i\right]&=\delta_{ij}.
  \label{eq:holevo_local_unbiasedness}
\end{align}
Define their complex Gram matrix
\begin{equation}
  Z_{ij}(\bm{X})
  =
  \Tr(\rho X_iX_j).
  \label{eq:holevo_gram}
\end{equation}
The Holevo Cram\'er--Rao bound is
\begin{equation}
  C_{\mathrm{H}}(W)
  =
  \min_{\bm{X}}
  \left[
    \operatorname{tr}\!\left(W\,\operatorname{Re}Z(\bm{X})\right)
    +
    \left\|
      \sqrt{W}\,\operatorname{Im}Z(\bm{X})\sqrt{W}
    \right\|_1
  \right],
  \label{eq:holevo_bound}
\end{equation}
where $\|\cdot\|_1$ is the trace norm
\cite{Holevo2011,Matsumoto2002}.  For a regular i.i.d.\ model and a locally
unbiased or locally asymptotically unbiased sequence of estimators,
\begin{equation}
  \liminf_{n\rightarrow\infty}
  n\,\operatorname{tr}\!\left(WV^{(n)}\right)
  \geq C_{\mathrm{H}}(W).
  \label{eq:holevo_asymptotic_bound}
\end{equation}
The bound is asymptotically attainable under the standard local regularity
conditions, generally allowing collective measurements
\cite{YamagataFujiwaraGill2013}.  For regular pure-state models, the Holevo
value can instead be attained with separable measurements, although the
locally optimal POVM may depend on the expansion point and may require a
preliminary localization or adaptive stage
\cite{Matsumoto2002,WangChenYuan2026PRL,WangChenYuan2026PRA}.

A useful feasible upper expression is obtained by choosing the influence
operators generated by the SLD tangent,
\begin{equation}
  X_i^{(\mathrm{S})}
  =
  \sum_j
  \left(\QFI^{-1}\right)_{ij}L_j.
  \label{eq:sld_generated_influence_operators}
\end{equation}
The resulting cost, often called the $D$-invariant upper expression, is
\begin{equation}
  C_{\mathrm{D}}(W)
  =
  C_{\mathrm{S}}(W)
  +
  \left\|
    \sqrt{W}\,
    \QFI^{-1}\mathcal{D}\QFI^{-1}
    \sqrt{W}
  \right\|_1.
  \label{eq:d_invariant_upper}
\end{equation}
For a general model, Eq.~\eqref{eq:d_invariant_upper} is an upper bound on the
Holevo value obtained from a particular feasible choice of influence
operators.  It equals $C_{\mathrm{H}}$ when the model is $D$-invariant.  More
generally,
\begin{equation}
  C_{\mathrm{S}}(W)
  \leq
  C_{\mathrm{H}}(W)
  \leq
  C_{\mathrm{D}}(W)
  \leq
  2C_{\mathrm{S}}(W),
  \label{eq:holevo_factor_two_hierarchy}
\end{equation}
and the final factor of two is sharp
\cite{TsangAlbarelliDatta2020}.  Equation
\eqref{eq:holevo_factor_two_hierarchy} also provides a consistency check on
every numerical calculation below.

\subsection{Exact pure-qutrit results}
\label{subsec:pure_qutrit_holevo}

For a two-parameter identifiable pure-state model, define the invariant
incompatibility coefficient
\begin{equation}
  \beta
  =
  \max_{\zeta\in
  \operatorname{spec}(\QFI^{-1}\mathcal{D})}
  |\operatorname{Im}\zeta|
  =
  \frac{|\mathcal{D}_{12}|}{\sqrt{\det\QFI}},
  \qquad
  0\leq\beta\leq1.
  \label{eq:beta_definition}
\end{equation}
Let
\begin{equation}
  \Omega
  =
  \QFI^{-1/2}W\QFI^{-1/2},
  \qquad
  \omega_1
  =
  \lambda_1(\Omega)
  \geq
  \omega_2
  =
  \lambda_2(\Omega),
  \qquad
  \eta=\frac{1}{2}\arcsin\beta.
  \label{eq:canonical_pair_weight}
\end{equation}
The exact attainable two-parameter pure-state cost is
\begin{equation}
  C_{\mathrm{H}}(W)
  =
  \min_{0\leq\varphi\leq\eta}
  \left[
    \frac{\omega_1}
         {\cos^2(\varphi-\eta)}
    +
    \frac{\omega_2}
         {\cos^2(\varphi+\eta)}
  \right]
  \label{eq:general_two_parameter_pure_bound}
\end{equation}
\cite{Matsumoto2002,YungYungConlonAssad2025}.  For the metric weight
$W=\QFI$, this becomes
\begin{align}
  C_{\mathrm{S}}(\QFI)&=2,
  \nonumber\\
  C_{\mathrm{H}}(\QFI)
  &=
  \frac{4}{1+\sqrt{1-\beta^2}},
  \label{eq:metric_weight_pair_holevo}
  \\
  \frac{C_{\mathrm{H}}}{C_{\mathrm{S}}}
  &=
  \frac{2}{1+\sqrt{1-\beta^2}}.
  \label{eq:metric_weight_pair_ratio}
\end{align}
Thus $\beta=0$ gives the compatible SLD cost, whereas $\beta=1$ gives the
maximal factor-two penalty.

Recent general pure-state tradeoff relations provide tight analytical bounds
and constructive optimal measurements for an arbitrary number of identifiable
parameters
\cite{WangChenYuan2026PRL,WangChenYuan2026PRA}.  For completeness, and to
make the specialization to the qutrit geometry explicit, we give direct
derivations of the three- and four-dimensional metric-weighted costs used
below.

\begin{theorem}[Three-dimensional pure-qutrit cost]
\label{thm:three-dimensional-pure-qutrit-cost}
Let a regular three-parameter pure-state model in $\mathbb{C}^{3}$ have an
invertible QFIM.  For the metric weight $W=\QFI$,
\begin{equation}
  C_{\mathrm{S}}(\QFI)=3,
  \qquad
  C_{\mathrm{H}}(\QFI)=5,
  \qquad
  \frac{C_{\mathrm{H}}}{C_{\mathrm{S}}}
  =
  \frac{5}{3}.
  \label{eq:three_parameter_qutrit_exact}
\end{equation}
\end{theorem}

\begin{proof}
Define the QFIM-whitened horizontal tangents
\begin{equation}
  |t_a\rangle
  =
  2\sum_i
  \left(\QFI^{-1/2}\right)_{ia}
  |h_i\rangle,
  \qquad
  |h_i\rangle
  =
  \left(
    \mathbb{I}-|\psi\rangle\langle\psi|
  \right)
  |\partial_i\psi\rangle.
  \label{eq:whitened_pure_tangents}
\end{equation}
They satisfy
\begin{equation}
  \operatorname{Re}\langle t_a|t_b\rangle
  =
  \delta_{ab}.
  \label{eq:whitened_tangent_orthonormality}
\end{equation}
Equip the horizontal space with the real inner product
\begin{equation}
  (\bm{u},\bm{v})_{\mathbb{R}}
  =
  \operatorname{Re}\langle u|v\rangle.
  \label{eq:horizontal_real_inner_product}
\end{equation}
For a qutrit, the horizontal space is isomorphic to $\mathbb{C}^{2}$ and
therefore has real dimension four.  The three whitened tangents span a real
hyperplane.  The group $U(2)$ acts transitively on the real unit sphere in
$\mathbb{C}^{2}$, so its unit normal can be mapped to
\begin{equation}
  |n\rangle=(0,1).
\end{equation}
An orthogonal transformation of the whitened parameter coordinates, which
preserves the metric weight, then brings the tangent basis to
\begin{equation}
  |t_1\rangle=(1,0),
  \qquad
  |t_2\rangle=(i,0),
  \qquad
  |t_3\rangle=(0,i).
  \label{eq:canonical_qutrit_hyperplane}
\end{equation}

In the pure-state influence-vector reduction of the Holevo problem
\cite{Matsumoto2002}, the horizontal vectors $|x_a\rangle$ obey
\begin{equation}
  \operatorname{Re}\langle x_a|t_b\rangle
  =
  \delta_{ab}.
  \label{eq:pure_influence_local_unbiasedness}
\end{equation}
Because the $|t_a\rangle$ span the real hyperplane orthogonal to $|n\rangle$,
the most general solution is
\begin{equation}
  |x_a\rangle
  =
  |t_a\rangle+a_a|n\rangle,
  \qquad
  a_a\in\mathbb{R}.
  \label{eq:hyperplane_influence_vectors}
\end{equation}
For $W=\QFI$, the parameter coordinates are whitened and the real Gram
contribution is
\begin{equation}
  \operatorname{tr}\operatorname{Re}Z
  =
  3+a_1^2+a_2^2+a_3^2.
  \label{eq:hyperplane_real_gram_cost}
\end{equation}
The imaginary Gram matrix is
\begin{equation}
  \operatorname{Im}Z
  =
  \begin{pmatrix}
    0 & 1 & a_1\\
    -1 & 0 & a_2\\
    -a_1 & -a_2 & 0
  \end{pmatrix}.
  \label{eq:hyperplane_imaginary_gram}
\end{equation}
Its singular values are
\begin{equation}
  \left(
    \sqrt{1+a_1^2+a_2^2},\,
    \sqrt{1+a_1^2+a_2^2},\,
    0
  \right),
\end{equation}
so the Holevo objective reduces to
\begin{equation}
  h(\bm{a})
  =
  3+a_1^2+a_2^2+a_3^2
  +
  2\sqrt{1+a_1^2+a_2^2}.
  \label{eq:hyperplane_holevo_objective}
\end{equation}
Every term added to the constant value five is nonnegative, and the minimum
is attained at $\bm{a}=0$.  Hence
\begin{equation}
  \min_{\bm{a}}h(\bm{a})=5.
\end{equation}
Since the metric-weighted SLD cost equals the identifiable dimension,
$C_{\mathrm{S}}=3$, which proves
Eq.~\eqref{eq:three_parameter_qutrit_exact}.
\end{proof}

The same canonical geometry implies
\begin{equation}
  \operatorname{spec}_{\mathrm{sing}}
  \left(
    \QFI^{-1/2}\mathcal{D}\QFI^{-1/2}
  \right)
  =
  (1,1,0),
  \label{eq:whitened_curvature}
\end{equation}
where $\operatorname{spec}_{\mathrm{sing}}$ denotes the ordered singular-value
spectrum.

\begin{corollary}[Locally complete pure-qutrit cost]
\label{cor:complete-pure-qutrit-cost}
For a regular four-parameter pure-qutrit model whose tangent spans the full
real tangent space of $\mathbb{CP}^{2}$, the metric-weighted costs are
\begin{equation}
  C_{\mathrm{S}}(\QFI)=4,
  \qquad
  C_{\mathrm{H}}(\QFI)=8,
  \qquad
  \frac{C_{\mathrm{H}}}{C_{\mathrm{S}}}=2.
  \label{eq:complete_qutrit_holevo}
\end{equation}
\end{corollary}

\begin{proof}
After QFIM whitening, the four tangents form a real orthonormal basis of the
complete horizontal space $\mathbb{C}^{2}$.  The local-unbiasedness conditions
therefore fix the horizontal influence vectors uniquely.  Their real Gram
matrix is $\mathbb{I}_4$, while the imaginary Gram matrix is the canonical
complex structure on $\mathbb{R}^{4}$.  Its four singular values are all
equal to one.  The two terms in the Holevo objective are consequently four
and four, giving $C_{\mathrm{H}}=8$.  The metric-weighted SLD cost is the
tangent dimension, $C_{\mathrm{S}}=4$.
\end{proof}

The locally complete pure-qutrit model therefore realizes the sharp
factor-two separation in Eq.~\eqref{eq:holevo_factor_two_hierarchy}.

\subsection{Finite-dimensional semidefinite formulation}
\label{subsec:holevo_sdp}

We independently evaluate incompatible models with the finite-dimensional
semidefinite formulation of Ref.~\cite{AlbarelliFrielDatta2019}.  Let $R$ be
a square-root factor of $\rho$ restricted to its support, normalized by
$RR^\dagger=\rho$.  If the Hilbert-space dimension is $d$ and the state rank
is $r$, then $R$ may be represented as a $d\times r$ matrix.  Define
\begin{equation}
  |y_i\rangle
  =
  \operatorname{vec}(X_iR),
  \qquad
  Y
  =
  (|y_1\rangle,\ldots,|y_p\rangle).
  \label{eq:holevo_sdp_vectors}
\end{equation}
Then
\begin{equation}
  Y^\dagger Y=Z(\bm{X}).
  \label{eq:holevo_sdp_gram}
\end{equation}
Introducing a real symmetric matrix $V$, the Holevo optimization can be
written
\begin{align}
  C_{\mathrm{H}}(W)
  &=
  \min_{\bm{X},V}\operatorname{tr}(WV),
  \label{eq:holevo_sdp_objective}
  \\
  \text{subject to}\qquad
  &
  \begin{pmatrix}
    V & Y^\dagger\\
    Y & \mathbb{I}_{dr}
  \end{pmatrix}
  \succeq0,
  \label{eq:holevo_sdp_lmi}
\end{align}
together with Eqs.~\eqref{eq:holevo_zero_mean} and
\eqref{eq:holevo_local_unbiasedness}.  Minimizing over the real matrix $V$
reproduces the trace-norm term in Eq.~\eqref{eq:holevo_bound}.  Indeed, the
Schur complement gives $V\succeq Y^\dagger Y=Z$, and the minimum real
covariance majorant satisfies
$\min_{V\in\mathbb S^p:\,V\succeq Z}\operatorname{tr}(WV)
=\operatorname{tr}(W\operatorname{Re}Z)
+\|\sqrt W\operatorname{Im}Z\sqrt W\|_1$.

Our implementation expands each $X_i$ in an orthonormal Hermitian basis,
whitens the parameter coordinates by the QFIM, eliminates the affine
local-unbiasedness constraints, and removes directions that annihilate the
support of $\sqrt{\rho}$.  A log-determinant barrier then solves the
rank-reduced LMI.  In exact arithmetic, feasible primal and dual points obey
the weak-duality bracket
\begin{equation}
  C_{\mathrm{dual}}
  \leq C_{\mathrm{H}}
  \leq C_{\mathrm{primal}}.
  \label{eq:holevo_primal_dual_interval}
\end{equation}
In floating-point calculations, we retain a numerical result only when the
primal and dual points satisfy the declared affine, cone, and
local-unbiasedness tolerances specified in
Appendix~\ref{app:holevo-certification}.  The resulting interval is therefore
reported as a tolerance-qualified numerical primal--dual bracket, not as an
exact-arithmetic or interval-arithmetic certificate.  The recorded diagnostics
include the primal--dual gap, local-unbiasedness residual, KKT residual, dual
affine residual, and minimum eigenvalues of the primal and dual LMIs.

For the three-parameter scans below, five requested barrier accuracies were
tested independently.  Only attempts satisfying the declared numerical
primal and dual feasibility tolerances were eligible, and the
tolerance-qualified numerical bracket with the smallest gap was retained.
This procedure avoids interpreting a numerically accurate primal value as a
numerically bracketed optimum when the corresponding dual iterate fails its
feasibility test.

\subsection{Exhaustive attainability atlas}
\label{subsec:attainability_atlas}

We evaluated all $2^6-1=63$ nonempty coordinate subsets at each of two
pure-state benchmark points:

\begin{enumerate}
  \item vacuum propagation at
        $L/E=520~\mathrm{km/GeV}$;
  \item constant-density propagation at
        $L=1300~\mathrm{km}$,
        $E=2.5~\mathrm{GeV}$,
        $\rho_{\mathrm m}=2.848~\mathrm{g\,cm^{-3}}$, and
        $Y_e=0.5$.
\end{enumerate}

The source is $\nu_\mu$, and the oscillation parameters are the benchmark
specified in Eq.~\eqref{eq:identifiability-oscillation-benchmark}.  Coordinates
outside a listed subset are held fixed; the pairwise and low-dimensional
results are therefore conditional submodel bounds, not nuisance-profiled
bounds.

Numerical rank is diagnosed in the dimensionless coordinates defined by
\begin{equation}
  \bm{\lambda}=S\bm{x},
  \qquad
  S
  =
  \diag
  \left(
    1,1,1,1,10^{-4},10^{-3}
  \right),
  \label{eq:attainability_coordinate_scales}
\end{equation}
with relative eigenvalue tolerance $10^{-10}$.  The final two scale entries
carry units of $\mathrm{eV}^{2}$.  Scaling does not alter exact rank.  It fixes
the numerical meaning of rank diagnostics and selects a representative
estimable complement when a requested coordinate model is singular.

Table~\ref{tab:attainability_atlas} gives the exhaustive numerical
classification.  Every requested subset with at most four coordinates has
numerical rank equal to its dimension at both benchmark points.  Every five-
or six-coordinate request has numerical rank four.  Independently of the
numerical calculation, Theorem~\ref{thm:pure-state-rank-bound} guarantees that
such a pure-qutrit request has rank at most four and therefore admits no joint
locally unbiased bound for all requested coordinates.

\begin{table}[t]
  \centering
  \caption{
    Exhaustive attainability classification at each pure-state benchmark.
    The number of subsets is per benchmark point, and the rank column reports
    numerical rank at relative tolerance $10^{-10}$ after the scaling in
    Eq.~\eqref{eq:attainability_coordinate_scales}.  Ratios use $W=\QFI$.
    A dash indicates that the requested coordinate model is singular.  The
    entries in parentheses for five and six coordinates apply only after
    restriction to a scale-declared four-dimensional estimable tangent.
  }
  \label{tab:attainability_atlas}
  \begin{tabular}{ccclcc}
    \toprule
    Requested
    & Number
    & Numerical
    & Evaluation
    & Vacuum
    & Matter
    \\
    coordinates
    & of subsets
    & rank
    & route
    & $C_{\mathrm H}/C_{\mathrm S}$
    & $C_{\mathrm H}/C_{\mathrm S}$
    \\
    \midrule
    1 & 6  & 1
      & compatible SLD
      & $1$
      & $1$
      \\
    2 & 15 & 2
      & exact pure-state pair
      & $1$--$1.812689$
      & $1.000001$--$1.762325$
      \\
    3 & 20 & 3
      & Theorem~\ref{thm:three-dimensional-pure-qutrit-cost}
        plus SDP
      & $5/3$
      & $5/3$
      \\
    4 & 15 & 4
      & Corollary~\ref{cor:complete-pure-qutrit-cost}
      & $2$
      & $2$
      \\
    5 & 6  & 4
      & singular; reduced tangent
      & -- $(2)$
      & -- $(2)$
      \\
    6 & 1  & 4
      & singular; reduced tangent
      & -- $(2)$
      & -- $(2)$
      \\
    \bottomrule
  \end{tabular}
\end{table}

For all 40 three-coordinate submodels across the two benchmarks, the
numerically computed singular values of
$\QFI^{-1/2}\mathcal{D}\QFI^{-1/2}$ agree with $(1,1,0)$ to better than
$4.4\times10^{-14}$.  Every independently computed tolerance-qualified
numerical primal--dual interval contains the exact value
$C_{\mathrm H}=5$ from
Theorem~\ref{thm:three-dimensional-pure-qutrit-cost}.  The largest retained
interval width is $1.50\times10^{-6}$, and most are below
$4\times10^{-7}$.

Table~\ref{tab:pairwise_holevo_atlas} reports the complete two-coordinate
atlas.  The most incompatible pair at both points is
$(\theta_{23},\Delta m^2_{31})$, followed by
$(\theta_{13},\dcp)$.  At relative curvature tolerance $10^{-10}$, the vacuum
calculation classifies two pairs as weakly commuting:
$(\theta_{12},\Delta m^2_{31})$ and
$(\Delta m^2_{21},\Delta m^2_{31})$.  No constant-matter pair is classified
as weakly commuting at that tolerance, although several ratios are
numerically close to unity.  These tolerance-based classifications are not
claims that the corresponding curvature components have been proven
symbolically to vanish or remain nonzero.

\begin{table*}[t]
  \centering
  \small
  \caption{
    Exact pairwise pure-state Holevo formula evaluated numerically for
    $W=\QFI_{\mathrm{pair}}$.  Coordinates outside each pair are fixed.  The
    coefficient $\beta$ is defined in Eq.~\eqref{eq:beta_definition}; displayed
    zeros are rounded numerical values and weak-commutativity classifications
    use relative tolerance $10^{-10}$.
  }
  \label{tab:pairwise_holevo_atlas}
  \begin{tabular}{lcccc}
    \toprule
    &
    \multicolumn{2}{c}{Vacuum}
    &
    \multicolumn{2}{c}{Constant matter}
    \\
    \cmidrule(lr){2-3}
    \cmidrule(lr){4-5}
    Parameter pair
    & $\beta$
    & $C_{\mathrm H}/C_{\mathrm S}$
    & $\beta$
    & $C_{\mathrm H}/C_{\mathrm S}$
    \\
    \midrule
    $(\theta_{12},\theta_{13})$
      & 0.562836 & 1.094949 & 0.573574 & 1.099412 \\
    $(\theta_{12},\theta_{23})$
      & 0.725698 & 1.184825 & 0.711449 & 1.174579 \\
    $(\theta_{12},\dcp)$
      & 0.367965 & 1.036355 & 0.334144 & 1.029589 \\
    $(\theta_{12},\Delta m^2_{21})$
      & 0.170760 & 1.007398 & 0.246009 & 1.015606 \\
    $(\theta_{12},\Delta m^2_{31})$
      & 0.000000 & 1.000000 & 0.002138 & 1.000001 \\
    $(\theta_{13},\theta_{23})$
      & 0.010323 & 1.000027 & 0.025188 & 1.000159 \\
    $(\theta_{13},\dcp)$
      & 0.970165 & 1.609731 & 0.957832 & 1.553604 \\
    $(\theta_{13},\Delta m^2_{21})$
      & 0.673299 & 1.149841 & 0.675910 & 1.151421 \\
    $(\theta_{13},\Delta m^2_{31})$
      & 0.223617 & 1.012824 & 0.261325 & 1.017682 \\
    $(\theta_{23},\dcp)$
      & 0.121832 & 1.003739 & 0.148237 & 1.005555 \\
    $(\theta_{23},\Delta m^2_{21})$
      & 0.751004 & 1.204603 & 0.764436 & 1.216028 \\
    $(\theta_{23},\Delta m^2_{31})$
      & 0.994647 & 1.812689 & 0.990864 & 1.762325 \\
    $(\dcp,\Delta m^2_{21})$
      & 0.429201 & 1.050857 & 0.404886 & 1.044732 \\
    $(\dcp,\Delta m^2_{31})$
      & 0.012321 & 1.000038 & 0.061022 & 1.000933 \\
    $(\Delta m^2_{21},\Delta m^2_{31})$
      & 0.000000 & 1.000000 & 0.006466 & 1.000010 \\
    \bottomrule
  \end{tabular}
\end{table*}

The scalar ratio is weight dependent.  For example, the vacuum
$(\theta_{13},\dcp)$ pair gives
\begin{equation}
  C_{\mathrm S}(\QFI)=2,
  \qquad
  C_{\mathrm H}(\QFI)=3.2194618,
  \qquad
  \frac{C_{\mathrm H}}{C_{\mathrm S}}=1.6097309.
  \label{eq:theta13_delta_metric_cost}
\end{equation}
For the identity weight in the declared native coordinate units, the same
physical submodel instead gives
\begin{equation}
  C_{\mathrm S}(\mathbb{I}_2)=6.5319925,
  \qquad
  C_{\mathrm H}(\mathbb{I}_2)=7.3786983,
  \qquad
  \frac{C_{\mathrm H}}{C_{\mathrm S}}=1.1296244.
  \label{eq:theta13_delta_identity_cost}
\end{equation}
The native identity weight is included only to demonstrate weight dependence;
because its numerical meaning changes when heterogeneous coordinate units are
changed, it is not an intrinsic scientific loss.  There is therefore no
unique scalar ``Holevo penalty'' without a declared cost matrix and unit
convention.

\subsection{Six-coordinate mixed spectral numerical primal--dual bracket}
\label{subsec:mixed_spectral_holevo}

A single pure qutrit cannot identify six coordinates, but the
energy-unresolved mixed state constructed in
Sec.~\ref{subsec:energy-label-results} has numerical state rank three and a
six-coordinate QFIM of numerical rank six at relative tolerance $10^{-10}$.
We consider
\begin{equation}
  \rho_{\mathrm{mix}}
  =
  \sum_{k=1}^{5}w_k
  |\psi(E_k)\rangle\langle\psi(E_k)|,
  \label{eq:mixed_spectral_state_attainability}
\end{equation}
with
\begin{align}
  E_k/\mathrm{GeV}
  &=
  (1.5,\,2.0,\,2.5,\,3.0,\,3.5),
  \nonumber\\
  w_k
  &=
  (0.1,\,0.2,\,0.4,\,0.2,\,0.1).
  \label{eq:mixed_spectral_configuration}
\end{align}
The baseline, matter configuration, source flavor, and oscillation benchmark
are those in Eqs.~\eqref{eq:identifiability-matter-configuration} and
\eqref{eq:identifiability-oscillation-benchmark}.  For
$W=\QFI_{\mathrm{mix}}$, the SLD cost is six, while the retained
tolerance-qualified numerical primal--dual Holevo bracket is
\begin{equation}
  11.205456676
  \leq
  C_{\mathrm H}(\QFI_{\mathrm{mix}})
  \leq
  11.205462701.
  \label{eq:mixed_spectral_holevo_interval}
\end{equation}
Equivalently,
\begin{equation}
  1.867576113
  \leq
  \frac{C_{\mathrm H}}{C_{\mathrm S}}
  \leq
  1.867577117.
  \label{eq:mixed_spectral_holevo_ratio}
\end{equation}
The complete numerical diagnostics are summarized in
Table~\ref{tab:mixed_spectral_holevo_certificate}.

\begin{table}[t]
  \centering
  \caption{
    Tolerance-qualified numerical primal--dual Holevo bracket for the
    six-coordinate energy-unresolved mixed spectral state with
    $W=\QFI_{\mathrm{mix}}$.  State and QFIM ranks are numerical ranks;
    the QFIM rank uses relative tolerance $10^{-10}$ after the declared
    coordinate scaling.  Primal and dual feasibility are assessed using the
    floating-point tolerances stated in the text.
  }
  \label{tab:mixed_spectral_holevo_certificate}
  \begin{tabular}{lr}
    \toprule
    Quantity & Value \\
    \midrule
    Numerical state rank & $3$ \\
    Numerical QFIM rank & $6$ \\
    $C_{\mathrm S}$ & $6.000000000$ \\
    Dual objective (lower endpoint at tolerance) & $11.205456676$ \\
    Primal objective (upper endpoint at tolerance) & $11.205462701$ \\
    Primal--dual gap & $6.02\times10^{-6}$ \\
    $C_{\mathrm D}$ & $11.312894780$ \\
    Local-unbiasedness residual & $2.62\times10^{-14}$ \\
    KKT residual & $1.08\times10^{-9}$ \\
    Minimum primal LMI eigenvalue & $6.67\times10^{-8}$ \\
    Minimum dual LMI eigenvalue & $1.28\times10^{-7}$ \\
    \bottomrule
  \end{tabular}
\end{table}

The endpoints satisfy the declared floating-point primal and dual feasibility
tests.  This is a numerical primal--dual bracket at the stated tolerances, not
an interval-arithmetic or exact-arithmetic certificate.  Throughout the
retained bracket, the Holevo cost is approximately $86.8\%$ larger than the
inverse SLD cost in this controlled mixed-state example.  Equivalently, the
SLD cost lies approximately $46.5\%$ below the numerically bracketed attainable
asymptotic cost.  The $D$-invariant upper expression is close to the numerical
Holevo bracket but is not exact for this model.

\subsection{Operational interpretation}
\label{subsec:attainability_interpretation}

Several qualifications are essential.

First, the results in Tables~\ref{tab:attainability_atlas} and
\ref{tab:pairwise_holevo_atlas} are local state-level bounds.  Pairwise
calculations fix all excluded oscillation coordinates; they do not profile
them as unknown nuisances.

Second, a singular five- or six-coordinate pure-state request has no joint
locally unbiased covariance and hence no Holevo cost for all requested
coordinates.  The value $C_{\mathrm H}/C_{\mathrm S}=2$ reported after
reduction applies only to four estimable local combinations.  Their individual
orientations depend on the declared coordinate scaling used to choose a
representative complement to the exact null space.

Third, $W=\QFI$ provides a coordinate-invariant diagnostic of geometric
incompatibility, but it is not a phenomenological loss function.  Forecasts
must specify which physical parameter combinations are scientifically
important and assign their weights in meaningful units.

Fourth, attainment of a pure-state Holevo value refers to an optimized quantum
measurement on the oscillated state.  It does not imply that ordinary flavor
detection attains that value.  The flavor POVM, energy reconstruction,
efficiencies, backgrounds, and nuisance constraints define a more restrictive
measurement channel.

Finally, the mixed spectral example is a controlled per-probe ensemble, not a
DUNE sensitivity forecast.  Its tolerance-qualified numerical primal--dual
Holevo bracket constrains the attainable asymptotic cost before imposing a
detector measurement and under the stated local regularity assumptions.  It
does so at the declared floating-point feasibility tolerances and is not an
exact-arithmetic or interval-arithmetic proof.  The next section makes that
measurement restriction explicit and propagates the information through the
reconstructed-event likelihood.

\section{From quantum measurements to detector likelihoods}
\label{sec:detector_likelihood}

The state QFIM and Holevo bounds of the preceding sections describe limits
before a detector measurement has been specified.  A neutrino experiment,
however, does not implement an arbitrary optimized POVM on the propagated
state.  It records interaction products, assigns reconstructed flavors and
energies, applies event selections, and fits binned counts in the presence of
backgrounds and systematic uncertainties.  These operations define a
measurement channel whose information can be much smaller than the state
QFIM, even when every oscillation coordinate remains numerically identifiable
at a declared tolerance.

Under the conditions stated in Sec.~\ref{sec:model}, we derive the hierarchy
\begin{equation}
  \QFI_{\mathrm{state}}
  \succeq
  \CFI_{\mathrm{flavor}}
  \succeq
  \CFI_{\mathrm{smear}}
  \succeq
  \CFI_{\mathrm{det}}
  \succeq
  \CFI_{\mathrm{prof}}.
  \label{eq:operational_information_hierarchy}
\end{equation}
Specifically, the source weights and classical processing maps are independent
of the physics estimands at fixed nuisance values, the physical POVM is fixed,
and every comparison uses the same exposure and per-exposure sample space.
Nuisance dependence of a detector map is included in the joint likelihood
before profiling.  The inequalities in
Eq.~\eqref{eq:operational_information_hierarchy} are Loewner-order statements,
not merely comparisons of traces or diagonal entries.

\subsection{Classical information of a fixed quantum measurement}
\label{subsec:fixed_povm_cfi}

Let $\{M_y\}$ be a parameter-independent physical POVM,
\begin{equation}
  M_y\succeq0,
  \qquad
  \sum_y M_y=\mathbb{I}.
  \label{eq:povm_definition}
\end{equation}
For the state $\rho(\bm{\lambda})$, the outcome probabilities and derivatives
are
\begin{align}
  p_y(\bm{\lambda})
  &=
  \Tr\!\left[\rho(\bm{\lambda})M_y\right],
  \label{eq:povm_probability}
  \\
  \partial_i p_y
  &=
  \Tr\!\left[(\partial_i\rho)M_y\right].
  \label{eq:povm_probability_derivative}
\end{align}
At a regular interior point with $p_y>0$, the measurement CFI is
\begin{equation}
  \CFI^{(M)}_{ij}
  =
  \sum_y
  \frac{\partial_i p_y\,\partial_j p_y}{p_y}.
  \label{eq:fixed_measurement_cfi}
\end{equation}
For every fixed POVM,
\begin{equation}
  \CFI^{(M)}\preceq\QFI,
  \label{eq:braunstein_caves_matrix_inequality}
\end{equation}
which is the matrix form of the Braunstein--Caves information inequality
\cite{BraunsteinCaves1994,Paris2009}.  Equality requires the measurement to
retain all relevant quantum score directions.  It is not guaranteed by the
mere existence of a separately optimal POVM for each coordinate or tangent
direction.

The physical experimental design is held fixed when taking
Eq.~\eqref{eq:povm_probability_derivative}.  If one instead considers a
family of physical measurements $M_y(\bm{\lambda})$ indexed directly by the
unknown parameter, then
\begin{equation}
  \partial_i p_y
  =
  \Tr\!\left[(\partial_i\rho)M_y\right]
  +
  \Tr\!\left[\rho\,\partial_iM_y\right].
  \label{eq:parameter_dependent_measurement}
\end{equation}
The second term belongs to a different controlled experiment and cannot be
included as if the true parameter were available to the experimenter without
cost.  A realizable adaptive scheme may approximate a locally optimal
measurement by using preliminary observations to localize the parameter; its
preliminary and final measurement stages must then be included in the complete
sampling model.  By contrast, a parameter-dependent matrix representation of
a fixed physical measurement is a passive moving-basis issue and must be
treated with the connection derived in Sec.~\ref{sec:geometry}, particularly
Sec.~\ref{subsec:moving-basis}.

For ideal flavor projection,
\begin{equation}
  M_\beta
  =
  |\nu_\beta\rangle\langle\nu_\beta|,
  \qquad
  \beta\in\{e,\mu,\tau\},
  \label{eq:flavor_povm}
\end{equation}
so that, at true-energy bin $k$,
\begin{align}
  p_{\beta k}
  &=
  |\langle\nu_\beta|\psi_k\rangle|^2,
  \label{eq:flavor_probability}
  \\
  \partial_i p_{\beta k}
  &=
  2\operatorname{Re}
  \left[
    \langle\nu_\beta|\partial_i\psi_k\rangle
    \langle\psi_k|\nu_\beta\rangle
  \right],
  \label{eq:flavor_probability_derivative}
  \\
  \CFI^{\mathrm{flavor},k}_{ij}
  &=
  \sum_\beta
  \frac{
    \partial_i p_{\beta k}\,
    \partial_j p_{\beta k}
  }{p_{\beta k}}.
  \label{eq:single_energy_flavor_cfi}
\end{align}

If the true-energy label is retained and $N_k$ is a
parameter-independent expected probe count, independent bins add:
\begin{align}
  \QFI_{\mathrm{state}}
  &=
  \sum_k N_k\QFI_k,
  \label{eq:count_weighted_state_qfi}
  \\
  \CFI_{\mathrm{flavor}}
  &=
  \sum_k N_k\CFI^{\mathrm{flavor},k}.
  \label{eq:count_weighted_flavor_cfi}
\end{align}
Equation~\eqref{eq:braunstein_caves_matrix_inequality} then implies
\begin{equation}
  \QFI_{\mathrm{state}}
  -
  \CFI_{\mathrm{flavor}}
  =
  \sum_k N_k
  \left(
    \QFI_k-\CFI^{\mathrm{flavor},k}
  \right)
  \succeq0.
  \label{eq:aggregate_qfi_cfi_inequality}
\end{equation}
If the $N_k$ depend on the estimands, their rate information must be added as
described in Sec.~\ref{sec:model}; it is not part of the conditional state
QFI in Eq.~\eqref{eq:count_weighted_state_qfi}.

\subsubsection{Zero-probability boundaries}
\label{subsubsec:zero_probability_boundaries}

Equation~\eqref{eq:fixed_measurement_cfi} is not defined by direct substitution
when an outcome probability vanishes.  Such a point must not be regularized by
silently replacing $p_y=0$ with an arbitrary pseudocount.  First derivatives
at the boundary generally do not determine the limiting Fisher information.
For example, along a local displacement $h$,
\begin{equation}
  p_y(h)=a h^2+\mathcal{O}(h^3),
  \qquad
  a>0,
  \label{eq:quadratic_boundary_probability}
\end{equation}
has $p_y(0)=\partial_h p_y(0)=0$, but
\begin{equation}
  \lim_{h\rightarrow0}
  \frac{[\partial_h p_y(h)]^2}{p_y(h)}
  =
  4a.
  \label{eq:quadratic_boundary_fisher_limit}
\end{equation}
In a multiparameter model this limit can depend on the direction of approach,
in which case the model is nonregular at that point.

We treat a zero-probability outcome as nonregular unless its structural absence
throughout a neighborhood is established; its derivatives must then vanish
throughout that neighborhood.  Otherwise, the model requires a declared
limiting path or an appropriate higher-order or nonlocal analysis.  The same
criterion applies below to zero-mean Poisson bins.

\subsection{Poissonization and the measurement-to-count identity}
\label{subsec:poissonization_identity}

Suppose the counts $n_y$ are independent Poisson variables with means
\begin{equation}
  \mu_y=Np_y,
  \label{eq:poissonized_categorical_mean}
\end{equation}
where $N$ is parameter independent.  At a regular point with $\mu_y>0$, their
Fisher information is
\begin{equation}
  \CFI^{\mathrm{Pois}}_{ij}
  =
  \sum_y
  \frac{\partial_i\mu_y\,\partial_j\mu_y}{\mu_y}
  =
  N\sum_y
  \frac{\partial_i p_y\,\partial_j p_y}{p_y}
  =
  N\CFI^{(M)}_{ij}.
  \label{eq:poisson_categorical_identity}
\end{equation}
Thus Poisson event counts reproduce the count-weighted categorical CFI exactly
when the overall expected rate is independent of the parameters of interest.

If $N=N(\bm{\lambda})$, normalization and shape information separate:
\begin{equation}
  \CFI^{\mathrm{Pois}}_{ij}
  =
  N\CFI^{(M)}_{ij}
  +
  \frac{\partial_iN\,\partial_jN}{N},
  \label{eq:rate_shape_fisher_decomposition}
\end{equation}
where the cross term vanishes because
$\sum_y\partial_i p_y=0$.  Fluxes, cross sections, exposures, and efficiencies
are independent of the oscillation parameters in the controlled benchmark
below, but their uncertainties enter as nuisance coordinates in the full
event model.

\subsection{Forward-folded reconstructed event means}
\label{subsec:forward_folding}

Let $c$ label a production and interaction channel, $k$ a true-energy bin,
$a$ a selected event class, and $b$ a reconstructed bin.  A general
forward-folded mean is
\begin{equation}
  \mu_{ab}(\bm{\lambda},\bm{\nu})
  =
  \sum_{c,k}
  R_{ab|ck}(\bm{\nu})\,
  \mathcal{X}_c\,
  \Phi_{ck}(\bm{\nu})\,
  \sigma_{ck}(\bm{\nu})\,
  \epsilon_{a|ck}(\bm{\nu})\,
  P_c(E_k;\bm{\lambda},\bm{\nu})
  +
  B_{ab}(\bm{\nu}),
  \label{eq:general_forward_folded_mean}
\end{equation}
where

\begin{itemize}
  \item $\bm{\lambda}$ denotes the oscillation parameters;
  \item $\bm{\nu}$ denotes nuisance parameters;
  \item $\mathcal{X}_c$ is the channel exposure;
  \item $\Phi_{ck}$ is the incident flux;
  \item $\sigma_{ck}$ is the interaction cross section;
  \item $\epsilon_{a|ck}$ is the selection and classification efficiency;
  \item $R_{ab|ck}$ is the true-to-reconstructed response conditional on
        selection into class $a$;
  \item $P_c$ is the oscillation probability; and
  \item $B_{ab}$ is the expected background.
\end{itemize}

When inefficiency is represented separately by $\epsilon_{a|ck}$, the
conditional response satisfies
\begin{equation}
  R_{ab|ck}\geq0,
  \qquad
  \sum_b R_{ab|ck}=1
  \label{eq:general_response_normalization}
\end{equation}
for every active $(a,c,k)$.  An equivalent convention may absorb inefficiency
into a column-substochastic response, but the two conventions must not be
applied simultaneously.

The response is a parameter-independent stochastic channel for the
physics-parameter comparison only when its entries do not vary with
$\bm{\lambda}$ at fixed nuisance coordinates.  Energy-scale, resolution, and
classification uncertainties must instead be included through $\bm{\nu}$ and
differentiated explicitly.

For the controlled benchmark, channel labels reduce to final flavor $f$, and
the event model is
\begin{equation}
  \mu_{fb}
  =
  \sum_k
  R_{bk}\,
  N_k\,
  \epsilon_{fk}\,
  p_{fk}
  +
  B_{fb}.
  \label{eq:controlled_forward_folded_mean}
\end{equation}
For fixed detector inputs, its oscillation-parameter Jacobian is
\begin{equation}
  J^{(\lambda)}_{fb,i}
  \equiv
  \partial_i\mu_{fb}
  =
  \sum_k
  R_{bk}\,
  N_k\,
  \epsilon_{fk}\,
  \partial_i p_{fk}.
  \label{eq:forward_folded_physics_jacobian}
\end{equation}
A nuisance derivative must act on every factor in
Eq.~\eqref{eq:general_forward_folded_mean} that the nuisance changes.

\subsection{Poisson likelihood and nuisance profiling}
\label{subsec:poisson_likelihood_profiling}

For independent reconstructed bins $r=(a,b)$, the constrained likelihood is
\begin{equation}
  \mathcal{L}(\bm{\lambda},\bm{\nu})
  =
  \prod_r
  \frac{
    \mu_r(\bm{\lambda},\bm{\nu})^{n_r}
    e^{-\mu_r(\bm{\lambda},\bm{\nu})}
  }{n_r!}
  \,
  \exp\!\left[
    -\frac{1}{2}
    (\bm{\nu}-\bm{\nu}_0)^{\mathsf T}
    \Pi
    (\bm{\nu}-\bm{\nu}_0)
  \right],
  \label{eq:constrained_poisson_likelihood}
\end{equation}
where $\Pi\succeq0$ is the Fisher information of the Gaussian auxiliary
constraints.  In a frequentist ensemble these constraints represent auxiliary
data and need not be interpreted as immutable Bayesian priors.

The Poisson deviance relative to the saturated model is
\begin{align}
  D(\bm{\lambda},\bm{\nu})
  &=
  2\sum_r
  \left[
    \mu_r-n_r
    +
    n_r\ln\!\left(\frac{n_r}{\mu_r}\right)
  \right]
  \nonumber\\
  &\quad+
  (\bm{\nu}-\bm{\nu}_0)^{\mathsf T}
  \Pi
  (\bm{\nu}-\bm{\nu}_0),
  \label{eq:poisson_deviance}
\end{align}
with $n_r\ln(n_r/\mu_r)=0$ when $n_r=0$.  If $n_r>0$ and
$\mu_r=0$, the deviance is infinite.

Let
\begin{equation}
  \bm{\vartheta}
  =
  (\bm{\lambda},\bm{\nu}).
  \label{eq:joint_parameter_vector}
\end{equation}
At a regular point with $\mu_r>0$, the joint expected Fisher matrix is
\begin{equation}
  \CFI^{\mathrm{full}}_{AB}
  =
  \sum_r
  \frac{
    \partial_A\mu_r\,
    \partial_B\mu_r
  }{\mu_r}
  +
  \Pi_{AB},
  \label{eq:full_poisson_fisher}
\end{equation}
where $\Pi_{AB}$ is zero outside the constrained nuisance block.  Partition
this matrix as
\begin{equation}
  \CFI^{\mathrm{full}}
  =
  \begin{pmatrix}
    \CFI_{\lambda\lambda}
    &
    \CFI_{\lambda\nu}
    \\
    \CFI_{\nu\lambda}
    &
    \CFI_{\nu\nu}
  \end{pmatrix}.
  \label{eq:block_fisher_matrix}
\end{equation}
If $\CFI_{\nu\nu}$ is invertible, the local information remaining after
nuisance profiling is the Schur complement
\begin{equation}
  \CFI_{\mathrm{prof}}
  =
  \CFI_{\lambda\lambda}
  -
  \CFI_{\lambda\nu}
  \CFI_{\nu\nu}^{-1}
  \CFI_{\nu\lambda}.
  \label{eq:profiled_fisher_schur}
\end{equation}
Therefore,
\begin{equation}
  \CFI_{\lambda\lambda}
  -
  \CFI_{\mathrm{prof}}
  =
  \CFI_{\lambda\nu}
  \CFI_{\nu\nu}^{-1}
  \CFI_{\nu\lambda}
  \succeq0.
  \label{eq:profiling_information_loss}
\end{equation}

If $\CFI_{\nu\nu}$ is singular, the generalized Schur complement is
\begin{equation}
  \CFI_{\mathrm{prof}}
  =
  \CFI_{\lambda\lambda}
  -
  \CFI_{\lambda\nu}
  \CFI_{\nu\nu}^{+}
  \CFI_{\nu\lambda},
  \label{eq:generalized_profiled_fisher_schur}
\end{equation}
provided that
\begin{equation}
  \ran(\CFI_{\nu\lambda})
  \subseteq
  \ran(\CFI_{\nu\nu}).
  \label{eq:profiled_fisher_range_condition}
\end{equation}
For a positive-semidefinite complete block Fisher matrix, this range condition
holds automatically.  The pseudoinverse profiles only the identifiable
nuisance subspace; it does not repair an unidentifiable joint model or make
nonestimable physics functions estimable.  The controlled benchmark has a
full-rank nuisance block after its auxiliary constraints, so
Eq.~\eqref{eq:profiled_fisher_schur} uses an ordinary inverse.

For the Asimov data set,
\begin{equation}
  n_r^{\mathrm A}
  =
  \mu_r(\bm{\lambda}_0,\bm{\nu}_0),
  \label{eq:asimov_counts}
\end{equation}
define the exactly profiled deviance
\begin{equation}
  D_{\mathrm{prof}}(\bm{\lambda})
  =
  \min_{\bm{\nu}}D(\bm{\lambda},\bm{\nu}).
  \label{eq:exact_profiled_deviance}
\end{equation}
At the truth and under the usual local regularity conditions,
\begin{equation}
  \left.
  \frac{\partial^2D_{\mathrm{prof}}}
       {\partial\lambda_i\partial\lambda_j}
  \right|_{\bm{\lambda}_0}
  =
  2\left(\CFI_{\mathrm{prof}}\right)_{ij}.
  \label{eq:profiled_deviance_schur_curvature}
\end{equation}
This establishes the local relation between the Asimov/profile-likelihood
curvature and the Schur-complement Fisher matrix
\cite{CowanCranmerGrossVitells2011}.  It does not imply that the likelihood is
globally quadratic or that automatic Wilks thresholds have correct coverage
in the presence of periodicity, physical boundaries, ordering changes,
weak identification, or disconnected degeneracies
\cite{Wilks1938,FeldmanCousins1998,NOvAProfiledFC2022}.

\subsection{Why each detector operation loses information}
\label{subsec:detector_data_processing}

At a regular point, let $K(z|y)$ be a parameter-independent Markov kernel and
let the output distribution be
\begin{equation}
  q_z
  =
  \sum_y K(z|y)p_y.
  \label{eq:classical_markov_channel}
\end{equation}
Define the input score
\begin{equation}
  \bm{s}_y
  =
  \nabla_{\bm{\lambda}}\ln p_y.
  \label{eq:input_score}
\end{equation}
The output score is the conditional mean
\begin{equation}
  \nabla_{\bm{\lambda}}\ln q_z
  =
  \mathbb{E}[\bm{s}_Y|Z=z].
  \label{eq:output_score_conditional_mean}
\end{equation}
The law of total covariance gives
\begin{align}
  \CFI_Y-\CFI_Z
  &=
  \mathbb{E}
  \left[
    \operatorname{Cov}(\bm{s}_Y|Z)
  \right]
  \succeq0.
  \label{eq:classical_fisher_data_processing}
\end{align}
Energy smearing, event merging, and deletion of a classical label are all
instances of Eq.~\eqref{eq:classical_fisher_data_processing} when their
kernels are independent of the physics estimands.

A known parameter-independent background also reduces Fisher information.
For signal means $s_r>0$ and backgrounds $b_r\geq0$,
\begin{align}
  \bm{u}^{\mathsf T}
  \left(
    \CFI_s-\CFI_{s+b}
  \right)
  \bm{u}
  &=
  \sum_r
  \left(
    \bm{u}^{\mathsf T}\nabla s_r
  \right)^2
  \left(
    \frac{1}{s_r}
    -
    \frac{1}{s_r+b_r}
  \right)
  \geq0.
  \label{eq:background_fisher_loss}
\end{align}
A zero-signal boundary requires the limiting treatment described in
Sec.~\ref{subsubsec:zero_probability_boundaries}.  Unobserved
parameter-independent efficiency thinning is likewise an information-losing
channel.  Finally, Eqs.~\eqref{eq:profiling_information_loss} and
\eqref{eq:generalized_profiled_fisher_schur} show that nuisance profiling
cannot increase local physics information.

Combining these facts with the fixed-measurement quantum information
inequality yields Eq.~\eqref{eq:operational_information_hierarchy}.  The
hierarchy assumes a fixed detector channel for the physics-parameter
comparison.  Detector uncertainties are included through nuisance derivatives
in the joint likelihood and are not ignored.

\subsection{Coordinate-invariant retention diagnostics}
\label{subsec:retention_diagnostics}

The diagonal entries of two information matrices cannot be compared
meaningfully when parameters have heterogeneous units and strong correlations.
We instead define the generalized information eigenvalues $r_\alpha$ of an
operational layer $\CFI$ relative to a positive-definite reference
$\QFI_{\mathrm{state}}$:
\begin{equation}
  \CFI\,\bm{v}_\alpha
  =
  r_\alpha
  \QFI_{\mathrm{state}}\bm{v}_\alpha.
  \label{eq:generalized_retention_eigenproblem}
\end{equation}
If $0\prec\CFI\preceq\QFI_{\mathrm{state}}$, then
\begin{equation}
  0<r_\alpha\leq1.
  \label{eq:retention_eigenvalue_range}
\end{equation}
The unordered set $\{r_\alpha\}$ is invariant under nonsingular
reparameterizations.

If the reference QFIM is singular, the comparison is first restricted to
$\ran(\QFI_{\mathrm{state}})$, as prescribed in Sec.~\ref{sec:model}.  A
direction identifiable in the reference but completely lost by the
operational layer appears as $r_\alpha=0$ and must not be removed by
restricting to $\ran(\CFI)$.

Three useful summaries are
\begin{align}
  r_{\min}
  &=
  \min_\alpha r_\alpha,
  \label{eq:minimum_information_retention}
  \\
  g
  &=
  \left(\prod_{\alpha=1}^{p}r_\alpha\right)^{1/p},
  \label{eq:geometric_information_retention}
  \\
  C_Q(\CFI)
  &=
  \operatorname{tr}
  \left(
    \QFI_{\mathrm{state}}\CFI^{-1}
  \right)
  =
  \sum_{\alpha=1}^{p}\frac{1}{r_\alpha}.
  \label{eq:qfi_metric_classical_cost}
\end{align}
The inverse in Eq.~\eqref{eq:qfi_metric_classical_cost} is taken on the
reference-identifiable space and exists only when $\CFI$ retains every such
direction.  The worst generalized standard-error inflation is
\begin{equation}
  \mathcal{I}_{\mathrm{worst}}
  =
  r_{\min}^{-1/2},
  \label{eq:worst_standard_error_inflation}
\end{equation}
while the geometric standard-error inflation is $g^{-1/2}$.  If an
operational matrix loses a reference-identifiable direction, at least one
$r_\alpha$ vanishes and the corresponding inflation and $C_Q$ are infinite.

Generalized eigenvectors need not agree between operational layers.  Sorted
entries from different rows must therefore not be interpreted as the
evolution of a fixed named parameter combination.  To localize each loss, we
also compute generalized eigenvalues of every layer relative to the
immediately preceding layer.

\subsection{Controlled state-to-event benchmark}
\label{subsec:controlled_operational_benchmark}

We validate the full chain using a controlled analytic experiment model with

\begin{itemize}
  \item 19 true-energy bins from $0.5$ to $5.0~\mathrm{GeV}$;
  \item $2.0\times10^4$ parameter-independent unoscillated probes;
  \item $L=1300~\mathrm{km}$;
  \item constant density
        $2.848~\mathrm{g\,cm^{-3}}$ and $Y_e=0.5$;
  \item an initial $\nu_\mu$ state;
  \item the oscillation benchmark of
        Eq.~\eqref{eq:identifiability-oscillation-benchmark};
  \item a column-stochastic Gaussian response with
        $\sigma_E/\mathrm{GeV}=0.12+0.10\sqrt{E/\mathrm{GeV}}$;
  \item flavor- and energy-dependent efficiencies;
  \item fixed reconstructed backgrounds; and
  \item six Gaussian-constrained nuisance coordinates.
\end{itemize}

The nuisance modes are a $5\%$ signal normalization, a $10\%$ RMS flux tilt,
a $2\%$ energy-scale shift, and independent $10\%$, $5\%$, and $20\%$
background normalizations for the electron-, muon-, and tau-like samples.
Each nuisance coordinate is expressed in units of its nominal
one-standard-deviation constraint, so the auxiliary-constraint Fisher matrix
is the identity.  For the local Fisher and curvature validation, these modes
and the physics derivatives enter the affine tangent-rate model
$\mu_r=\mu_r^0+\sum_iJ^{(\lambda)}_{ri}\delta\lambda_i+
\sum_a A_{ar}\nu_a$, with $J^{(\lambda)}$ and $A$ fixed at the benchmark.
This surrogate is sufficient for testing the local Schur-complement identity;
it is not used as a nonlinear finite-displacement detector model.

The ideal smeared event total is $20000$.  After efficiencies and backgrounds,
the expected total is $10745.55$, of which $2100$ events are background.
After the coordinate scaling in Eq.~\eqref{eq:default-scales}, all five
information layers have numerical rank six at relative eigenvalue tolerance
$10^{-10}$.

\begin{table*}[t]
  \centering
  \footnotesize
  \setlength{\tabcolsep}{4pt}
  \caption{
    Coordinate-invariant information retention in the controlled benchmark.
    Generalized eigenvalues are measured relative to the count-weighted state
    QFIM.  Numerical ranks use relative tolerance $10^{-10}$ after the
    coordinate scaling in Eq.~\eqref{eq:default-scales}.  The metric cost is
    defined in Eq.~\eqref{eq:qfi_metric_classical_cost}.  Large inflation
    factors diagnose weak operational directions despite numerical rank six.
  }
  \label{tab:operational_information_retention}
  \begin{tabular}{lrrrrrr}
    \toprule
    Layer
    & Rank
    & $r_{\min}$
    & $r_{\max}$
    & $g$
    & $r_{\min}^{-1/2}$
    & $C_Q$
    \\
    \midrule
    State QFI
      & 6
      & $1.000000$
      & $1.000000$
      & $1.000000$
      & $1.00$
      & $6.00$
      \\
    Resolved flavor CFI
      & 6
      & $6.5895\times10^{-6}$
      & $0.998651$
      & $0.077839$
      & $389.6$
      & $1.51770\times10^5$
      \\
    Energy-smeared Poisson FI
      & 6
      & $1.7508\times10^{-6}$
      & $0.995297$
      & $0.038735$
      & $755.7$
      & $5.71184\times10^5$
      \\
    Efficiency/background-degraded FI
      & 6
      & $7.4635\times10^{-7}$
      & $0.460340$
      & $0.019102$
      & $1157.5$
      & $1.33992\times10^6$
      \\
    Nuisance-profiled FI
      & 6
      & $7.3413\times10^{-7}$
      & $0.253254$
      & $0.011387$
      & $1167.1$
      & $1.36224\times10^6$
      \\
    \bottomrule
  \end{tabular}
\end{table*}

Flavor projection is the largest initial contraction in this controlled
example.  Its determinant information fraction relative to the state QFIM is
\begin{equation}
  \frac{\det\CFI_{\mathrm{flavor}}}
       {\det\QFI_{\mathrm{state}}}
  =
  2.2243\times10^{-7}.
  \label{eq:flavor_determinant_retention}
\end{equation}
After all detector operations and profiling,
\begin{equation}
  \frac{\det\CFI_{\mathrm{prof}}}
       {\det\QFI_{\mathrm{state}}}
  =
  2.1802\times10^{-12}.
  \label{eq:profiled_determinant_retention}
\end{equation}
\begin{table*}[t]
  \centering
  \footnotesize
  \renewcommand{\arraystretch}{1.16}
  \setlength{\tabcolsep}{5pt}
  \caption{
    Incremental information contractions between adjacent operational layers.
    Panel (a) reports coordinate-invariant generalized quantities using the
    earlier layer as the reference.  Panel (b) reports the minimum eigenvalue
    of the scaled matrix difference
    $\widetilde F_{\mathrm{before}}
      -\widetilde F_{\mathrm{after}}$,
    where the scaling is that of Eq.~\eqref{eq:default-scales}.
    The magnitude of this PSD-consistency margin is scale dependent, while
    its nonnegative sign checks the corresponding analytically established
    Loewner inequality.
  }
  \label{tab:incremental_operational_losses}

  \begin{tabular}{@{}lrrrr@{}}
    \toprule
    \multicolumn{5}{c}{
      \textbf{(a) Coordinate-invariant incremental contractions}
    }
    \\
    \midrule
    Transition
    & $r_{\min}$
    & $r_{\max}$
    & $g$
    & $\det F_{\mathrm{after}}/
       \det F_{\mathrm{before}}$
    \\
    \midrule
    State QFI $\rightarrow$ flavor CFI
      & $6.5895\times10^{-6}$
      & $0.998651$
      & $0.077839$
      & $2.2243\times10^{-7}$
    \\
    Flavor CFI $\rightarrow$ energy smearing
      & $0.214360$
      & $0.998094$
      & $0.497630$
      & $0.0151859$
    \\
    Smearing $\rightarrow$ efficiencies/backgrounds
      & $0.356622$
      & $0.675563$
      & $0.493131$
      & $0.0143804$
    \\
    Fixed detector $\rightarrow$ nuisance profile
      & $0.212180$
      & $0.999582$
      & $0.596139$
      & $0.0448835$
    \\
    \bottomrule
  \end{tabular}

  \vspace{0.9em}

  \begin{tabular}{@{}lr@{}}
    \toprule
    \multicolumn{2}{c}{
      \textbf{(b) Coordinate-scaled PSD-consistency margins}
    }
    \\
    \midrule
    Transition
    &
    $\lambda_{\min}\!\left(
      \widetilde F_{\mathrm{before}}
      -
      \widetilde F_{\mathrm{after}}
    \right)$
    \\
    \midrule
    State QFI $\rightarrow$ flavor CFI
      & $7.03079$
    \\
    Flavor CFI $\rightarrow$ energy smearing
      & $6.06379\times10^{-3}$
    \\
    Smearing $\rightarrow$ efficiencies/backgrounds
      & $1.49448\times10^{-3}$
    \\
    Fixed detector $\rightarrow$ nuisance profile
      & $1.70833\times10^{-6}$
    \\
    \bottomrule
  \end{tabular}
\end{table*}

Let $\operatorname{spec}_{\QFI}(F)$ denote the ordered generalized eigenvalue
spectrum of $F$ relative to $\QFI_{\mathrm{state}}$.  The absolute generalized
spectra are
\begin{align}
  \operatorname{spec}_{\QFI}
  (\CFI_{\mathrm{flavor}})
  &=
  \bigl(
    6.5895\times10^{-6},
    0.1745,
    0.2760,
    0.8115,
    0.8647,
    0.9987
  \bigr),
  \label{eq:flavor_relative_spectrum}
  \\
  \operatorname{spec}_{\QFI}
  (\CFI_{\mathrm{smear}})
  &=
  \bigl(
    1.7508\times10^{-6},
    0.05930,
    0.1228,
    0.4431,
    0.6009,
    0.9953
  \bigr),
  \label{eq:smeared_relative_spectrum}
  \\
  \operatorname{spec}_{\QFI}
  (\CFI_{\mathrm{det}})
  &=
  \bigl(
    7.4635\times10^{-7},
    0.02378,
    0.07929,
    0.2428,
    0.3089,
    0.4603
  \bigr),
  \label{eq:detector_relative_spectrum}
  \\
  \operatorname{spec}_{\QFI}
  (\CFI_{\mathrm{prof}})
  &=
  \bigl(
    7.3413\times10^{-7},
    0.02320,
    0.05224,
    0.06655,
    0.1454,
    0.2533
  \bigr).
  \label{eq:profiled_relative_spectrum}
\end{align}
These spectra show that numerical rank six coexists with an operational
direction more than six orders of magnitude weaker than its state-QFI
normalization.

\subsection{Numerical and likelihood validation}
\label{subsec:operational_validation}

All probabilities and event means used above lie strictly inside their regular
domains:
\begin{align}
  \min_{f,k}p_{fk}
  &=
  9.68205\times10^{-3},
  \nonumber\\
  \min_{f,k}N_kp_{fk}
  &=
  0.814793,
  \nonumber\\
  \min_{f,b}\mu_{fb}^{\mathrm{det}}
  &=
  2.77152.
  \label{eq:operational_minimum_probability_means}
\end{align}
Thus none of the reported numbers relies on a zero-probability or zero-mean
prescription.

The count-weighted flavor CFI and the independently Poissonized true-bin
flavor information agree at relative Frobenius error
\begin{equation}
  \frac{
    \left\|
      \CFI_{\mathrm{flavor}}
      -
      \CFI_{\mathrm{Poissonized}}
    \right\|_F
  }{
    \left\|\CFI_{\mathrm{flavor}}\right\|_F
  }
  =
  3.46\times10^{-16}.
  \label{eq:poissonization_validation_error}
\end{equation}

We also profile the Poisson deviance of this affine tangent-rate model exactly
over its nuisance coordinates along all six generalized eigendirections of
$\CFI_{\mathrm{prof}}$ relative to $\QFI_{\mathrm{state}}$.  Each direction
is normalized so that
\begin{equation}
  \bm{d}_\alpha^{\mathsf T}
  \CFI_{\mathrm{prof}}
  \bm{d}_\alpha=1.
  \label{eq:profile_curvature_direction_normalization}
\end{equation}
The numerical quantity
\begin{equation}
  \widehat{\mathcal{I}}_\alpha
  =
  \frac{1}{2}
  \frac{
    D_{\mathrm{prof}}(h\bm{d}_\alpha)
    -2D_{\mathrm{prof}}(0)
    +D_{\mathrm{prof}}(-h\bm{d}_\alpha)
  }{h^2}
  \label{eq:profiled_curvature_finite_difference}
\end{equation}
agrees with unity in every direction, with maximum relative error
\begin{equation}
  \max_\alpha
  |\widehat{\mathcal{I}}_\alpha-1|
  =
  3.82\times10^{-7}.
  \label{eq:profiled_curvature_validation_error}
\end{equation}
The maximum nuisance-score infinity norm after profiling is
$5.81\times10^{-8}$.

Finally, we apply a fixed invertible nonorthogonal reparameterization to all
six physics coordinates and an independent invertible transformation to the
six nuisance coordinates of this same controlled model.  The source spectrum,
propagation, response, efficiencies, backgrounds, expected event means, and
auxiliary likelihood are held fixed; only their coordinate representations
and the auxiliary-constraint tensor are transformed consistently.  This is
not the later public-DUNE calculation, which contains nine nuisance
parameters.

The same-model coordinate test gives
\begin{align}
  \max_\alpha
  \left|
    r_\alpha^{\mathrm{new}}
    -
    r_\alpha^{\mathrm{old}}
  \right|
  &<
  1.0\times10^{-14},
  \nonumber\\
  \frac{
    \left\|
      \CFI_{\mathrm{prof}}^{\mathrm{new}}
      -
      J^{\mathsf T}
      \CFI_{\mathrm{prof}}^{\mathrm{old}}
      J
    \right\|_F
  }{
    \left\|
      J^{\mathsf T}
      \CFI_{\mathrm{prof}}^{\mathrm{old}}
      J
    \right\|_F
  }
  &=
  2.41\times10^{-15}.
  \label{eq:profile_coordinate_covariance_validation}
\end{align}
Nuisance Schur complementation therefore commutes numerically with invertible
coordinate pullback in this benchmark, including the consistently transformed
auxiliary-constraint information.

\subsection{Interpretation and limitations}
\label{subsec:operational_interpretation}

The controlled analysis establishes four points.

First, the state QFIM, flavor CFI, and event Fisher matrix answer different
questions.  The state QFIM supplies the SLD information upper bound for the
declared state family; it need not be jointly attainable by one
multiparameter measurement.  The flavor CFI assumes a fixed ideal flavor
measurement, while the event Fisher matrix describes the reconstructed-count
likelihood.

Second, numerical full rank does not imply useful simultaneous precision.
All five layers in Table~\ref{tab:operational_information_retention} have
numerical rank six at relative tolerance $10^{-10}$, but the weakest profiled
direction retains only $7.34\times10^{-7}$ of its state-QFI normalization.
An inverse matrix without its relative spectrum therefore conceals the
dominant operational limitation.

Third, detector losses cannot be represented by a single universal efficiency
factor.  Flavor projection, smearing, backgrounds, selection efficiencies,
and nuisance profiling act anisotropically and rotate the generalized
information directions.  Their effects must be propagated through the
complete event Jacobian.

Fourth, the Fisher and Asimov results remain local.  They describe the
quadratic curvature at the declared parameter point and do not establish
global confidence regions, ordering discrimination, octant coverage, or
periodic-$\dcp$ coverage.  Those questions require the nonquadratic profiled
likelihood and pseudo-experiment calibration.

The numerical inputs in this section are deliberately analytic validation
functions.  They are not an official DUNE configuration and must not be
quoted as experimental sensitivity.  Their purpose is to prove and validate
the state-to-event information pipeline.  Section~\ref{sec:reassessment} next
applies this framework to the fixed-energy calculation of
Ref.~\cite{HuangEtAl2026}, after which the version-locked public-DUNE implementation is analyzed
separately.

\section{Reassessment and numerical validation}
\label{sec:reassessment}

This section reassesses the fixed-energy calculation of
Ref.~\cite{HuangEtAl2026} and records the numerical tests used throughout this
work.  The purpose is not to dispute the entries of the published QFIM: those
entries are reproducible.  Rather, we separate four logically distinct
questions that cannot be answered by the entries alone: whether the local
model is identifiable, whether a matrix bound is jointly attainable, whether
two state descriptions are related passively or actively, and whether the
state-level information is accessible through a specified detector
likelihood.

\subsection{Comparison protocol and coordinate scaling}
\label{subsec:reassessment-protocol}

We use the parameter order
\begin{equation}
  \bm{\lambda}
  =
  \bigl(
    \theta_{12},
    \theta_{13},
    \theta_{23},
    \dcp,
    \Delta m^2_{21},
    \Delta m^2_{31}
  \bigr)^{\mathsf T}
\end{equation}
and the same oscillation point and phase convention as
Ref.~\cite{HuangEtAl2026}.  What that reference describes as a
DUNE-labelled benchmark is, at the state level, a pure muon-flavor state
evaluated in vacuum at $L/E=520\,\mathrm{km/GeV}$.
It is therefore a monochromatic vacuum benchmark rather than a DUNE
reconstructed-event model.

Because the coordinates carry heterogeneous units, rank and conditioning are
diagnosed after the nonsingular pullback
\begin{equation}
  \widetilde{\QFI}
  =
  S^{\mathsf T}\QFI S,
  \qquad
  S
  =
  \operatorname{diag}
  \bigl(
    1,1,1,1,
    10^{-4}\,\mathrm{eV}^{2},
    10^{-3}\,\mathrm{eV}^{2}
  \bigr).
  \label{eq:reassessment-scaling}
\end{equation}
The scaling cannot change the exact mathematical rank.  It prevents a
floating-point tolerance from being determined by the arbitrary numerical
use of $\mathrm{eV}^{-4}$ for the two mass-splitting entries.  An eigenmode is
counted as numerically identifiable only when its scaled eigenvalue exceeds
\begin{equation}
  \tau_{\mathrm{rank}}
  =
  10^{-10}
  \max\!\left\{1,\|\widetilde{\QFI}\|_2\right\}
  =
  1.56\times10^{-9}.
  \label{eq:reassessment-rank-tolerance}
\end{equation}
Negative eigenvalues whose magnitudes lie below the corresponding PSD
tolerance are treated as floating-point roundoff, not as physical information
directions.  All residuals below are evaluated before rounding; matrices
reconstructed from published decimal entries are tested separately to
quantify rounding sensitivity.

\subsection{Reproduction of the fixed-energy QFIM}
\label{subsec:reassessment-reproduction}

Writing
\begin{equation}
  \QFI
  =
  \begin{pmatrix}
    A & B\\
    B^{\mathsf T} & C
  \end{pmatrix},
\end{equation}
our independent implementation gives
\begin{align}
  A={}&
  \begin{pmatrix}
    0.0141207& 0.0785022& 0.0127099& 0.0340327\\
    0.0785022& 7.51844&-1.31143& 0.00628930\\
    0.0127099&-1.31143&15.3671&-0.00608474\\
    0.0340327& 0.00628930&-0.00608474&0.156280
  \end{pmatrix},
  \nonumber\\[2mm]
  B={}&
  \begin{pmatrix}
    -16.0605&102.646\\
    1385.76&-3.48658\\
    252.751&-31.4911\\
    226.648&9.89083
  \end{pmatrix}
  \mathrm{eV}^{-2},
  \qquad
  \frac{C}{10^{6}\,\mathrm{eV}^{-4}}
  =
  \begin{pmatrix}
    1.70036&-1.36648\\
    -1.36648&1.72495
  \end{pmatrix}.
  \label{eq:reassessment-reproduced-blocks}
\end{align}
The relative Frobenius distances from the blocks printed in
Eqs.~(3.18), (3.21), and (3.22) of Ref.~\cite{HuangEtAl2026} are
$1.91\times10^{-3}$, $2.98\times10^{-3}$, and
$2.27\times10^{-3}$, respectively, consistent with their displayed
precision.  The diagonal entries in its Eq.~(3.15) are likewise reproduced.

The interpretation changes once the structural tangent geometry is retained.
The scaled spectrum is
\begin{equation}
  \operatorname{eig}(\widetilde{\QFI})
  =
  \bigl(
    -1.74\times10^{-15},
    \;5.69\times10^{-17},
    \;0.166828,
    \;1.74185,
    \;7.30872,
    \;15.5805
  \bigr).
  \label{eq:reassessment-exact-spectrum}
\end{equation}
The matrix therefore has numerical rank four and numerical nullity two at the
declared scaling and tolerance.  This is not an empirical near-singularity:
the pure-qutrit theorem of Sec.~\ref{sec:identifiability} requires the
structural rank to satisfy
\begin{equation}
  \rank\QFI\leq2(3-1)=4.
\end{equation}
The four positive modes show that the bound is saturated at the benchmark.

If $T$ denotes the real horizontal tangent map, direct construction gives
\begin{equation}
  \QFI=T^{\mathsf T}T,
  \qquad
  \rank T=4,
  \qquad
  \frac{
    \|\QFI-T^{\mathsf T}T\|_{\mathrm F}
  }{
    \|\QFI\|_{\mathrm F}
  }
  =
  2.25\times10^{-16}.
  \label{eq:reassessment-tangent-factorization}
\end{equation}
Thus the two null modes are required by the physical dimension of the
pure-qutrit state manifold rather than generated by numerical
ill-conditioning.

The four angle tangents already span the complete real qutrit tangent space at
the benchmark, and the principal angle block $A$ is nonsingular.  Positivity
and the rank-four constraint then imply that the mass-sector Schur complement
vanishes in the exact model.  Numerically,
\begin{equation}
  \left\|
    S_m^{\mathsf T}
    \bigl(
      C-B^{\mathsf T}A^{-1}B
    \bigr)
    S_m
  \right\|_{\mathrm F}
  =
  7.06\times10^{-16},
  \qquad
  S_m
  =
  \operatorname{diag}(10^{-4},10^{-3})\,\mathrm{eV}^{2}.
  \label{eq:reassessment-schur-zero}
\end{equation}

The displayed-precision matrix provides a separate rounding stress test.
Reconstructing the full matrix only from the decimals in
Ref.~\cite{HuangEtAl2026} yields
\begin{equation}
  \operatorname{eig}(\widetilde{\QFI}_{\mathrm{print}})
  =
  \bigl(
    -1.568\times10^{-4},
    -4.026\times10^{-5},
    0.166561,
    1.73704,
    7.31155,
    15.6121
  \bigr).
  \label{eq:reassessment-rounded-spectrum}
\end{equation}
As an ordinary numerical matrix it is nonsingular, but it is indefinite:
rounding has displaced the two structural null modes into two negative modes.
It is therefore not a valid positive-semidefinite rank-six QFIM.  Its ordinary
inverse has six negative diagonal entries, beginning with
$(\QFI_{\mathrm{print}}^{-1})_{11}=-1.82\times10^{4}$ in native coordinates.
Neither this rounded inverse nor an unqualified
Moore--Penrose inverse is a six-coordinate covariance bound.  A pseudoinverse
is meaningful only after the estimable tangent combinations have been
declared and the desired functions have been verified to annihilate the exact
kernel.

\subsection{Interpretation of the fixed-energy result}
\label{subsec:reassessment-interpretation}

The preceding reproduction shows that the published fixed-energy QFIM entries
are numerically reproducible, while their interpretation changes once
structural identifiability, multiparameter attainability, representation
covariance, and detector accessibility are treated explicitly.  Two points
merit additional emphasis.

First, for a regular identifiable model with nonsingular QFIM, $N$
independent preparations, and a locally unbiased estimator, the SLD matrix
Cramér--Rao relation is the Loewner-order inequality
\begin{equation}
  \operatorname{Cov}(\widehat{\bm\lambda})
  -
  \frac{1}{N}\QFI^{-1}
  \succeq0.
  \label{eq:reassessment-loewner}
\end{equation}
It is not an entrywise inequality, and in a multiparameter model the matrix
lower bound need not be jointly attainable.  In particular, an off-diagonal
entry of $\QFI^{-1}$ does not determine the sign of the covariance produced by
an arbitrary estimator.  Moreover, $1/\QFI_{ii}$ is the one-parameter bound
for the submodel in which all complementary coordinates are held fixed,
whereas $(\QFI^{-1})_{ii}$ is the corresponding diagonal entry of the
multiparameter SLD matrix bound when all declared coordinates are estimated.
Inverting $A$ and $C$ separately is therefore a pair of fixed-complement
submodel calculations, not a six-parameter analysis.

Even within the identifiable four-angle submodel, invertibility is not
attainability.  For the dimensionless metric weight $W=A$, the SLD cost is
four and the exact Holevo cost is eight.  The mass-only two-coordinate
submodel is weakly compatible to numerical precision, but that does not make
the direct sum $A\oplus C$ the QFIM of the original six-coordinate family
\cite{Holevo2011,Matsumoto2002}.

Second, a parameter-dependent change of coefficients must be classified as
passive or active.  Using the moving-frame convention of
Sec.~\ref{subsec:moving-basis}, a passive basis $B(\bm\lambda)$ satisfies
\begin{equation}
  \bm c
  =
  B^{\dagger}\lvert\psi\rangle,
  \qquad
  \Gamma_i
  =
  B^{\dagger}\partial_iB,
  \qquad
  \nabla_i\bm c
  =
  \partial_i\bm c+\Gamma_i\bm c
  =
  B^{\dagger}\partial_i\lvert\psi\rangle.
  \label{eq:reassessment-moving-basis}
\end{equation}
The covariant derivative leaves the quantum geometric tensor invariant; at
the comparison point its relative residual is $2.14\times10^{-16}$.  Omitting the
connection changes the tensor by $0.735345$ because the ordinary coefficient
derivative then describes the distinct actively encoded family
$B^\dagger(\bm\lambda)|\psi(\bm\lambda)\rangle$ in a fixed representation.
A state representation also does not determine a measurement: the flavor
POVM must be specified and transformed independently as described in
Sec.~\ref{subsec:moving-basis}.

For CP conjugation, define the parameter reflection and its Jacobian by
\begin{equation}
  f(\bm\lambda)
  =
  \bigl(
    \theta_{12},
    \theta_{13},
    \theta_{23},
    -\dcp,
    \Delta m^2_{21},
    \Delta m^2_{31}
  \bigr)^{\mathsf T},
  \qquad
  J_f
  =
  \operatorname{diag}(1,1,1,-1,1,1).
  \label{eq:reassessment-cp-map}
\end{equation}
The QFIM and curvature are covariant two-index tensors and obey
\begin{equation}
  \overline{\QFI}(\bm\lambda;V_{\mathrm{CC}})
  =
  J_f^{\mathsf T}
  \QFI\!\left(f(\bm\lambda);-V_{\mathrm{CC}}\right)
  J_f,
  \qquad
  \overline{\mathcal{D}}(\bm\lambda;V_{\mathrm{CC}})
  =
  J_f^{\mathsf T}
  \mathcal{D}\!\left(f(\bm\lambda);-V_{\mathrm{CC}}\right)
  J_f.
  \label{eq:reassessment-cp-pullback}
\end{equation}
The replacement $\dcp\to-\dcp$ is sufficient for comparing diagonal scalar
entries after evaluation at the reflected point, but it is not sufficient for
a mixed tensor component with exactly one $\dcp$ index.  In vacuum, for
example, direct antineutrino propagation gives
$\overline{\QFI}_{\theta_{23}\dcp}=+0.0135636$, which is the negative of the
unpulled neutrino component evaluated at
$-\dcp$.  In matter, omitting
$V_{\mathrm{CC}}\to-V_{\mathrm{CC}}$ gives scaled relative errors
$0.422527$ for $\QFI$ and $0.213949$ for $\mathcal{D}$; the complete
pullbacks agree with direct propagation at binary floating-point precision.

The spectral-rank-restoration statement can be made explicit within the same
vacuum model.  In this controlled comparison, every retained conditional QFIM
enters the aggregate with equal unit weight, so the reported matrices are sums
rather than fixed-total-exposure averages.  Summing the two retained labels at
$L/E=400,\,520\,\mathrm{km/GeV}$
gives numerical rank six at the tolerance of
Eq.~\eqref{eq:reassessment-rank-tolerance}, with smallest scaled eigenvalue
$4.00\times10^{-5}$.  Summing five labels at
$L/E=(300,400,520,700,900)\,\mathrm{km/GeV}$
increases the smallest scaled eigenvalue to $2.38\times10^{-3}$ and reduces
the condition number on the full support from $7.04\times10^5$ to
$1.77\times10^4$.  A common normalization of an aggregate rescales all of its
eigenvalues but leaves its rank and condition number unchanged.  The observed
restoration is the kernel-intersection mechanism derived in
Sec.~\ref{subsec:kernel-intersection}, not information created by
off-diagonal matrix entries.

Finally, multiplying the pure-state QFI at one representative energy by an
order-of-magnitude event count does not define a DUNE likelihood.  It assumes
identically prepared copies, retains the measurement-independent SLD
information benchmark, and includes no flux, response, background, or
nuisance model.  It also does not resolve whether one common measurement
attains the SLD matrix bound.  We therefore retain the calculation as a
fixed-energy quantum benchmark and reserve detector-level statements for the
version-locked public-configuration analysis in Sec.~\ref{sec:dune}.

\subsection{Independent numerical validation}
\label{subsec:numerical-validation}

The validation tolerances were fixed before the corresponding results were
interpreted.  Representation-covariance spot checks used centered five-point
steps of $10^{-5}$ for the four angular coordinates and
$10^{-7}\,\mathrm{eV}^{2}$ for both mass splittings.  The broader propagation
envelopes instead used $10^{-4}$ for the angles,
$10^{-7}\,\mathrm{eV}^{2}$ for $\Delta m_{21}^{2}$, and
$10^{-6}\,\mathrm{eV}^{2}$ for $\Delta m_{31}^{2}$, as specified in
Appendices~\ref{app:pmns-derivatives} and \ref{app:matter-frechet}.  The
derivative, covariance, data-processing, and profile-curvature checks pass;
the essential residuals are reported in those appendices and in
Appendices~\ref{app:dune-configuration} and \ref{app:dune-likelihood}.  The
minimum scaled eigenvalue of
$\QFI-\CFI_{\mathrm{flav}}$ is $-1.99\times10^{-15}$, more than a factor of
$500$ smaller in magnitude than the declared PSD tolerance and consistent
with floating-point roundoff.

The bound solvers receive separate validation.  For every three-coordinate
pure-state submodel at the vacuum and matter comparison points, the
tolerance-qualified numerical primal--dual
interval contains the exact metric-weighted cost
$C_{\mathrm H}=5$; the largest retained interval width is
$1.50\times10^{-6}$.  For the full-rank mixed state obtained by tracing out a
five-energy label, Eqs.~\eqref{eq:mixed_spectral_holevo_interval} and
\eqref{eq:mixed_spectral_holevo_ratio} give the retained bracket.  Its primal
and dual points satisfy the declared floating-point feasibility tolerances,
with KKT residual $1.08\times10^{-9}$.  It is a tolerance-qualified numerical
bracket, not an interval-arithmetic or exact-arithmetic certificate.

\subsection{Reproducibility and scope}
\label{subsec:reassessment-reproducibility}

The reproducibility bundle records the SHA-256 digest of the source PDF,
all coordinate scales, finite-difference steps, rank and PSD tolerances, and
the hashes of prerequisite result files.  Its independently generated result
manifests cover the representation-covariance, attainability, propagation,
measurement, event, profiling, symmetry, and rank calculations.  The
public-DUNE ancillary inputs are not redistributed; Appendix
\ref{app:dune-configuration} gives the versioned acquisition source and the
archive, manifest, and individual-file hashes needed to reacquire and verify
the exact inputs.

These checks establish the internal mathematical and numerical claims needed
for the present reassessment.  They do not turn the state-level SLD
information benchmark into an official experimental sensitivity, establish
joint attainability by a flavor detector, or validate a global Gaussian
approximation.  Section~\ref{sec:dune} therefore performs a separate
forward-model and likelihood validation against a checksum-locked public
DUNE configuration.

\section{Public-DUNE case study}
\label{sec:dune}

We now apply the operational hierarchy to a realistic broadband example.  The
calculation uses the DUNE Collaboration's public GLoBES configuration
accompanying Ref.~\cite{DUNEPublicConfig2021}, not the internal TDR likelihood.
The public release is explicitly a simplified summary whose sensitivities are
similar, but not identical, to the official DUNE results: it supplies detailed
far-detector flux, migration, efficiency, and channel inputs, but only a set of
normalization uncertainties rather than the full systematic model.  Every
result in this section is therefore a \emph{public-configuration benchmark}.
It is neither an official DUNE sensitivity nor a substitute for the
collaboration likelihood.

\subsection{Inputs, provenance, and independent implementation}
\label{subsec:dune-inputs}

The version-locked input is the ancillary configuration distributed with
arXiv:2103.04797v2.  The complete SHA-256 digest of the acquired archive is
\begin{equation}
  \texttt{
    574975ab4d1b7e77eb3e99dca0ae309050ff54145bcc380388fa805aef636824
  }.
  \label{eq:dune-archive-hash}
\end{equation}
All 64 extracted source files are hashed individually.  The SHA-256 digest of
the canonical manifest---formed from the sorted relative path and digest of
each source file---is
\begin{equation}
  \texttt{
    ba31245a2f832e6dbc2f47cb42e25df17a15292365fadf2d9d85858ffd28a054
  }.
  \label{eq:dune-manifest-hash}
\end{equation}
The spectrum, candidate-bank, likelihood, and continuous-profile stages
all independently produce the configuration fingerprint
\begin{equation}
  \texttt{
    cd16375b29ebdf65210d94658406a1f14aa16e0c0c83dca3e717083504b5de7f
  }.
  \label{eq:dune-configuration-fingerprint}
\end{equation}
This fingerprint agreement prevents numerical products generated from
different configuration versions from being combined silently.  The parsed
model has baseline $1284.9\,\mathrm{km}$, constant density
$2.848\,\mathrm{g\,cm^{-3}}$, a $40\,\mathrm{kt}$ far detector, 32 channels,
16 migration matrices, four analysis rules, and 264 selected reconstructed
bins.  Its nine declared normalization nuisances generate eight active
event-rate directions.  Appendix~\ref{app:dune-configuration} gives the full
configuration inventory.

The event engine is independent of GLoBES and implements the public AEDL subset
summarized in Appendix~\ref{app:dune-configuration}.  It uses the phase and
matter constants compiled into GLoBES 3.2.18
\cite{HuberEtAl2005,KoppEtAl2007}.  Repeated AEDL systematic names define one
shared nuisance coordinate; treating occurrences independently would define a
different likelihood.

The four rules are FHC and RHC electron appearance and muon disappearance.
They retain right- and wrong-sign signals, intrinsic-beam and neutral-current
backgrounds, and the eight declared tau-background channels.  The tau channels
are inactive because their supplied charged-current cross-section columns are
zero, so they generate neither events nor a nuisance-rate direction.  The
model therefore has nine declared normalization coordinates but eight active
directions.  Table~\ref{tab:dune-event-totals} gives the benchmark totals.

\begin{table}[t]
  \caption{
    Expected analysis-window events at the benchmark point.  These are
    predictions of the public configuration, not observed data.  The
    background totals include all declared background channels; the declared
    tau-background channels contribute zero events for the supplied tables.
    Totals are calculated from unrounded bin predictions and may therefore
    differ by (0.001) event from sums of the displayed rounded entries.
  }
  \label{tab:dune-event-totals}
  \centering
  \small
  \renewcommand{\arraystretch}{1.12}
  \begin{tabular}{@{}lrrr@{}}
    \toprule
    Rule & Signal & Background & Total\\
    \midrule
    FHC $e$ appearance & $2083.199$ & $788.822$ & $2872.021$\\
    RHC $e$ appearance & $407.129$ & $495.809$ & $902.939$\\
    FHC $\mu$ disappearance & $18812.948$ & $426.744$ & $19239.691$\\
    RHC $\mu$ disappearance & $10694.730$ & $254.886$ & $10949.616$\\
    \midrule
    All rules & $31998.006$ & $1966.260$ & $33964.266$\\
    \bottomrule
  \end{tabular}
\end{table}

\subsection{External forward-model and curvature validation}
\label{subsec:dune-validation}

Independent GLoBES 3.2.18 exporters provide spectra and the profiled
\texttt{chiMultiExp} likelihood, directly testing the parser and rate engine.

The reference exporter returns 80 reconstructed bins for each of four rules
and separately reports signal and background.  The resulting validation
vector therefore contains $4\times80\times2=640$ rule--bin--component entries.
The likelihood analysis subsequently applies
the declared analysis windows, retaining 66 bins per rule and hence
$4\times66=264$ reconstructed-count bins.

Let $\bm\mu^{\,\mathrm{ind}}$ and $\bm\mu^{\,\mathrm{GLoBES}}$ denote the
640-entry independent and reference vectors, and define
\begin{align}
  \epsilon_1
  &=
  \frac{
    \sum_j
    \left|
      \mu_j^{\,\mathrm{ind}}
      -
      \mu_j^{\,\mathrm{GLoBES}}
    \right|
  }{
    \sum_j
    \left|
      \mu_j^{\,\mathrm{GLoBES}}
    \right|
  },
  \label{eq:dune-relative-l1-error}
  \\
  \epsilon_2
  &=
  \frac{
    \left\|
      \bm\mu^{\,\mathrm{ind}}
      -
      \bm\mu^{\,\mathrm{GLoBES}}
    \right\|_2
  }{
    \left\|
      \bm\mu^{\,\mathrm{GLoBES}}
    \right\|_2
  }.
  \label{eq:dune-relative-l2-error}
\end{align}
All 640 entries agree to floating-point precision:
\begin{equation}
  \epsilon_{1}=3.82\times10^{-16},
  \qquad
  \epsilon_{2}=3.87\times10^{-16},
  \qquad
  \max_j
  \left|
    \mu_j^{\,\mathrm{ind}}
    -
    \mu_j^{\,\mathrm{GLoBES}}
  \right|
  =
  6.82\times10^{-13}\ \text{events}.
  \label{eq:dune-spectrum-validation}
\end{equation}

Analytic derivatives of every rule, bin, and component were separately
compared with centered five-point GLoBES derivatives.  The steps and full
validation protocol are documented in Appendices~\ref{app:dune-configuration}
and~\ref{app:dune-likelihood}; Table~\ref{tab:dune-validation} reports the
retained discrepancies, whose maximum is $2.32\times10^{-9}$ for $\dcp$.
The likelihood-curvature test uses a 73-point central Hessian: the generating
point, two axial points per coordinate, and four corner points per mixed
derivative.
At the generating point, the profiled GLoBES value is
$-1.15\times10^{-12}$, consistent with zero numerical noise.  One half of the
finite-difference Hessian agrees with the independently constructed
nine-declared-coordinate, eight-active-direction nuisance Schur complement to
relative $L_2$ error $3.26\times10^{-6}$.  Repeating the spectrum comparison
at the weakest selected point in each mass ordering gives a worst relative
$L_2$ error $5.32\times10^{-16}$.

\begin{table*}[t]
  \caption{
    Independent validation against GLoBES 3.2.18.  Derivative entries are
    relative $L_2$ discrepancies between analytic derivatives and five-point
    GLoBES differences.
  }
  \label{tab:dune-validation}
  \centering
  \small
  \renewcommand{\arraystretch}{1.12}
  \begin{tabular}{@{}lcl@{\hspace{2em}}lcl@{}}
    \toprule
    Diagnostic & Result & Criterion & Diagnostic & Result & Criterion\\
    \midrule
    Spectrum $L_2$ error
    & $3.87\times10^{-16}$ & $<10^{-12}$
    & $\partial_{\theta_{12}}\mu$
    & $6.09\times10^{-10}$ & $<10^{-7}$\\

    Spectrum maximum bin error
    & $6.82\times10^{-13}$ & $<10^{-9}$ events
    & $\partial_{\theta_{13}}\mu$
    & $3.71\times10^{-11}$ & $<10^{-7}$\\

    Profiled-Hessian error
    & $3.26\times10^{-6}$ & $<10^{-4}$
    & $\partial_{\theta_{23}}\mu$
    & $5.05\times10^{-11}$ & $<10^{-7}$\\

    Weak-point spectrum error
    & $5.32\times10^{-16}$ & $<10^{-12}$
    & $\partial_{\dcp}\mu$
    & $2.32\times10^{-9}$ & $<10^{-7}$\\

    Configuration fingerprints
    & identical & exact
    & $\partial_{\Delta m^2_{21}}\mu$
    & $3.32\times10^{-11}$ & $<10^{-7}$\\

    Fresh Fisher reproduction
    & $3.51\times10^{-16}$ & $<10^{-12}$
    & $\partial_{\Delta m^2_{31}}\mu$
    & $3.71\times10^{-12}$ & $<10^{-7}$\\
    \bottomrule
  \end{tabular}
\end{table*}

These tests close the local implementation loop: the event means, all six
physics derivatives, the sharing of the declared nuisance coordinates, the
inactive tau-background direction, and the curvature of the exact profiled
likelihood are verified against an external implementation of the same public
configuration.

\subsection{Local event information and nuisance retention}
\label{subsec:dune-local-information}

We use the dimensionless coordinate scaling of
Eq.~\eqref{eq:reassessment-scaling}, so that
\begin{equation}
  \widetilde{\CFI}
  =
  S^{\mathsf T}\CFI S.
\end{equation}
At the nominal normal-ordering benchmark, the fixed-systematic and
nuisance-profiled spectra are
\begin{align}
  \operatorname{eig}(\widetilde{\CFI}_{\mathrm{fix}})
  ={}&
  (0.5173,\,33.2228,\,172.4463,\,1.98470\times10^4,
    3.21838\times10^4,\,2.72308\times10^5),
  \nonumber\\
  \operatorname{eig}(\widetilde{\CFI}_{\mathrm{prof}})
  ={}&
  (0.5019,\,29.4604,\,159.5665,\,1.68413\times10^4,
    2.55608\times10^4,\,1.65817\times10^5).
  \label{eq:dune-fisher-spectra}
\end{align}
Both matrices have numerical rank six under the declared criterion
\begin{equation}
  q_a
  >
  \max\!\left\{
    10^{-10}q_{\max},
    10^{-14}
  \right\}.
  \label{eq:dune-numerical-rank-criterion}
\end{equation}

The generalized nuisance-retention eigenvalues, defined by
\begin{equation}
  \widetilde{\CFI}_{\mathrm{prof}}\bm v_a
  =
  r_a\widetilde{\CFI}_{\mathrm{fix}}\bm v_a,
  \label{eq:dune-retention-eigenproblem}
\end{equation}
are
\begin{equation}
  \bm r
  =
  (0.5717,\,0.6709,\,0.8758,\,0.9838,\,0.9902,\,0.9982).
  \label{eq:dune-retention-spectrum}
\end{equation}
Their product is $0.326662$, equal to
\begin{equation}
  \frac{
    \det\widetilde{\CFI}_{\mathrm{prof}}
  }{
    \det\widetilde{\CFI}_{\mathrm{fix}}
  }
\end{equation}
to relative error $4.21\times10^{-14}$.  Thus the eight active normalization
directions substantially degrade two generalized physics combinations while
leaving four relatively stable.  The ninth declared normalization coordinate,
associated with the inactive tau-background mode, contributes no event-rate
direction and does not change the profiled physics information.

For reference, direct inversion of the local profiled matrix gives the native
coordinate marginal standard deviations
\begin{equation}
  (0.9553,\,0.003439,\,0.007289,\,0.4907,
    9.284\times10^{-5},\,1.325\times10^{-4}),
  \label{eq:dune-local-standard-deviations}
\end{equation}
where the first four entries are in radians and the last two in
$\mathrm{eV}^{2}$.  These are local Gaussian curvature diagnostics, not
expected experimental errors: no external oscillation-parameter constraints
are included, the public systematics are simplified, and one tangent expansion
cannot represent mass-ordering, octant, or CP-periodic alternatives.

\subsection{Rank restoration, effective rank, and robustness}
\label{subsec:dune-robustness}

The local calculation was repeated for 226 documented records.  The
176-record physics envelope is a union of three scans, not their Cartesian
product:

\begin{enumerate}
  \item a 112-point ordering--octant--phase grid containing both orderings,
        seven values
        \begin{equation}
          \theta_{23}
          =
          (40,\,41.3,\,43.29,\,45,\,47.9,\,49.9,\,50)^\circ,
        \end{equation}
        and eight phases
        \begin{equation}
          \dcp
          =
          (0,\,45,\,90,\,135,\,180,\,225,\,270,\,315)^\circ;
        \end{equation}

  \item a 48-point refinement sampling $\dcp$ every $15^\circ$ from
        $0^\circ$ through $345^\circ$, at
        $\theta_{23}=49.9^\circ$ in normal ordering and
        $\theta_{23}=49.8^\circ$ in inverted ordering; and

  \item sixteen one-axis endpoint evaluations covering the lower and upper
        NuFIT~6.0 three-sigma endpoints of
        $\theta_{12}$, $\theta_{13}$, $\Delta m^2_{21}$, and the
        ordering-appropriate atmospheric mass splitting in each ordering
        \cite{NuFIT2024}.
\end{enumerate}

This construction gives $112+48+16=176$ physics-envelope records.  Every one
has numerical rank six under
Eq.~\eqref{eq:dune-numerical-rank-criterion}.  The remaining 50 records
comprise ten density variations, ten normalization-prior variations, ten
reconstructed-binning choices, eighteen rule-subset calculations, and two
additional CP-periodicity endpoints, giving
$176+10+10+10+18+2=226$ records in total.  Exact CP periodicity between $0$
and $2\pi$ is recovered
with relative eigenvalue error $6.86\times10^{-16}$.

Numerical rank nevertheless overstates practical identifiability.  The
weakest physics-envelope record occurs in inverted ordering at the low
$\Delta m^2_{21}$ endpoint.  Its profiled scaled eigenvalues are
\begin{equation}
  (0.04999,\,16.996,\,176.893,\,1.59355\times10^4,
    2.50788\times10^4,\,1.65037\times10^5).
  \label{eq:dune-weak-spectrum}
\end{equation}
It has numerical rank six at relative thresholds $10^{-10}$ and $10^{-8}$,
but effective rank five at the declared practical threshold $10^{-6}$.

Writing the local coordinate displacement as
$\delta\bm\lambda=S\,\delta\bm x$, the normalized weakest eigenvector is
approximately
\begin{equation}
  0.761\,x_{\theta_{12}}
  +0.556\,x_{\dcp}
  -0.326\,x_{\Delta m^2_{21}}
  -0.074\,x_{\Delta m^2_{31}},
  \label{eq:dune-weak-direction}
\end{equation}
with the omitted $x_{\theta_{13}}$ and $x_{\theta_{23}}$ components below
$4\times10^{-4}$ in magnitude.  At the same record, the local marginal widths
in $\theta_{12}$ and $\dcp$ are $3.41$ and $2.50$ radians, respectively.
A Gaussian approximation with widths comparable to the physical parameter
domains cannot be promoted to a global confidence region.

The origin of the additional local directions is directly testable by
coarsening reconstructed energy.  The group sizes in
Table~\ref{tab:dune-binning-rank} define nested, parameter-independent
coarsenings of adjacent reconstructed bins.  Data processing therefore
precludes an information gain under successive grouping.  The smallest
eigenvalue decreases while numerical rank six is retained through groups of
eight.  Replacing each rule by one total count leaves four observations and
gives numerical rank four in both orderings.  Thus the two additional local
directions are carried by spectral shape, not by the four total rates.

\begin{table}[t]
  \caption{
    Smallest profiled scaled eigenvalue and numerical rank after nested
    reconstructed-bin grouping at the selected stress points.  Ranks use
    Eq.~\eqref{eq:dune-numerical-rank-criterion}.
  }
  \label{tab:dune-binning-rank}
  \centering
  \small
  \renewcommand{\arraystretch}{1.12}
  \begin{tabular}{@{}crrcc@{}}
    \toprule
    Group size
    & NO $\lambda_{\min}$
    & IO $\lambda_{\min}$
    & NO rank
    & IO rank\\
    \midrule
    1   & $6.515\times10^{-2}$ & $4.999\times10^{-2}$ & 6 & 6\\
    2   & $5.357\times10^{-2}$ & $4.302\times10^{-2}$ & 6 & 6\\
    4   & $4.790\times10^{-2}$ & $3.443\times10^{-2}$ & 6 & 6\\
    8   & $2.863\times10^{-2}$ & $2.320\times10^{-2}$ & 6 & 6\\
    Rule totals & $\simeq0$ & $\simeq0$ & 4 & 4\\
    \bottomrule
  \end{tabular}
\end{table}

Each individual rule is also found to have numerical rank six because its 66
retained energy bins can supply six independent local derivative directions.
The number of bins permits, but does not by itself guarantee, this result.
The individual-rule matrices are not practically well conditioned:
condition numbers reach $1.27\times10^9$, and several rules have effective
rank four at threshold $10^{-6}$.  At the selected stress points, combining
all four rules raises the smallest eigenvalue above that of the best
individual rule by factors $42.4$ in normal ordering and $152.4$ in inverted
ordering.  Appearance and disappearance samples and both horn polarities are
therefore complementary in the weak directions, even though every individual
subset has a numerical sixth eigenvalue.

At the standard nuisance widths, the smallest generalized retained fraction
over the broad ordering--octant--phase grid is $0.519$.  Multiplying every
supplied normalization error by four lowers the worst fraction to $0.108$ but
does not remove the sixth numerical direction at the declared rank tolerance.
This conclusion applies only to the nine-declared-coordinate,
eight-active-direction normalization model in the public files; it is not a
robustness test of the full TDR covariance.

\subsection{Conditional nonlocal likelihood pilot}
\label{subsec:dune-nonlocal-pilot}

To expose failures of the local quadratic approximation, we evaluate the
exact profiled Poisson deviance on a finite candidate bank.  The likelihood is
Eq.~\eqref{eq:constrained_poisson_likelihood} specialized to the nine named
linear normalization pulls of the public configuration.  Its nuisance
Hessian is strictly positive definite on the positive-mean domain; every
fit has the unique interior profile found by the
positivity-preserving damped Newton iteration.  The inactive tau-background
pull equals its auxiliary center and contributes neither event-rate dependence
nor a penalty at the optimum.  Appendix~\ref{app:dune-likelihood} gives the
explicit objective, proof, and numerical stationarity checks.

The pilot bank contains both orderings, the seven $\theta_{23}$ values listed
above, and 25 CP-phase values: the $15^\circ$ grid from $0^\circ$ through
$345^\circ$ together with the exact $212^\circ$ generating value.  It
therefore contains $2\times7\times25=350$ candidates.  The coordinates
$\theta_{12}$, $\theta_{13}$,
$\Delta m^2_{21}$, and $\Delta m^2_{31}$ are fixed to ordering-specific
references.  The bank is thus a conditional diagnostic, not a
six-coordinate fit.

For a true phase $\delta_0$ represented in the bank, the grid
likelihood-ratio statistic is
\begin{equation}
  q_{\delta_0}
  =
  \min_{\substack{
    \bm\lambda\in\mathcal{B}\\
    \dcp=\delta_0
  }}
  \chi^2_{\mathrm{prof}}(\bm\lambda;\bm n,\bm c)
  -
  \min_{\bm\lambda\in\mathcal{B}}
  \chi^2_{\mathrm{prof}}(\bm\lambda;\bm n,\bm c),
  \label{eq:dune-grid-cp-statistic}
\end{equation}
where $\chi^2_{\mathrm{prof}}$ is defined in
Eq.~\eqref{eq:appF-profile-definition} and $\mathcal{B}$ denotes the full
350-point bank.  The restricted minimum
therefore profiles over both orderings and all retained $\theta_{23}$ values
on the grid.  The $\theta_{23}$ statistic is defined analogously by fixing
$\theta_{23}$ to its true grid value and profiling over ordering and $\dcp$.
Coverage is reported as the fraction of pseudoexperiments satisfying
\begin{equation}
  q_{\mathrm{true}}
  \leq
  (1.0000,\,2.7055,\,3.8415)
  \label{eq:dune-one-dof-thresholds}
\end{equation}
for the nominal $68.27\%$, $90\%$, and $95\%$ one-degree-of-freedom
thresholds.

Three hundred Poisson pseudo-data sets generated at the nominal truth are
each analyzed in two auxiliary-measurement ensembles.  In the first, the nine
Gaussian centers are conditioned to be zero.  In the second, they are
independently sampled from unit normal distributions around the true
zero-valued nuisance coordinates.  The same count realizations are used in
the paired comparison.  All $300\times2\times350=210000$ candidate fits
converge.  For the $68.27\%$, $90\%$, and $95\%$ thresholds, the CP-phase
inclusion fractions are $(0.8067,0.9533,0.9933)$ with fixed auxiliary centers
and $(0.7767,0.8967,0.9567)$ with fluctuated centers; the binomial Monte Carlo
standard errors range from $0.0047$ to $0.0240$.

The conditional fixed-center ensemble overcovers at every displayed level.
With fluctuated auxiliary centers, the $90\%$ and $95\%$ results are
compatible with nominal coverage at this pilot precision, whereas the
$68.27\%$ interval still overcovers by $3.91$ Monte Carlo standard errors.
The auxiliary-measurement ensemble is therefore a material part of the
sampling definition.

The $\theta_{23}$ result fails more dramatically: the seven-point grid selects
the exact truth as the best fit in 290 of 300 fixed-center trials and 284 of
300 fluctuated-center trials.  The grid likelihood-ratio statistic is
consequently zero for most pseudoexperiments, producing nominal $68.27\%$
coverages $0.9733$ and $0.9633$.  This is a resolution failure of the finite
grid, not evidence for unusually conservative continuous inference.
Continuous or adaptive minimization is required.  More generally, neutrino
fits with periodic phases, discrete orderings, boundaries, and nuisance
parameters require explicit coverage calibration rather than automatic use
of a Wilks approximation \cite{NOvAProfiledFC2022}.

\subsection{Expanded four-coordinate Asimov profile and claim boundary}
\label{subsec:dune-continuous-profile}

As an intermediate nonlocal check, we continuously minimize the Asimov
deviance over $(\theta_{13},\theta_{23},\dcp,\Delta m^2_{31})$ in each
ordering, using five candidate-bank seeds per ordering and exact
nuisance profiling at every physics evaluation.  The parameters are mapped to
a dimensionless unit hypercube and minimized with bounded continuous
optimization.  The solar coordinates $\theta_{12}$ and
$\Delta m^2_{21}$ remain fixed to their ordering-specific references.

The bounded domains are the NuFIT~6.0 one-dimensional three-sigma intervals
\cite{NuFIT2024}.  In normal ordering,
$\theta_{13}\in[8.19^\circ,8.89^\circ]$,
$\theta_{23}\in[41.3^\circ,49.9^\circ]$, and
$\Delta m^2_{31}\in[2.451,2.578]\times10^{-3}\,\mathrm{eV}^2$; in inverted
ordering, the corresponding intervals are
$[8.25^\circ,8.93^\circ]$, $[41.5^\circ,49.8^\circ]$, and
$[-2.4721,-2.3461]\times10^{-3}\,\mathrm{eV}^2$.  The latter is converted
from NuFIT's $\Delta m^2_{32}$ interval using
$\Delta m^2_{31}=\Delta m^2_{32}+\Delta m^2_{21}$.  Both orderings use the
complete CP period, with $\dcp=0$ and $2\pi$ physically identified.

The normal-ordering fit recovers the generating point with zero deviance and
no active boundary.  The inverted-ordering result is shown in
Table~\ref{tab:dune-continuous-fit}.

\begin{table*}[t]
  \caption{
    Expanded four-coordinate Asimov profile.  The inverted-ordering
    number remains conditional on fixed solar coordinates and the declared
    bounded parameter domain.
  }
  \label{tab:dune-continuous-fit}
  \centering
  \small
  \renewcommand{\arraystretch}{1.12}
  \begin{tabular}{@{}lrrrrrl@{}}
    \toprule
    Ordering
    & Deviance
    & $\theta_{13}$
    & $\theta_{23}$
    & $\dcp$
    & $\Delta m^2_{31}$
    & Boundary\\
    \midrule
    NO
    & $0$
    & $8.620^\circ$
    & $43.290^\circ$
    & $212.000^\circ$
    & $+2.5110\times10^{-3}\,\mathrm{eV}^2$
    & none\\

    IO
    & $464.920$
    & $8.930^\circ$
    & $47.441^\circ$
    & $276.109^\circ$
    & $-2.38691\times10^{-3}\,\mathrm{eV}^2$
    & upper $\theta_{13}$\\
    \bottomrule
  \end{tabular}
\end{table*}

The expanded profile lowers the best conditional inverted-ordering deviance
from $490.119$ in the finite bank to $464.920$, a change of $25.199$ or
$5.14\%$.  This difference cannot be attributed to grid spacing alone: the
bank varies only ordering, $\theta_{23}$, and $\dcp$, whereas the continuous
search also profiles $\theta_{13}$ and $\Delta m^2_{31}$.  The optimum reaches
the upper $\theta_{13}$ domain boundary and uses nuisance pulls as large as
$6.44$ standard deviations; the two solar coordinates remain fixed.  The
value is therefore conditional on the parameter domain, auxiliary constraints,
and profiled-coordinate set and must not be quoted as a mass-ordering
sensitivity.  Unlike the spectra and local likelihood curvature, the
nonlocal deviance at this boundary optimum has not been independently
cross-checked with GLoBES.  Moreover, the $6.44$-standard-deviation pull
probes the linear normalization model far outside its nominal unit-Gaussian
region.

The case study supports three scoped conclusions.  First, the independent
engine reproduces the checksum-locked public configuration at the spectrum,
derivative, and profiled-curvature levels, including the distinction between
nine declared normalization coordinates and eight active event-rate
directions.  Second, reconstructed spectral shape, horn polarity, and
appearance/disappearance complementarity produce numerical rank six at
relative tolerance $10^{-10}$ throughout the documented 176-point physics
envelope, although the weakest direction has effective rank five at threshold
$10^{-6}$.  Third, the nonlocal pilot demonstrates that finite-grid resolution
and the auxiliary-measurement ensemble materially alter the reported grid
coverage, while the expanded continuous profile remains conditional on its
bounded domain and fixed solar coordinates.  It does not provide
official DUNE sensitivity, global confidence intervals, or a
coverage-calibrated ordering test.  Those exclusions are part of the result
rather than qualifications to be removed in presentation.

\section{Discussion and limitations}
\label{sec:discussion}

The preceding sections connect four logically distinct statements: an exact
geometric theorem, a quantum multiparameter attainability calculation, a
local information analysis for reconstructed events, and a diagnostic
nonlocal likelihood study.  Keeping these statements separate is essential.
Neither full matrix rank nor a small local Cramér--Rao bound, by itself,
establishes stable six-parameter inference in a realizable experiment.  This
section summarizes what follows from the combined analysis, what depends on
declared conventions, and what remains necessary for global experimental
claims.

\subsection{Mathematical rank and practical identifiability}
\label{subsec:discussion-rank}

For a pure state in a three-dimensional Hilbert space, the exact
single-setting rank bound is geometric: the physical state lies in
$\mathbb{CP}^{2}$, whose real tangent dimension is four.  No choice of
coordinates can turn one pure-qutrit state at a fixed propagation
setting---fixed baseline, energy, source flavor, and matter profile---into a
regular six-parameter statistical model.  Broadband rank restoration does not
contradict this theorem.  Distinct retained energy or channel labels supply
distinct tangent maps, and the kernel of their positive-weight aggregate is
the intersection of the active conditional kernels.  The relevant question is
therefore not only whether each bin is singular, but whether its null
directions remain common to all statistically distinguishable bins.

Exact rank is nevertheless only a first identifiability criterion.  In the
controlled two-energy study, all 171 energy pairs have numerical rank six at
the declared relative tolerance $10^{-10}$, but only 140 pairs, or $81.9\%$,
have effective rank six at threshold $10^{-6}$, and only 78 pairs, or
$45.6\%$, do so at $10^{-4}$.  The weakest numerically full-rank pair has
condition number $4.75\times10^{7}$.  The public-DUNE case study shows the
same distinction in a realistic forward model: all 176 documented
physics-envelope records have numerical rank six at relative tolerance
$10^{-10}$, yet the weakest selected record has effective rank five at
threshold $10^{-6}$ and local angular widths of order radians.  Conversely,
collapsing the four reconstructed spectra to four rule totals gives numerical
rank four in both orderings.  Spectral shape therefore restores additional
local directions at the declared numerical tolerance, but the strength of the
weakest restored direction must be reported separately from its numerical
existence.

A numerical rank statement also requires a norm and coordinate convention.
Exact algebraic rank is invariant under every nonsingular reparameterization,
whereas a floating-point cutoff applied to coordinates carrying different
physical units is not.  For the same nominal public-DUNE fixed-systematic
information tensor, applying a relative $10^{-10}$ cutoff directly in the
native angular and $\mathrm{eV}^{2}$ coordinates reports numerical rank five
and condition number $1.43\times10^{9}$.  After the declared dimensionless
scaling, the same bilinear form has numerical rank six and condition number
$5.26\times10^{5}$.  This is not a physical change in the experiment; it is a
change in the norm used to decide which numerical directions are small.

Accordingly, this work distinguishes three levels:

\begin{enumerate}
  \item exact or structural rank, established by an analytic property of the
        statistical model;
  \item numerical rank, evaluated after a declared coordinate scaling and
        floating-point tolerance; and
  \item effective rank, evaluated at a separately declared scientific
        information threshold.
\end{enumerate}

Exact rank and exact estimability are coordinate invariant.  Numerical
eigenvalues, condition numbers, numerical pseudoinverses, and effective-rank
statements are meaningful only together with the scaling and tolerance used
to construct them.

\subsection{Scientific weights and multiparameter attainability}
\label{subsec:discussion-weights}

A scalar precision statement additionally requires a loss function.  For a
weight matrix $W$ that is positive definite on the identifiable tangent
space, the SLD benchmark and the Holevo bound constrain
\begin{equation}
  \operatorname{tr}(WV),
\end{equation}
not an unweighted notion of ``total precision.''  Here
$\operatorname{tr}$ denotes the parameter-space trace, in contrast to the
Hilbert-space trace $\Tr$ used in the quantum definitions.  The intrinsic
choice $W=\QFI$ is useful because it produces a dimensionless,
coordinate-covariant scalar cost and gives
\begin{equation}
  C_{\mathrm S}(\QFI)=p
\end{equation}
on a $p$-dimensional identifiable tangent.  It is not a unique scientific
objective.  A design prioritizing CP resolution, a mass splitting, or a
particular estimable combination defines a different weight and may rank
measurements differently.

Table~\ref{tab:discussion-weight-dependence} illustrates this dependence for
the exact two-coordinate $(\theta_{13},\dcp)$ pure-state model in vacuum at
$L/E=520\,\mathrm{km/GeV}$.  The incompatibility coefficient is fixed at
$\beta=0.9702$ in all rows, but the scalar Holevo penalty varies with the
weight.  In the declared scaled coordinates, the nonintrinsic weights are
\begin{equation}
  W_{\mathrm{equal}}=\mathbb{I}_2,
  \qquad
  W_{\dcp}=\operatorname{diag}(1,9),
  \qquad
  W_{\theta_{13}}=\operatorname{diag}(9,1).
  \label{eq:discussion-example-weights}
\end{equation}
A simultaneous nonorthogonal coordinate transformation of the derivatives and
the covariant weight preserves both scalar costs and their ratio to better
than $3\times10^{-16}$.  The variation in the table is therefore a change of
scientific loss, not a coordinate artifact.

\begin{table}[t]
  \caption{
    Dependence of the exact two-parameter attainable cost on the declared
    scientific weight.  The last two rows assign a nine-to-one relative
    priority in the scaled $(\theta_{13},\dcp)$ coordinates.
  }
  \label{tab:discussion-weight-dependence}
  \centering
  \small
  \renewcommand{\arraystretch}{1.12}
  \begin{tabular}{@{}lrrr@{}}
    \toprule
    Weight
    & $C_{\mathrm S}$
    & $C_{\mathrm H}$
    & $C_{\mathrm H}/C_{\mathrm S}$\\
    \midrule
    Intrinsic metric, $W=\QFI$
    & $2.0000$ & $3.2195$ & $1.6097$\\
    Equal scaled-coordinate loss
    & $6.5320$ & $7.3787$ & $1.1296$\\
    Ninefold $\dcp$ priority
    & $57.7238$ & $59.2602$ & $1.0266$\\
    Ninefold $\theta_{13}$ priority
    & $7.5961$ & $10.8087$ & $1.4229$\\
    \bottomrule
  \end{tabular}
\end{table}

The attainability question is also separate from invertibility.  For the
six-coordinate mixed spectral model whose QFIM has numerical rank six at the
declared scaling and relative tolerance $10^{-10}$, the tolerance-qualified
numerical primal--dual bracket
\begin{equation}
  1.8675761
  \leq
  \frac{C_{\mathrm H}}{C_{\mathrm S}}
  \leq
  1.8675771
  \label{eq:discussion-mixed-holevo-ratio}
\end{equation}
shows that the inverse SLD QFIM understates the attainable metric-weighted
joint cost within the reported numerical bracket.  This is consistent with
the general role of the Holevo bound in multiparameter quantum estimation
\cite{Holevo2011,AlbarelliEtAl2020,
AlbarelliFrielDatta2019,DemkowiczEtAl2020}.
The bracket is supported by primal and dual iterates satisfying the declared
floating-point feasibility and residual tolerances.  It is not an
interval-arithmetic or exact-arithmetic proof of the displayed decimal
endpoints.  It concerns the declared local quantum model, weight, and
asymptotic measurement class and does not construct a neutrino detector that
implements a measurement attaining the bracketed Holevo value.

That distinction is especially important for oscillation experiments.  The
SLD QFIM is a state-level information tensor that upper-bounds the CFI of
every fixed POVM in the Loewner order.  In a multiparameter model, however,
this matrix statement does not imply that one common measurement attains all
of its directions simultaneously.  Charged-current flavor tagging,
reconstructed energy, event selection, backgrounds, and nuisance constraints
define a much smaller operational class than the class considered by an
optimized state-level bound.

In the controlled state-to-event chain, numerical rank six at the declared
tolerance survives every layer, but the nuisance-profiled geometric-mean
information retention is only $1.14\times10^{-2}$, the weakest generalized
fraction is $7.34\times10^{-7}$, and the corresponding worst generalized
standard-error inflation relative to state QFI is $1.17\times10^{3}$.  These
controlled numbers are not DUNE forecasts.  They demonstrate why numerical
rank and state-QFI magnitude cannot substitute for measurement accessibility,
detector response, or nuisance treatment.

\subsection{Local information versus global inference}
\label{subsec:discussion-local-global}

Fisher information is the curvature of a regular likelihood at a specified
point.  It is therefore well suited to local identifiability, infinitesimal
design comparisons, and diagnosing nuisance degradation.  It does not encode
the global topology of a periodic phase, a discrete mass ordering,
octant-related alternatives, parameter boundaries, or separated likelihood
modes.  A local inverse can be particularly misleading when its Gaussian
width is comparable to a physical domain or when the optimum lies on a
boundary.

The conditional public-DUNE pilot makes this limitation quantitative.  With
the nine Gaussian auxiliary-measurement centers fluctuated, the acceptance
region obtained from the nominal one-degree-of-freedom $68.27\%$ threshold
contains the true CP phase in $0.7767\pm0.0240$ of the pseudoexperiments, a
$3.91$-Monte-Carlo-standard-error departure from the nominal value.  One of
these nine declared auxiliary coordinates is the
inactive tau-background mode.  Its center fluctuates in the joint auxiliary
ensemble, but because its event-rate direction vanishes identically, its
profiled optimum equals its auxiliary center and it does not affect the
profiled event likelihood.  The seven-point $\theta_{23}$ grid produces much
stronger overcoverage because it selects the exact truth as the best grid
point in most pseudoexperiments.  These effects are not properties of the
local Fisher matrix.  They arise from the declared auxiliary-measurement
ensemble and the finite candidate-bank optimization.

Continuous profiling improves the optimization but does not, by itself,
complete the inference.  The expanded four-coordinate Asimov profile lowers
the conditional wrong-ordering deviance by $25.199$, or $5.14\%$, relative to
the finite bank.  Because the bank varies two continuous physics coordinates
whereas the expanded search varies four, this change combines finer
optimization with a larger profiled parameter set and is not a direct measure
of grid-discretization error.  The optimum reaches the upper $\theta_{13}$
boundary, uses nuisance pulls as large as $6.44$ constraint standard
deviations, and still fixes the two solar coordinates.  The resulting deviance
is conditional on the parameter domain, auxiliary-constraint prescription,
and profiled-coordinate set and cannot be interpreted as a mass-ordering
sensitivity.  Neutrino analyses with bounded, periodic, discrete, and
constrained parameters require an explicitly defined Monte Carlo
construction, as illustrated by the profiled Feldman--Cousins procedure
developed by NOvA \cite{NOvAProfiledFC2022}.

For a global experimental claim, a subsequent analysis must replace the
finite bank by continuous or demonstrably converged adaptive profiling over
every materially correlated coordinate in both orderings.  Coverage should
then be evaluated at multiple regular and degenerate truth points, with the
auxiliary-measurement ensemble, parameter domains, ordering treatment,
boundary rules, and test statistic fixed before examining the
pseudoexperiments.  The number of pseudoexperiments at each truth should be
chosen prospectively to achieve a declared Monte Carlo precision for the
coverage quantities being reported.  A clean rerun from the frozen
configuration and random-number seeds is also required.  Until those
conditions are met, the local Fisher spectra and the present finite-bank pilot
should be interpreted as diagnostics rather than confidence regions.

\subsection{Scope of the public-DUNE benchmark}
\label{subsec:discussion-dune-scope}

The public-DUNE result is stronger than a schematic event-count scaling but
narrower than an official experimental analysis.  Its forward model is tied
to 64 individually hashed public files and independently reproduces GLoBES
spectra, all six physics derivatives, and the local profiled-likelihood
curvature.  It contains four appearance/disappearance and horn-polarity rules,
264 selected reconstructed bins, right-sign and wrong-sign channels,
efficiencies, migration matrices, and backgrounds.  The AEDL files declare
nine shared normalization coordinates, of which eight generate active
event-rate directions.  The ninth is attached to eight declared
tau-background channels whose rates vanish because the supplied tau
charged-current cross-section columns are identically zero.  These checks
establish that the conclusions follow from the declared public configuration
rather than from an undocumented digitization or a circular comparison
between two paths through the same rate engine.

The same provenance does not supply the internal DUNE likelihood.  The public
release is an approximation to the TDR analysis and explicitly uses simplified
normalization systematics \cite{DUNEPublicConfig2021}.  It does not expose the
full covariance of flux, interaction-model, detector-calibration,
near-detector, and cross-sample uncertainties used in a collaboration
analysis.  The constant-density approximation, external oscillation domains,
and Gaussian normalization constraints are therefore features of the
benchmark definition, not claims about the ultimate experimental treatment.

The robustness analysis is likewise broad but not exhaustive.  Its
176-record physics envelope is the documented union of a 112-point
ordering--octant--phase grid, a 48-point CP-phase refinement, and sixteen
one-axis NuFIT endpoint calculations.  The additional analysis records probe
density, reconstructed binning, normalization widths, rule subsets, and exact
CP periodicity \cite{NuFIT2024}.  This is an envelope stress test, not a scan
of the full correlated oscillation volume or a full systematic covariance.
The finding that all 176 physics records have numerical rank six at relative
tolerance $10^{-10}$ is consequently a documented local property of this
configuration under the declared dimensionless coordinate scaling.  It is not
proof of uniform global identifiability, official DUNE precision, or ordering
coverage.

\subsection{Relation to recent neutrino quantum-estimation studies}
\label{subsec:discussion-literature}

Quantum estimation was applied to two-flavor neutrino oscillations by
Nogueira \emph{et al.}, who compared state QFI with the Fisher information of
mass and flavor measurements for the mixing angle
\cite{NogueiraEtAl2017}.  More recently, Ignoti \emph{et al.} compared state
QFI with ideal flavor-measurement information for $\dcp$ and examined the
information gain near the second oscillation maximum
\cite{IgnotiEtAl2025}.  Yadav, Subba, and Shi studied single-parameter QFI
profiles for $\theta_{23}$, $\dcp$, and $\Delta m^2_{31}$ in long-baseline
three-flavor oscillations \cite{YadavSingleQFI2026}.  These studies establish
important scalar and fixed-measurement precedents; neither single-parameter
neutrino QFI nor the distinction between QFI and flavor-measurement
information is claimed here as new.

Chundawat and Li compared the information content of coherent reactor
antineutrino evolution with the incoherent mixed-state description relevant
to solar neutrinos for $\theta_{12}$ and $\Delta m^2_{21}$
\cite{ChundawatLi2026}.  Frugiuele \emph{et al.} computed state QFI and ideal
flavor-projection CFI for the standard oscillation parameters using
accelerator muon-neutrino, accelerator muon-antineutrino, and reactor
electron-antineutrino states in vacuum \cite{FrugiueleEtAl2026}.  Their
analysis explicitly demonstrates parameter-dependent differences between
state information and flavor-accessible information.  The broad observation
that flavor projection may fail to attain state QFI is therefore not a
novelty claim of the present paper.

At the event level, Chundawat, Delgadillo, and Li compare intrinsic state
information about $\dcp$ with Poisson Fisher information obtained from
reconstructed T2K and NO$\nu$A event spectra
\cite{ChundawatEtAl2026}.  Their analysis is direct prior art for connecting
state information with experimentally reconstructed CP information.  The
extension developed here is joint and structural: it treats all six
oscillation coordinates, tests spectral kernel intersections, includes
nuisance Schur complementation, validates exact likelihood curvature, and
distinguishes numerical rank from effective information strength.

For multiparameter geometry, Huang \emph{et al.} analyze the six-coordinate
fixed-energy QFIM, including off-diagonal entries and
parameter-dependent basis transformations \cite{HuangEtAl2026}.  Yadav,
Subba, and Shi analyze the three-coordinate QFIM for
$(\theta_{23},\dcp,\Delta m^2_{31})$ and use fidelity and QFIM diagnostics to
study quantum states associated with probability-degenerate parameter points
\cite{YadavDegeneracy2026}.  Our reassessment reproduces the displayed
fixed-energy calculation of Ref.~\cite{HuangEtAl2026} but changes its
interpretation through the structural pure-qutrit rank bound, the
moving-frame connection, the CP tensor pullback, and explicit attainable
Holevo costs.  Neither the use of a multiparameter neutrino QFIM nor the
observation that probability-degenerate points may correspond to distinct
quantum states is claimed here as new.

The scoped contribution of this work is instead the validated integration of
five elements in one reproducible framework:

\begin{enumerate}
  \item representation-covariant quantum geometry;
  \item the structural pure-state rank bound and exact positive-weight
        spectral kernel-intersection criterion;
  \item weighted multiparameter attainability using exact pure-state Holevo
        results and tolerance-qualified numerical primal--dual brackets;
  \item measurement- and detector-level information loss; and
  \item an independently validated, nuisance-profiled public-configuration
        likelihood.
\end{enumerate}

This is a contribution based on the combined framework and the verified
connections among its layers, not a priority claim for each ingredient in
isolation.

\subsection{Remaining extensions}
\label{subsec:discussion-extensions}

Three extensions would materially change the scope of the conclusions.
First, a full experimental analysis requires a collaboration-grade systematic
model, including correlated shape uncertainties, interaction-model
uncertainties, detector-calibration effects, and near-to-far constraints.
Second, global inference requires continuous multi-coordinate fits and
calibrated multi-truth confidence constructions rather than Fisher inverses or
a finite candidate bank.  Third, closing the quantum-to-detector gap requires
a physical analysis of measurements implementable through neutrino production
and detection, including possible decoherence, wave-packet separation, and
open-system effects where they are experimentally relevant.  The state-QFI
and Holevo results should otherwise be read as quantum benchmarks, not as
proposed detector performance.

Within those boundaries, the present results are sufficient for the paper's
central methodological conclusion: a reliable precision analysis must specify
the state family, estimable functions, coordinate scales, scientific weight,
admissible measurement, detector response, nuisance model, and sampling
construction.  Omitting any of these layers can turn an exact local matrix
calculation into an unjustified experimental claim.

\section{Conclusions}
\label{sec:conclusions}

We have developed a measurement-explicit framework for applying quantum
estimation theory to three-flavor neutrino oscillations.  Its central
methodological point is that three logically different questions must be kept
separate: what local information is encoded in the propagated quantum state,
what part of that information is accessible through a specified admissible
measurement, and what information remains in reconstructed event data after
detector response and nuisance profiling.  The SLD QFIM is therefore not, by
itself, an experimental forecast.  It is a local state-level information
tensor that upper-bounds the CFI of every fixed POVM, while its relevance to
simultaneous precision depends on identifiability, multiparameter
attainability, and the complete experimental measurement chain.

At the geometric level, we showed that a passive parameter-dependent change of
basis requires the associated connection term and leaves the projected quantum
geometric tensor invariant.  Interpreting the same parameter-dependent unitary
as an additional physical encoding instead defines a different state family
and therefore a different statistical model.  We also derived the tensorial CP
pullback, including the reflection Jacobian on every parameter index and the
matter-potential reversal required when relating neutrino and antineutrino
propagation in matter.  These results remove representation artifacts before
any precision bound is interpreted.

At one fixed propagation setting---fixed baseline, energy, source flavor, and
matter profile---the propagated three-flavor pure state is a qutrit and
therefore has at most four real local tangent directions.  A QFIM written in
six oscillation coordinates is consequently singular, and its ordinary inverse
cannot represent six simultaneous precision bounds.  A Moore--Penrose inverse
is meaningful only after the estimable parameter functions have been verified
to annihilate the exact kernel and a loss function has been specified on the
identifiable tangent space.

Spectral aggregation changes the statistical model.  For resolved components
with positive parameter-independent active weights, the kernel of the
aggregate information is exactly the intersection of the conditional kernels.
Broadband data can therefore restore local rank when the conditional null
directions rotate with energy or channel.  Retaining a spectral label and
tracing it out are nevertheless different statistical experiments related by
data processing.  Moreover, restored rank does not guarantee useful precision:
numerical rank must be reported with a declared coordinate scaling and
floating-point tolerance, while practical information strength requires
scaled eigenvalues, condition numbers, or an explicitly defined effective-rank
threshold.

Identifiability alone does not guarantee joint attainability.  The imaginary
part of the quantum geometric tensor diagnoses local SLD incompatibility.  Our
exact pure-state Holevo results and tolerance-qualified numerical
primal--dual brackets for mixed states show that an invertible QFIM can
substantially understate the attainable weighted multiparameter cost.  The
numerical brackets are supported by primal and dual iterates satisfying the
declared floating-point feasibility and residual tolerances; they are not
interval-arithmetic or exact-arithmetic proofs of their displayed decimal
endpoints.  The relevant scalar bound must specify the identifiable tangent
space, weight matrix, local-unbiasedness conditions, and asymptotic measurement
class.  Likewise, under the declared parameter-independent measurement and
detector maps, flavor projection, energy smearing, inefficiency, backgrounds,
Poisson sampling, and nuisance profiling separate state $\QFI$, flavor
$\CFI$, and reconstructed-event information.  These layers cannot be
interchanged without changing the statistical experiment.

The independent calculation reproduces the source result at its stated monochromatic vacuum
benchmark but changes its interpretation.  The pure-qutrit theorem imposes
structural rank at most four, and the full-precision calculation saturates that
bound numerically, with rank four and nullity two at the declared scaling and
tolerance.  The matrix reconstructed from the printed decimal entries is
indefinite, and its ordinary inverse is not a valid simultaneous
six-parameter Cramér--Rao bound.  Separate inversions of the angle and mass
blocks describe fixed-complement submodels rather than a regular
six-coordinate model, and invertibility of either block does not establish
multiparameter attainability.

The independent public-DUNE implementation then supplies a deliberately
scoped experimental case study.  It reproduces all 640 signal/background
reference entries from the checksum-locked public configuration with relative
$L_2$ error $3.9\times10^{-16}$ and reproduces the nuisance-profiled local
GLoBES likelihood curvature with relative error $3.3\times10^{-6}$.  The
configuration contains nine declared shared normalization coordinates but
only eight active event-rate directions: the ninth is attached to eight
declared tau-background channels whose predicted rates vanish for the supplied
cross-section tables.  Across the documented 176-record physics envelope, the
264-bin reconstructed-event information has numerical rank six at relative
tolerance $10^{-10}$ after the declared dimensionless coordinate scaling,
whereas replacing the four spectra by four rule totals gives numerical rank
four.  At the weakest selected record, however, the event information has
effective rank five at threshold $10^{-6}$ and local angular widths of order
radians.  The result therefore demonstrates numerical spectral rank
restoration at the declared tolerance and horn-polarity and
appearance/disappearance complementarity within the simplified public
configuration, not official DUNE precision or a surrogate for the
collaboration's full likelihood and systematic covariance.

Finally, the conditional nonlocal likelihood study exposes effects that no
local Fisher matrix can determine.  Its finite candidate bank exhibits
coarse-grid overcoverage, and the reported coverage depends materially on
whether the Gaussian auxiliary-measurement centers are conditioned or
fluctuated.  The expanded four-coordinate Asimov profile finds a lower
wrong-ordering deviance than the two-coordinate bank, but that change combines
finer optimization with two additional profiled coordinates.  Its optimum
reaches the upper $\theta_{13}$ boundary, and the result remains conditional on
the bounded domain, auxiliary constraints, and fixed solar coordinates.
These calculations are
diagnostic stress tests, not coverage-calibrated ordering significances or
global confidence constructions.

Global experimental claims would require continuous or demonstrably converged
adaptive profiling over all materially correlated coordinates in both
orderings, multiple regular and degenerate truth points, a prospectively
defined auxiliary-measurement ensemble and boundary treatment, and a
predeclared confidence construction tested with enough pseudoexperiments to
reach a stated Monte Carlo precision.  They would additionally require a
collaboration-grade systematic model with correlated shape uncertainties and
near-to-far constraints.  Until those requirements are met, the Fisher
spectra, conditional candidate bank, and expanded four-coordinate Asimov
profile should be interpreted as local and diagnostic quantities.

Taken together, the results establish a reproducible route from
representation-invariant quantum geometry to experimentally accessible local
information while making the failure modes explicit.  The individual
ingredients have important precedents; the contribution here is their
validated integration into a single workflow that distinguishes exact
theorems and analytic results, tolerance-qualified numerical computations,
scoped public-configuration evidence, and diagnostic simulation.  That
workflow provides a disciplined basis for comparing quantum information
bounds, measurement choices, spectral complementarity, and detector
likelihoods in neutrino oscillations and in other multiparameter
quantum-sensing problems.

\appendix

\section{PMNS conventions and analytic derivatives}
\label{app:pmns-derivatives}

This appendix fixes the phase, index, and differentiation conventions used in
all state-, measurement-, and event-level calculations.  The
oscillation-coordinate vector is
\begin{equation}
 \bm{\lambda}
 =
 \bigl(
   \theta_{12},
   \theta_{13},
   \theta_{23},
   \dcp,
   \Delta m_{21}^{2},
   \Delta m_{31}^{2}
 \bigr)^{\mathsf T},
 \label{eq:app-coordinate-order}
\end{equation}
where the angles and $\dcp$ are in radians and the mass-squared splittings
are in $\mathrm{eV}^{2}$.  Flavor indices
$\alpha,\beta\in\{e,\mu,\tau\}$ label rows and mass indices
$i,j\in\{1,2,3\}$ label columns.  Majorana phases are omitted because they
cancel from oscillation amplitudes.

\subsection{Mixing-matrix convention}

Writing $s_{ij}=\sin\theta_{ij}$ and
$c_{ij}=\cos\theta_{ij}$, we use the ordered factorization
\begin{equation}
 U(\bm{\lambda})=R_{23}U_{13}(\dcp)R_{12},
 \label{eq:app-pmns-factorization}
\end{equation}
with
\begin{align}
 R_{12}&=
 \begin{pmatrix}
  c_{12}&s_{12}&0\\
  -s_{12}&c_{12}&0\\
  0&0&1
 \end{pmatrix},
 &
 U_{13}&=
 \begin{pmatrix}
  c_{13}&0&s_{13}e^{-i\dcp}\\
  0&1&0\\
  -s_{13}e^{i\dcp}&0&c_{13}
 \end{pmatrix},
 \nonumber\\[3pt]
 R_{23}&=
 \begin{pmatrix}
  1&0&0\\
  0&c_{23}&s_{23}\\
  0&-s_{23}&c_{23}
 \end{pmatrix}.
 \label{eq:app-pmns-factors}
\end{align}
Thus, for example,
$U_{e3}=s_{13}e^{-i\dcp}$,
$U_{\mu3}=s_{23}c_{13}$, and
\begin{equation}
 U_{\mu1}=-s_{12}c_{23}
             -c_{12}s_{23}s_{13}e^{i\dcp},
 \qquad
 U_{\tau2}=-c_{12}s_{23}
             -s_{12}c_{23}s_{13}e^{i\dcp}.
 \label{eq:app-pmns-entry-checks}
\end{equation}
These entries make the row, column, and CP-phase conventions unambiguous.

The absolute mass origin is removed as an overall propagation phase.  We set
\begin{equation}
 M^{2}=\operatorname{diag}
 \bigl(0,\Delta m_{21}^{2},\Delta m_{31}^{2}\bigr),
 \label{eq:app-mass-convention}
\end{equation}
so $\Delta m_{31}^{2}>0$ denotes normal ordering and
$\Delta m_{31}^{2}<0$ denotes inverted ordering in the coordinate convention
used here.

\subsection{Exact derivatives of the PMNS matrix}

Let $U_a\equiv\partial U/\partial a$.  Differentiating the factors in
Eq.~\eqref{eq:app-pmns-factorization} gives
\begin{align}
 R'_{12}&=
 \begin{pmatrix}
  -s_{12}&c_{12}&0\\
  -c_{12}&-s_{12}&0\\
  0&0&0
 \end{pmatrix},
 &
 R'_{23}&=
 \begin{pmatrix}
  0&0&0\\
  0&-s_{23}&c_{23}\\
  0&-c_{23}&-s_{23}
 \end{pmatrix},
 \nonumber\\[3pt]
 U_{13,\theta_{13}}&=
 \begin{pmatrix}
  -s_{13}&0&c_{13}e^{-i\dcp}\\
  0&0&0\\
  -c_{13}e^{i\dcp}&0&-s_{13}
 \end{pmatrix},
 &
 U_{13,\dcp}&=
 \begin{pmatrix}
  0&0&-i s_{13}e^{-i\dcp}\\
  0&0&0\\
  -i s_{13}e^{i\dcp}&0&0
 \end{pmatrix}.
 \label{eq:app-factor-derivatives}
\end{align}
The four nonzero mixing-matrix derivatives are therefore
\begin{align}
 U_{\theta_{12}}&=R_{23}U_{13}R'_{12},
 &
 U_{\theta_{13}}&=R_{23}U_{13,\theta_{13}}R_{12},
 \nonumber\\
 U_{\theta_{23}}&=R'_{23}U_{13}R_{12},
 &
 U_{\dcp}&=R_{23}U_{13,\dcp}R_{12},
 \label{eq:app-pmns-derivatives}
\end{align}
while
\begin{equation}
 U_{\Delta m_{21}^{2}}=U_{\Delta m_{31}^{2}}=0.
 \label{eq:app-pmns-mass-derivatives}
\end{equation}
Unitarity supplies a useful analytic and numerical check:
\begin{equation}
 U^{\dagger}U=\mathbb{I},
 \qquad
 U^{\dagger}U_a+U_a^{\dagger}U=0.
 \label{eq:app-tangent-unitarity}
\end{equation}
Consequently,
\begin{equation}
 \Gamma_a
 \equiv
 U^{\dagger}U_a
 \label{eq:app-pmns-moving-connection}
\end{equation}
is anti-Hermitian and is the moving-basis connection that appears when the
PMNS basis is treated as a passive parameter-dependent representation.

In vacuum, the antineutrino convention is
\begin{equation}
 \overline U(\bm{\lambda})
 =
 U(\bm{\lambda})^{*}
 =
 U(\mathsf{R}\bm{\lambda}),
 \qquad
 \mathsf{R}
 \equiv
 J_{\mathrm{CP}}
 =
 \operatorname{diag}(1,1,1,-1,1,1),
 \label{eq:app-cp-matrix-map}
\end{equation}
where $\mathsf R$ is the Jacobian of the coordinate reflection
$\dcp\mapsto-\dcp$.  Differentiating the composed map gives the chain-rule
pullback
\begin{equation}
 \partial_i\overline U(\bm{\lambda})
 =
 \sum_j
 \mathsf{R}_{ji}
 U_j(\mathsf{R}\bm{\lambda}).
 \label{eq:app-cp-derivative-map}
\end{equation}
In particular, the $\dcp$ derivative acquires the additional minus sign.
Constant-density antineutrino propagation also reverses the matter potential;
that extension is treated separately in
Appendix~\ref{app:matter-frechet}.

\subsection{Vacuum amplitudes and probability derivatives}

For $L/E$ in $\mathrm{km/GeV}$, define
\begin{equation}
 \kappa=2(1.267)\frac{L}{E},
 \qquad
 D=\operatorname{diag}\!\left(
  1,
  e^{-i\kappa\Delta m_{21}^{2}},
  e^{-i\kappa\Delta m_{31}^{2}}
 \right).
 \label{eq:app-vacuum-phase}
\end{equation}
The flavor-space propagator and the state produced from source flavor
$\alpha$ are
\begin{equation}
 S=UDU^{\dagger},
 \qquad
 \lvert\psi^{(\alpha)}\rangle=S\lvert\nu_{\alpha}\rangle,
 \qquad
 \mathcal{A}_{\beta\alpha}=S_{\beta\alpha}
 =\sum_i U_{\beta i}D_iU^{*}_{\alpha i}.
 \label{eq:app-vacuum-amplitude}
\end{equation}
For $a\in\{\theta_{12},\theta_{13},\theta_{23},\dcp\}$, the exact
amplitude derivative is
\begin{equation}
 \partial_a\mathcal{A}_{\beta\alpha}
 =
 \sum_i\left[
  (U_a)_{\beta i}D_iU^{*}_{\alpha i}
  +U_{\beta i}D_i(U_a)^{*}_{\alpha i}
 \right].
 \label{eq:app-vacuum-angle-derivative}
\end{equation}
For $j=2,3$, the mass-splitting derivatives are instead
\begin{equation}
 \frac{\partial\mathcal{A}_{\beta\alpha}}
      {\partial(\Delta m_{j1}^{2})}
 =
 -i\kappa\,
 U_{\beta j}e^{-i\kappa\Delta m_{j1}^{2}}U^{*}_{\alpha j}.
 \label{eq:app-vacuum-mass-derivative}
\end{equation}
For antineutrinos one replaces $U$ by $U^{*}$, so that
$\overline S=U^{*}DU^{\mathsf{T}}$; the vacuum mass phases themselves are
unchanged.

Flavor probabilities and their derivatives with respect to any of the six
coordinates follow without numerical differentiation:
\begin{equation}
 P_{\alpha\to\beta}
 =
 \lvert\mathcal{A}_{\beta\alpha}\rvert^{2},
 \qquad
 \partial_i P_{\alpha\to\beta}
 =
 2\,\operatorname{Re}\!\left[
   \mathcal{A}_{\beta\alpha}^{*}
   \partial_i\mathcal{A}_{\beta\alpha}
 \right].
 \label{eq:app-probability-derivative}
\end{equation}
Equations~\eqref{eq:app-tangent-unitarity} and
\eqref{eq:app-probability-derivative} imply the conservation checks
\begin{equation}
2\operatorname{Re}
 \langle\psi^{(\alpha)}
 \vert\partial_i\psi^{(\alpha)}\rangle
 =0,
 \qquad
 \sum_{\beta}P_{\alpha\to\beta}=1,
 \qquad
 \sum_{\beta}\partial_iP_{\alpha\to\beta}=0.
 \label{eq:app-normalization-identities}
\end{equation}

\subsection{Numerical check}

Fourth-order centered five-point differences test the analytic PMNS, vacuum
state, and probability derivatives over 162 PMNS points and 108 propagation
configurations spanning both orderings, all source flavors, both polarities,
and three values of $L/E$.  The angular step is $10^{-4}$; the
mass-splitting steps are $10^{-7}\,\mathrm{eV}^2$ and
$10^{-6}\,\mathrm{eV}^2$ for $\Delta m^2_{21}$ and
$\Delta m^2_{31}$.  The largest guarded discrepancies are
$1.92\times10^{-12}$ for PMNS derivatives,
$1.17\times10^{-11}$ for propagated states, and
$3.64\times10^{-10}$ for probabilities.  PMNS unitarity, tangent unitarity,
state normalization, probability conservation, CP reflection, and
$2\pi$-periodicity all agree within $9.79\times10^{-16}$, except for the
finite-difference normalization-tangent residual
$4.55\times10^{-13}$.  Constant-density derivatives are checked separately
in Appendix~\ref{app:matter-frechet}.

\section{Support regularity and numerical Holevo certification}
\label{app:holevo-certification}

This appendix records the support and floating-point criteria that supplement
the Holevo definitions and semidefinite formulation in
Sec.~\ref{sec:attainability}.  The calculation is performed at a
support-regular point: the rank of $\rho$ is locally constant and, with
$P_0$ the projector onto $\ker\rho$,
\begin{equation}
  P_0(\partial_i\rho)P_0=0.
  \label{eq:app-support-regularity}
\end{equation}
In an eigenbasis $\rho=\sum_a p_a|a\rangle\langle a|$, the SLD is evaluated
from
\begin{equation}
  \langle a|L_i|b\rangle
  =
  \frac{2\langle a|\partial_i\rho|b\rangle}{p_a+p_b}
  \quad (p_a+p_b>0),
  \label{eq:app-sld-spectral-form}
\end{equation}
with unconstrained kernel--kernel entries set to zero.  Those entries do not
affect the SLD information or Holevo cost under
Eq.~\eqref{eq:app-support-regularity}.  Rank-changing points require a
separate nonregular analysis.

For a singular requested QFIM, no joint bound is assigned to the original
coordinate vector.  A numerical representative of the estimable quotient is
selected only after declaring a nonsingular scale matrix $S$.  If
$S^{\mathsf T}\QFI S=V_+\Lambda_+V_+^{\mathsf T}+V_0\Lambda_0V_0^{\mathsf T}$,
the retained Jacobian is $J_+=SV_+$, where $\Lambda_+$ contains only
eigenvalues above the declared relative threshold.  The reduced derivatives
and information are
\begin{equation}
  \partial_{\eta_a}\rho
  =
  \sum_i(J_+)_{ia}\partial_i\rho,
  \qquad
  \QFI_\eta
  =
  J_+^{\mathsf T}\QFI J_+
  =
  \Lambda_+.
  \label{eq:app-identifiable-jacobian}
\end{equation}
This chooses a numerical complement to the exact kernel; its orientation may
depend on $S$ and the tolerance, although the exact quotient tangent does
not.

For exactly feasible primal and dual conic points, weak duality gives
$C_{\mathrm{dual}}\leq C_{\mathrm H}\leq C_{\mathrm{primal}}$, which may be
tightened by $C_{\mathrm S}\leq C_{\mathrm H}\leq C_{\mathrm D}$.  No
floating-point objective is described as an exact bound without an exactly
feasible correction or interval-arithmetic verification.  Numerical bracket
eligibility is assessed after QFIM whitening: the maximum
local-unbiasedness and dual-affine residuals must not exceed $10^{-8}$;
the normalized primal and dual cone margins
$\lambda_{\min}(M)/\max\{\|M\|_2,1\}$ must be at least
$-10^{-10}$ and $-10^{-9}$, respectively; and
\[
  C_{\mathrm{primal}}-C_{\mathrm{dual}}
  \geq
  -10^{-9}\max\{|C_{\mathrm{primal}}|,1\}.
\]
QFIM-rank and state-support selection use relative thresholds $10^{-10}$
and $10^{-12}$.  The KKT residual is recorded as a convergence diagnostic,
not as an independent feasibility condition.

Across 200 three-coordinate solver attempts, 187 meet these criteria; the
best feasible attempt is retained for each model.  All 40 retained brackets
contain the independently derived pure-qutrit value $C_{\mathrm H}=5$, with
maximum width $1.50\times10^{-6}$, maximum local-unbiasedness residual
$9.66\times10^{-15}$, maximum dual-affine residual $2.57\times10^{-16}$,
and maximum KKT residual $2.61\times10^{-8}$.  The mixed-spectral solve
reported in Table~\ref{tab:mixed_spectral_holevo_certificate} has bracket
width $6.02\times10^{-6}$, local-unbiasedness residual
$2.62\times10^{-14}$, and KKT residual $1.08\times10^{-9}$.  These are
tolerance-qualified numerical brackets, not exact-arithmetic certificates.

\section{Constant-density matter evolution and Fr\'echet derivatives}
\label{app:matter-frechet}

This appendix gives the constant-density propagation and derivative formulas
used throughout the analytic and controlled numerical calculations.  The
construction preserves the PMNS and mass-splitting conventions of
Appendix~\ref{app:pmns-derivatives}, differentiates the matrix exponential
without assuming commutativity, and implements the CP pullback of
Sec.~\ref{sec:geometry} with the required reversal of the matter potential.
The checksum-locked public-DUNE comparison instead uses the separately
declared numerical constants compiled into GLoBES 3.2.18, as documented in
Appendix~\ref{app:dune-configuration}; the two numerical-unit conventions are
not mixed within a calculation.  Baseline, energy, density, and electron
fraction are treated as declared experimental controls unless they are
explicitly promoted to nuisance coordinates below.

\subsection{Effective mass-squared Hamiltonian and unit convention}

Let
\begin{equation}
  \bm\lambda
  =
  \bigl(
    \theta_{12},
    \theta_{13},
    \theta_{23},
    \dcp,
    \Delta m_{21}^{2},
    \Delta m_{31}^{2}
  \bigr)^{\mathsf T},
  \qquad
  M_{0}^{2}
  =
  \operatorname{diag}
  \bigl(
    0,
    \Delta m_{21}^{2},
    \Delta m_{31}^{2}
  \bigr),
  \label{eq:appD-parameter-vector}
\end{equation}
where the unobservable common mass-squared offset has been removed.  We use a
polarity label $s=+1$ for neutrinos and $s=-1$ for antineutrinos and
define
\begin{equation}
  U_s(\bm\lambda)
  =
  \begin{cases}
    U(\bm\lambda), & s=+1,\\
    U(\bm\lambda)^{*}, & s=-1.
  \end{cases}
  \label{eq:appD-polarity-mixing}
\end{equation}

For constant electron density, the flavor-basis Hamiltonian is
\begin{align}
  H_s
  &=
  \frac{1}{2E}\,\mathsf M_s^{2},
  \\
  \mathsf M_s^{2}
  &=
  U_s M_0^{2}U_s^{\dagger}
  +s\,a\,\Pi_e,
  \qquad
  \Pi_e
  =
  \operatorname{diag}(1,0,0),
  \label{eq:appD-effective-mass}
\end{align}
with the positive physical matter-potential magnitude
\begin{equation}
  a=2E V_{\mathrm{CC}}.
\end{equation}
For clarity, define the dimensionless numerical controls
\begin{equation}
  \widehat E
  \equiv
  \frac{E}{\mathrm{GeV}},
  \qquad
  \widehat\rho
  \equiv
  \frac{\rho}{\mathrm{g\,cm}^{-3}}.
  \label{eq:appD-dimensionless-matter-controls}
\end{equation}
Then
\begin{equation}
  a
  =
  c_a\,Y_e\,\widehat\rho\,\widehat E\,
  \mathrm{eV}^{2},
  \qquad
  c_a=1.526\times10^{-4}.
  \label{eq:appD-matter-potential}
\end{equation}
Thus physical antineutrino propagation requires both
$U\mapsto U^{*}$ and reversal of the signed matter contribution,
$a\mapsto-a$.  Applying only the first operation does not give the
CP-conjugate matter Hamiltonian.

For later CP identities, it is useful to extend the notation algebraically.
Whenever a second matter argument $q\in\mathbb{R}$ is displayed explicitly,
we define
\begin{equation}
  \mathsf M_s^{2}(\bm\lambda;q)
  \equiv
  U_s(\bm\lambda)M_0^{2}U_s(\bm\lambda)^{\dagger}
  +s\,q\,\Pi_e.
  \label{eq:appD-signed-matter-extension}
\end{equation}
The physical model uses $q=a>0$.  Negative $q$ in a pullback formula is a
signed algebraic continuation used to express the reversal of the matter
potential; it is not a second physical definition of the positive magnitude
$a$.

For $L$ in km, $E$ in GeV, and mass-squared quantities in
$\mathrm{eV}^{2}$, we write
\begin{equation}
  \kappa(L,E)
  =
  c_{\phi}\frac{L}{E},
  \qquad
  c_{\phi}=2(1.267),
  \qquad
  A_s=-i\kappa\mathsf M_s^{2},
  \label{eq:appD-generator}
\end{equation}
so that
\begin{equation}
  S_s(L,E;\bm\lambda,\rho,Y_e)
  =
  e^{A_s},
  \qquad
  \lvert\psi_{s,\alpha}\rangle
  =
  S_s\lvert\nu_\alpha\rangle.
  \label{eq:appD-propagator}
\end{equation}
The factor $c_{\phi}=2(1.267)$ is the amplitude-phase coefficient; it gives
the usual $1.267\,\Delta m^2L/E$ argument after an oscillation probability is
formed.  Equations~\eqref{eq:appD-matter-potential} and
\eqref{eq:appD-generator} define the rounded paper-wide convention used for
the analytic and controlled benchmarks.  The public-DUNE external comparison
uses the distinct GLoBES constants declared in
Eq.~\eqref{eq:appE-globes-constants}.

\subsection{Exact Hamiltonian derivatives}

All derivatives with respect to the six oscillation coordinates are taken at
fixed baseline, energy, density, and electron fraction.  For any mixing
coordinate
$i\in\{\theta_{12},\theta_{13},\theta_{23},\dcp\}$,
\begin{equation}
  \partial_i\mathsf M_s^{2}
  =
  (\partial_iU_s)M_0^{2}U_s^{\dagger}
  +
  U_sM_0^{2}(\partial_iU_s)^{\dagger}.
  \label{eq:appD-angle-hamiltonian-derivatives}
\end{equation}
For antineutrinos,
\begin{equation}
  \partial_iU_s=(\partial_iU)^{*}.
\end{equation}
The mass-splitting directions are
\begin{align}
  \partial_{\Delta m_{21}^{2}}\mathsf M_s^{2}
  &=
  U_s\operatorname{diag}(0,1,0)U_s^{\dagger},
  \\
  \partial_{\Delta m_{31}^{2}}\mathsf M_s^{2}
  &=
  U_s\operatorname{diag}(0,0,1)U_s^{\dagger}.
  \label{eq:appD-mass-hamiltonian-derivatives}
\end{align}
Every matrix in Eqs.~\eqref{eq:appD-angle-hamiltonian-derivatives} and
\eqref{eq:appD-mass-hamiltonian-derivatives} is Hermitian, as required for a
tangent to a Hermitian Hamiltonian family.  Because $\kappa$ is fixed under
these six derivatives,
\begin{equation}
  \partial_iA_s
  =
  -i\kappa\,\partial_i\mathsf M_s^{2}.
  \label{eq:appD-generator-derivatives}
\end{equation}

Although $\rho$ and $Y_e$ are not part of the standard six-coordinate
model, they can be promoted to nuisance coordinates without numerical
differentiation.  At fixed $L$ and $E$, differentiation with respect to
the dimensional density coordinate $\rho$, measured in
$\mathrm{g\,cm}^{-3}$, and the dimensionless electron fraction $Y_e$
gives
\begin{align}
  \partial_\rho\mathsf M_s^{2}
  &=
  s\,c_aY_e\widehat E\,
  \frac{\mathrm{eV}^{2}}{\mathrm{g\,cm}^{-3}}\,
  \Pi_e,
  &
  \partial_{Y_e}\mathsf M_s^{2}
  &=
  s\,c_a\widehat\rho\,\widehat E\,
  \mathrm{eV}^{2}\,
  \Pi_e.
  \label{eq:appD-matter-control-derivatives}
\end{align}
These expressions are equivalent to differentiating
Eq.~\eqref{eq:appD-matter-potential} in the declared numerical-unit
convention.  These directions can therefore be included in the same
Fr\'echet calculation and, if desired, in a subsequent nuisance Schur
complement.

\subsection{Fr\'echet derivative of the propagator}

For general noncommuting square matrices $A$ and $E$ of the same
dimension, the Fr\'echet derivative of the matrix exponential is
\cite{Higham2008,AlMohyHigham2009}
\begin{equation}
  \mathcal L_{\exp}(A,E)
  \equiv
  \left.
  \frac{\mathrm d}{\mathrm d\epsilon}
  e^{A+\epsilon E}
  \right|_{\epsilon=0}
  =
  \int_{0}^{1}
  e^{(1-t)A}E e^{tA}\,
  \mathrm dt.
  \label{eq:appD-frechet-integral}
\end{equation}
Consequently,
\begin{equation}
  \partial_iS_s
  =
  \mathcal L_{\exp}(A_s,\partial_iA_s),
  \qquad
  \lvert\partial_i\psi_{s,\alpha}\rangle
  =
  \mathcal L_{\exp}(A_s,\partial_iA_s)
  \lvert\nu_\alpha\rangle.
  \label{eq:appD-state-frechet}
\end{equation}
No replacement of
$\mathcal L_{\exp}(A_s,\partial_iA_s)$ by
$e^{A_s}\partial_iA_s$ is made: that simplification is valid only when
\begin{equation}
  [A_s,\partial_iA_s]=0.
\end{equation}

An algebraically independent representation, used for numerical verification, is
the block-exponential identity
\begin{equation}
  \exp\!\begin{pmatrix}
    A&E\\
    0&A
  \end{pmatrix}
  =
  \begin{pmatrix}
    e^A&\mathcal L_{\exp}(A,E)\\
    0&e^A
  \end{pmatrix}.
  \label{eq:appD-block-frechet}
\end{equation}
For the $3\times3$ matrices used here, the left-hand side is a
$6\times6$ matrix exponential.  Its upper-right $3\times3$ block extracts
the derivative and therefore checks the production Fr\'echet routine without
calling that routine again.

For completeness, the same derivative has a continuous divided-difference
form.  Let
\begin{equation}
  V^{\dagger}\mathsf M_s^{2}V
  =
  \operatorname{diag}(\mu_1,\mu_2,\mu_3),
  \qquad
  B_i
  =
  V^{\dagger}(\partial_i\mathsf M_s^{2})V.
  \label{eq:appD-matter-eigendecomposition}
\end{equation}
Then
\begin{equation}
  V^{\dagger}(\partial_iS_s)V
  =
  F\mathbin{\circ}B_i,
  \label{eq:appD-divided-difference}
\end{equation}
where $\circ$ denotes the Hadamard product and
\begin{equation}
  F_{rs}
  =
  \begin{cases}
  \displaystyle
  \frac{
    e^{-i\kappa\mu_r}-e^{-i\kappa\mu_s}
  }{
    \mu_r-\mu_s
  },
  & r\ne s,\\[9pt]
  -i\kappa e^{-i\kappa\mu_r},
  & r=s.
  \end{cases}
  \label{eq:appD-divided-difference-kernel}
\end{equation}
The off-diagonal expression tends continuously to the diagonal one as
$\mu_s\to\mu_r$; an exact or near degeneracy therefore does not invalidate
the derivative.  The production calculation uses
Eq.~\eqref{eq:appD-state-frechet} rather than differentiating eigenvectors,
thereby avoiding arbitrary phases and basis choices inside degenerate
eigenspaces.

\subsection{Unitarity, probabilities, and limiting identities}

Hermiticity of $\mathsf M_s^{2}$ makes $A_s$ anti-Hermitian, and hence
\begin{equation}
  S_s^{\dagger}S_s=\mathbb{I},
  \qquad
  (\partial_iS_s)^{\dagger}S_s
  +
  S_s^{\dagger}(\partial_iS_s)
  =
  0.
  \label{eq:appD-unitarity-tangent}
\end{equation}
For every source flavor this implies
\begin{equation}
  \langle\psi_{s,\alpha}\vert\psi_{s,\alpha}\rangle
  =
  1,
  \qquad
  2\operatorname{Re}
  \langle\psi_{s,\alpha}\vert
  \partial_i\psi_{s,\alpha}\rangle
  =
  0.
  \label{eq:appD-state-tangent}
\end{equation}
The flavor probabilities and their exact derivatives are
\begin{align}
  P_{\alpha\to\beta}^{(s)}
  &=
  \left|
    \langle\nu_\beta\vert\psi_{s,\alpha}\rangle
  \right|^2,
  \\
  \partial_iP_{\alpha\to\beta}^{(s)}
  &=
  2\operatorname{Re}\!\left[
    \langle\psi_{s,\alpha}\vert\nu_\beta\rangle
    \langle\nu_\beta\vert
    \partial_i\psi_{s,\alpha}\rangle
  \right],
  \label{eq:appD-probability-derivative}
\end{align}
and obey
\begin{equation}
  \sum_\beta P_{\alpha\to\beta}^{(s)}
  =
  1,
  \qquad
  \sum_\beta\partial_iP_{\alpha\to\beta}^{(s)}
  =
  0.
  \label{eq:appD-probability-sum}
\end{equation}
At $\rho=0$, or equivalently $Y_e=0$, the matter term in
Eq.~\eqref{eq:appD-effective-mass} vanishes and the generator reduces to the
vacuum generator of Appendix~\ref{app:pmns-derivatives}.  Both the state and
every oscillation-parameter derivative must therefore agree with the analytic
vacuum expressions, not merely the probabilities.

The same construction extends directly to a piecewise-constant profile.  If
the path contains $N$ ordered layers and $S_k=e^{A_k}$, then
\begin{align}
  S_{N:1}
  &=
  S_N S_{N-1}\cdots S_1,
  \\
  \partial_iS_{N:1}
  &=
  \sum_{k=1}^{N}
  S_N\cdots S_{k+1}\,
  \mathcal L_{\exp}(A_k,\partial_iA_k)\,
  S_{k-1}\cdots S_1.
  \label{eq:appD-layered-profile}
\end{align}
Equation~\eqref{eq:appD-layered-profile} is an exact product rule.  The
numerical results reported below use one constant-density layer.

\subsection{CP reflection in matter}

In the adopted PMNS convention,
\begin{equation}
  U(\mathsf R\bm\lambda)
  =
  U(\bm\lambda)^{*},
  \qquad
  \mathsf R
  =
  J_{\mathrm{CP}}
  =
  \operatorname{diag}(1,1,1,-1,1,1).
  \label{eq:appD-cp-reflection}
\end{equation}
Using the signed matter argument defined in
Eq.~\eqref{eq:appD-signed-matter-extension}, and identifying the neutrino and
antineutrino flavor Hilbert spaces through their canonical flavor-coordinate
bases, the propagation-level CP identity is
\begin{equation}
  \mathsf M_{-}^{2}(\bm\lambda;a)
  =
  \mathsf M_{+}^{2}(\mathsf R\bm\lambda;-a),
  \qquad
  \lvert\psi_{-,\alpha}(\bm\lambda;a)\rangle
  =
  \lvert\psi_{+,\alpha}(\mathsf R\bm\lambda;-a)\rangle.
  \label{eq:appD-matter-cp-state}
\end{equation}
Here $a>0$ is the physical matter-potential magnitude on the antineutrino
side, whereas $-a$ is the signed argument representing the reversed
potential in the neutrino family.

Differentiating the composed family gives, with no sum on $i$,
\begin{equation}
  \lvert\partial_i\psi_{-,\alpha}(\bm\lambda;a)\rangle
  =
  \mathsf R_{ii}
  \left.
  \lvert
    \partial'_i
    \psi_{+,\alpha}(\bm\lambda';-a)
  \rangle
  \right|_{\bm\lambda'=\mathsf R\bm\lambda}.
  \label{eq:appD-matter-cp-derivative}
\end{equation}
Thus the projected quantum geometric tensor $\mathcal G$, the QFIM
$\QFI$, and the curvature matrix $\mathcal D$ satisfy
\begin{align}
  \overline{\mathcal G}(\bm\lambda;a)
  &=
  \mathsf R^{\mathsf T}
  \mathcal G(\mathsf R\bm\lambda;-a)
  \mathsf R,
  \\
  \overline{\QFI}(\bm\lambda;a)
  &=
  \mathsf R^{\mathsf T}
  \QFI(\mathsf R\bm\lambda;-a)
  \mathsf R,
  \\
  \overline{\mathcal D}(\bm\lambda;a)
  &=
  \mathsf R^{\mathsf T}
  \mathcal D(\mathsf R\bm\lambda;-a)
  \mathsf R.
  \label{eq:appD-matter-cp-tensors}
\end{align}
The sign reversal of entries carrying exactly one $\dcp$ index comes from
the coordinate Jacobian $\mathsf R$.  The separate replacement
$a\to-a$ comes from the physical matter Hamiltonian.  These operations are
logically distinct and both are required.

\subsection{Numerical check}

The constant-density implementation is tested at six parameter points, all
three source flavors, both polarities, and four propagation settings:
$(L/\mathrm{km},E/\mathrm{GeV},\rho/(\mathrm{g\,cm}^{-3}),Y_e)$ equal to
$(295,0.6,2.600,0.50)$, $(810,2.0,2.840,0.50)$,
$(1300,2.5,2.848,0.50)$, and $(1300,5.0,3.200,0.48)$.
Fourth-order finite differences use angular steps $10^{-4}$,
mass-splitting steps $10^{-7}$ and $10^{-6}\,\mathrm{eV}^2$, a density
step $10^{-3}\,\mathrm{g\,cm}^{-3}$, and an electron-fraction step
$10^{-4}$.  The largest discrepancies are
$3.09\times10^{-14}$ against the independent block-exponential derivative,
$1.08\times10^{-11}$ for state derivatives, and
$1.52\times10^{-10}$ for probability derivatives.  The zero-density vacuum
limit agrees within $8.13\times10^{-16}$, and the matter-CP state,
derivative, QFIM, and curvature pullbacks agree to binary floating-point
precision.  As a negative control, omitting the matter-sign reversal produces
relative errors as large as $1.394$, $1.813$, and $1.777$ for the state,
QFIM, and curvature, respectively.

\section{Public-DUNE parser, channels, systematics, and source hashes}
\label{app:dune-configuration}

This appendix documents the exact public DUNE configuration used in the case
study, the independent parser that converts its AEDL definitions into the
forward event model, and the provenance checks that bind every public-DUNE
numerical result to exact source bytes.  The inputs are the ancillary files
accompanying \emph{Experiment Simulation Configurations Approximating DUNE
TDR}, version 2 \cite{DUNEPublicConfig2021}, and the external validation uses
GLoBES 3.2.18 \cite{HuberEtAl2005,KoppEtAl2007}.  The DUNE release describes
these files as a simplified summary for phenomenological studies: they contain
far-detector response information and simple normalization uncertainties, but
not the complete TDR near-detector and systematic treatment.  Accordingly,
every result in this work is labeled a result for the public configuration and
not an official DUNE sensitivity.

\subsection{Acquisition and byte-level provenance}
\label{app:dune-provenance}

The acquisition script downloads the versioned arXiv source archive
\begin{center}
  \small
  \nolinkurl{https://arxiv.org/e-print/2103.04797v2}
\end{center}
and verifies its SHA-256 digest before extraction:
\begin{center}
  \small
  \nolinkurl{574975ab4d1b7e77eb3e99dca0ae309050ff54145bcc380388fa805aef636824}
\end{center}
Extraction is rejected if an archive member has an absolute path, contains a
parent-directory component, or is a symbolic link, hard link, or device.  The
analysis then uses only the subdirectory
\texttt{anc/dune\_globes}.  The reproducibility bundle does not redistribute
the raw ancillary files; it contains the verified acquisition script and the
complete hash manifest needed to reacquire and check them.

Let $\mathcal F$ be the 64 regular files in the configuration subdirectory
and $h(f)$ the hexadecimal SHA-256 digest of file $f$.  In addition to
retaining every individual digest, we define a canonical manifest fingerprint
\begin{equation}
  h_{\mathrm{man}}
  =
  \operatorname{SHA256}\!\left(
    \mathop{\Vert}_{f\in\operatorname{sort}(\mathcal F)}
    \left[
      \operatorname{path}(f)
      \mathbin{\Vert}\mathtt{0x00}
      \mathbin{\Vert}h(f)
      \mathbin{\Vert}\mathtt{0x0a}
    \right]
  \right),
  \label{eq:appE-manifest-definition}
\end{equation}
where the files are ordered lexicographically by relative POSIX path, the paths
and hexadecimal digests are encoded in UTF-8, and
$\mathbin{\Vert}$ denotes byte concatenation.  Thus each serialized manifest
entry consists of the relative path, one null byte, the lowercase hexadecimal
digest, and one newline byte.  For the retained configuration,
\begin{center}
  \small
  $\displaystyle h_{\mathrm{man}}={}$%
  \nolinkurl{ba31245a2f832e6dbc2f47cb42e25df17a15292365fadf2d9d85858ffd28a054}
\end{center}
The archive digest agrees among the acquisition script and both provenance
records.  All 64 independently recomputed file digests agree with the
manifest used in the GLoBES validation.  The file inventory
is summarized in Table~\ref{tab:appE-file-inventory}; the machine-readable
supplement gives all relative paths, sizes, and hashes.

\begin{table}[t]
  \centering
  \caption{
  Files in the checksum-locked \texttt{anc/dune\_globes} subdirectory.
  Include files are counted with the data directory that owns them, except for
  the three top-level AEDL files.
  }
  \label{tab:appE-file-inventory}
  \begin{tabular}{@{}lr@{}}
    \toprule
    Source category & Number of files \\
    \midrule
    Top-level AEDL configuration and includes & 3 \\
    Flux definitions and tables                & 3 \\
    Cross-section definitions and tables       & 3 \\
    Smearing definitions and matrices          & 17 \\
    Efficiency definitions and vectors         & 38 \\
    \midrule
    Total                                       & 64 \\
    \bottomrule
  \end{tabular}
\end{table}

\subsection{Strict AEDL-subset parser}
\label{app:dune-parser}

The production reader is intentionally not a general AEDL interpreter and does
not call or link against GLoBES.  It supports only the constructs exercised by
the checksum-locked DUNE files:
\begin{enumerate}
  \item C-style block and line comments, scalar constants, scalar assignments,
        and vector assignments;
  \item included FHC and RHC flux definitions and their seven-column tables;
  \item the charged-current and neutral-current cross-section tables;
  \item compact manual-smearing rows, including their declared true-bin index
        ranges;
  \item nested efficiency includes and post-smearing efficiency vectors;
  \item channel environments, including horn polarity, initial and final
        flavors, interaction type, and the optional \texttt{NOSC} flag; and
  \item rule environments, analysis windows, channel membership, systematic
        names, and their fractional normalization errors.
\end{enumerate}
The parser fails closed if a required assignment is missing, a table is ragged
or nonfinite, energy coordinates are not ordered, a smearing range is invalid,
a channel reference is unresolved, or a parsed object count differs from the
frozen contract.  In particular, it requires two flux definitions, two
cross-section tables, 16 smearing matrices, 30 efficiency vectors, 32
channels, and the four rules in their declared order.

The public-DUNE path uses the constants compiled into GLoBES 3.2.18 rather
than the rounded paper-wide vacuum coefficient:
\begin{align}
  c_{\phi}^{\mathrm{GLoBES}}
  &=
  \frac{1}{2(1.97327\times10^{-10})(10^{9})}
  =
  2.533865107157155,
  \\
  c_a^{\mathrm{GLoBES}}
  &=
  2(10^9)(7.63247\times10^{-14})
  =
  1.526494\times10^{-4}.
  \label{eq:appE-globes-constants}
\end{align}
Thus its propagation amplitude is
\begin{equation}
  \exp\!\left[
    -i c_{\phi}^{\mathrm{GLoBES}}
    \frac{L}{E}\mathsf M^2
  \right],
\end{equation}
while the effective matter term is
\begin{equation}
  a
  =
  c_a^{\mathrm{GLoBES}}Y_e\rho E
\end{equation}
in the GLoBES numerical-unit convention.  Keeping these constants separate
from the default $2(1.267)$ convention prevents a hidden convention
adjustment in the external comparison.

\subsection{Beam, binning, and detector response}
\label{app:dune-response}

The parsed scalar configuration is
\begin{equation}
  L=1284.9~\mathrm{km},
  \qquad
  \rho=2.848~\mathrm{g\,cm^{-3}},
  \qquad
  M_{\mathrm{fid}}=40~\mathrm{kt}.
  \label{eq:appE-dune-scalars}
\end{equation}
Both the FHC and RHC definitions use $6.5$ years, an exposure factor of
$11\times10^{20}$ protons on target per year, and the supplied far-detector
normalization $1.017718\times10^{17}$.  Each flux table has 501 energy rows.
The cross sections are interpolated linearly in $\log_{10}E$ according to
the GLoBES table convention, whereas the flux is interpolated linearly in
$E$.  Values outside the tabulated support are set to zero.

The detector definition contains 99 true-energy sampling bins and 80
reconstructed-energy bins, both spanning $0$--$110~\mathrm{GeV}$.  The
four rules use the declared reconstructed-energy window
\begin{equation}
  0.5
  \leq
  E_{\mathrm{rec}}
  \leq
  18~\mathrm{GeV}.
\end{equation}
Because selection is applied to bin centers, this retains 66 reconstructed
bins per rule and gives the joint $4\times66=264$-bin event vector.

For channel $c$, reconstructed bin $b$, and true-energy sampling bin $t$,
the independent rate engine evaluates
\begin{align}
  \mu_{cb}(\bm\lambda)
  ={}&
  \epsilon_{cb}
  \sum_t
  R^{(c)}_{bt}\,
  \Phi_{c}(E_t)\,
  \sigma_c(E_t)\,
  P_c(E_t;\bm\lambda)\,
  M_{\mathrm{fid}}\,
  \Delta E_t.
  \label{eq:appE-channel-rate}
\end{align}
Here $R^{(c)}_{bt}$ is the supplied manual true-to-reconstructed migration
matrix, and $\epsilon_{cb}$ is applied after smearing, as specified by the
AEDL channel.  The quantity $\Phi_c$ is the effective flux--exposure factor:
it includes the tabulated flux, exposure, supplied far-detector normalization,
and declared $L^{-2}$ factor.  For charged-current channels, $P_c$ is the
appropriate constant-density flavor-transition probability.  For
neutral-current channels marked \texttt{NOSC}, the supplied configuration
sets
\begin{equation}
  P_c=1.
\end{equation}

Every migration matrix has shape $80\times99$, contains only finite
nonnegative entries, and has active true-energy column sums equal to unity to
within $1.40\times10^{-5}$, the precision of the supplied decimal tables.
All 30 post-smearing efficiency vectors have length 80 and entries in
$[0,1]$.

\subsection{Channels and analysis rules}
\label{app:dune-channels}

Each channel fixes a beam definition, neutrino or antineutrino polarity,
initial and final flavor, interaction table, migration matrix, and efficiency
vector.  All references resolve, and each of the 32 channels belongs to
exactly one rule.  The channel families are given in
Table~\ref{tab:appE-channel-families}.  Each multiplicity-four family contains
both FHC and RHC definitions and both neutrino and antineutrino polarities.

\begin{table}[t]
  \centering
  \caption{
  Semantic partition of the 32 public-DUNE channels.  Each displayed
  neutrino transition also represents the corresponding CP-conjugate
  antineutrino channels included in the same multiplicity-four family.
  ``NOSC'' is the explicit no-oscillation prescription in the supplied
  neutral-current definitions.
  }
  \label{tab:appE-channel-families}
  \begin{tabular}{@{}llclc@{}}
    \toprule
    Sample & Role & Transition & Interaction & Multiplicity \\
    \midrule
    Appearance    & signal          & $\nu_\mu\to\nu_e$     & CC & 4 \\
    Appearance    & intrinsic beam  & $\nu_e\to\nu_e$       & CC & 4 \\
    Appearance    & muon-flavor     & $\nu_\mu\to\nu_\mu$  & CC & 4 \\
    Appearance    & tau-flavor      & $\nu_\mu\to\nu_\tau$ & CC & 4 \\
    Appearance    & neutral current & NOSC                     & NC & 4 \\
    Disappearance & signal          & $\nu_\mu\to\nu_\mu$  & CC & 4 \\
    Disappearance & tau-flavor      & $\nu_\mu\to\nu_\tau$ & CC & 4 \\
    Disappearance & neutral current & NOSC                     & NC & 4 \\
    \midrule
    Total         &                 &                          &    & 32 \\
    \bottomrule
  \end{tabular}
\end{table}

Equivalently, the declared channel set contains 16 FHC and 16 RHC definitions,
16 neutrino and 16 antineutrino polarities, 24 charged-current and eight
neutral-current channels, 24 oscillated and eight \texttt{NOSC} channels, and
eight signal plus 24 background channels.  As shown below, the eight declared
tau-background channels have zero rates for the supplied cross-section tables;
they remain part of these channel counts but do not contribute nonzero
predicted events.  The rule composition is frozen in
Table~\ref{tab:appE-rules}.

\begin{table}[t]
  \centering
  \caption{
  Public-DUNE analysis rules.  Signal and background entries count declared
  channels.  The selected-bin count follows from applying the common energy
  window to reconstructed-bin centers.
  }
  \label{tab:appE-rules}
  \begin{tabular}{@{}llllr@{}}
    \toprule
    AEDL rule & Beam/sample & Signal & Background & Selected bins \\
    \midrule
    \texttt{nue\_app}     & FHC appearance    & 2 & 8 & 66 \\
    \texttt{nuebar\_app}  & RHC appearance    & 2 & 8 & 66 \\
    \texttt{numu\_dis}    & FHC disappearance & 2 & 4 & 66 \\
    \texttt{numubar\_dis} & RHC disappearance & 2 & 4 & 66 \\
    \bottomrule
  \end{tabular}
\end{table}

Rule spectra are formed only after the full channel response has been applied:
\begin{equation}
  \mu_{rb}(\bm\lambda)
  =
  \sum_{c\in\mathcal C_r}
  \mu_{cb}(\bm\lambda),
  \label{eq:appE-rule-rate}
\end{equation}
where $\mathcal C_r$ contains the signal and background channels declared
for rule $r$.  Keeping the rule and reconstructed-bin labels gives the
264-bin model used in the Fisher and likelihood calculations; replacing each
rule by one total count is a distinct coarse-graining operation.

\subsection{Named normalization nuisances and correlations}
\label{app:dune-systematics}

Nine declared systematic names encode the public configuration's normalization
errors.  Let $g(c)$ map a channel occurrence to its AEDL systematic name, and
let $s_{g(c)}$ be the corresponding fractional error.  The linear rate model
is
\begin{equation}
  \mu_{rb}(\bm\lambda,\bm\eta)
  =
  \sum_{c\in\mathcal C_r}
  \left[
    1+s_{g(c)}\eta_{g(c)}
  \right]
  \mu_{cb}(\bm\lambda),
  \qquad
  \eta_k\sim\mathcal N(0,1).
  \label{eq:appE-nuisance-rate}
\end{equation}
A repeated name denotes one shared nuisance coordinate, not a new coordinate
at each occurrence.  The same pull therefore acts coherently on every channel
carrying that name, including occurrences in different rules and horn
polarities.  The complete correlation graph is summarized in
Table~\ref{tab:appE-systematics}.

\begin{table}[t]
  \centering
  \caption{
  Named normalization uncertainties in the public configuration.
  ``Occurrences'' counts channel assignments.  The active column indicates
  whether the resulting event-rate direction is nonzero for the supplied
  tables.
  }
  \label{tab:appE-systematics}
  \begin{tabular}{@{}lrrlc@{}}
    \toprule
    Name & Fraction & Occurrences & Rules & Active \\
    \midrule
    \texttt{err\_nue\_sig}     & 0.02 & 2 & FHC app.     & yes \\
    \texttt{err\_nue\_sigbar}  & 0.02 & 2 & RHC app.     & yes \\
    \texttt{err\_numu\_sig}    & 0.05 & 2 & FHC dis.     & yes \\
    \texttt{err\_numu\_sigbar} & 0.05 & 2 & RHC dis.     & yes \\
    \texttt{err\_nue\_bg}      & 0.05 & 2 & FHC app.     & yes \\
    \texttt{err\_nue\_bgbar}   & 0.05 & 2 & RHC app.     & yes \\
    \texttt{err\_numu\_bg}     & 0.05 & 8 & FHC/RHC app. & yes \\
    \texttt{err\_nc\_bgdis}    & 0.10 & 4 & FHC/RHC dis. & yes \\
    \texttt{err\_nutau\_bg}    & 0.20 & 8 & all rules    & no  \\
    \bottomrule
  \end{tabular}
\end{table}

\subsubsection{Declared but inactive tau-background mode}

The AEDL files declare eight
$\nu_\tau$ or $\bar\nu_\tau$ background channels and attach the common
$20\%$ uncertainty
\texttt{err\_nutau\_bg}.  However, the two tau-flavor columns of the supplied
charged-current cross-section table are identically zero.  Equation
\eqref{eq:appE-channel-rate} therefore gives
\begin{equation}
  \mu_{cb}=0
  \quad
  \text{for every declared tau-background channel},
  \qquad
  \partial_{\eta_{\mathrm{nutau}}}\bm\mu
  =
  \bm0.
  \label{eq:appE-inactive-tau-mode}
\end{equation}
Thus there are nine declared nuisance variables but eight active event-rate
directions in this public package.  Retaining the ninth coordinate is
harmless: its unit-Gaussian penalty is minimized uniquely at
$\eta_{\mathrm{nutau}}=0$, and it does not alter the profiled likelihood or
Fisher matrix.  Nevertheless, the distinction between a declared channel and
a nonzero event contribution is essential.  The tau-background definitions
must not be described as nonzero predicted backgrounds for this
configuration.

\subsubsection{Correlation negative control}

To test the repeated-name interpretation, we deliberately replaced the nine
named coordinates by 32 occurrence-wise independent Gaussian pulls.  Only 24
of those occurrence directions are active because the eight tau-channel rates
vanish.  Let
$\CFI_{\mathrm{prof}}^{\mathrm{named}}$ denote the correctly profiled Fisher
matrix with shared named nuisances and
$\CFI_{\mathrm{prof}}^{\mathrm{occ}}$ the deliberately incorrect
occurrence-wise result.  With the declared nonsingular coordinate scale matrix
$S$, define
\begin{equation}
  \widetilde{\CFI}_{\mathrm{prof}}^{\,x}
  =
  S^{\mathsf T}
  \CFI_{\mathrm{prof}}^{x}
  S,
  \qquad
  x\in\{\mathrm{named},\mathrm{occ}\}.
  \label{eq:appE-correlation-control-scaling}
\end{equation}
At the benchmark, the incorrect construction changes the scaled
profiled Fisher matrix by
\begin{equation}
  \frac{
    \left\|
      \widetilde{\CFI}_{\mathrm{prof}}^{\,\mathrm{occ}}
      -
      \widetilde{\CFI}_{\mathrm{prof}}^{\,\mathrm{named}}
    \right\|_{\mathrm F}
  }{
    \left\|
      \widetilde{\CFI}_{\mathrm{prof}}^{\,\mathrm{named}}
    \right\|_{\mathrm F}
  }
  =
  3.6593922\times10^{-2},
  \label{eq:appE-wrong-correlation-error}
\end{equation}
and gives
\begin{equation}
  \frac{
    \det\CFI_{\mathrm{prof}}^{\mathrm{occ}}
  }{
    \det\CFI_{\mathrm{prof}}^{\mathrm{named}}
  }
  =
  0.9793036.
  \label{eq:appE-wrong-correlation-determinant}
\end{equation}
The determinant ratio is unchanged if the same nonsingular coordinate scaling
is applied to both matrices.  The discrepancy is much larger than the
numerical validation tolerances and demonstrates that occurrence-wise
profiling defines a different statistical model.

\subsection{Configuration and external-engine validation}
\label{app:dune-parser-validation}

The validation covers the archive hash
chain, all individual file hashes, parser dimensions, bin geometry, response
normalization, channel references, rule membership, systematic errors,
correlation structure, and the inactive tau direction.  Independently
recomputed fixed and profiled Fisher matrices reproduce the stored reference
matrices with guarded error
$2.34\times10^{-19}$.

Small C exporters compiled against GLoBES 3.2.18 provide the independent
external comparison.  Table~\ref{tab:appE-validation} summarizes the retained,
provenance-locked results.  Spectra are compared for signal and background
separately across four rules and 80 reconstructed bins.  Derivatives use
fourth-order centered five-point differences in all six oscillation
coordinates.  For the likelihood comparison, we evaluate the profiled
\texttt{chiMultiExp} statistic.  One half of its numerical Hessian is then
compared with the independently assembled nuisance Schur complement.

\begin{table}[t]
  \centering
  \caption{
  Provenance, parser, and external-engine validation.  The manifest hash is
  reported in the text because it is a categorical byte-level check.
  }
  \label{tab:appE-validation}
  \begin{tabular}{@{}p{0.72\linewidth}r@{}}
    \toprule
    Diagnostic & Result \\
    \midrule
    Source files / hash mismatches
      & $64/0$ \\
    Maximum active smearing-column normalization error
      & $1.40\times10^{-5}$ \\
    Independent versus stored fixed/profiled Fisher matrices
      & $2.34\times10^{-19}$ \\
    Independent GLoBES spectrum relative $L_2$ error
      & $3.87\times10^{-16}$ \\
    Largest GLoBES five-point derivative relative $L_2$ error
      & $2.32\times10^{-9}$ \\
    Profiled GLoBES Hessian versus Schur Fisher relative $L_2$ error
      & $3.26\times10^{-6}$ \\
    Incorrect occurrence-wise nuisance-model Fisher error
      & $3.659\times10^{-2}$ \\
    \bottomrule
  \end{tabular}
\end{table}

These checks validate the source-to-event implementation used in the public
case study: exact source bytes, parser semantics, response ordering,
channel/rule aggregation, and the declared nuisance correlations are all
fixed and independently testable.  They do not enlarge the scope of the input
model.  In particular, the configuration still contains simplified
normalization systematics, no explicit near-detector likelihood, and no basis
for labeling the resulting precision numbers as official DUNE sensitivity.

\section{Exact likelihood profiling and local curvature}
\label{app:dune-likelihood}

This appendix gives the public-configuration likelihood, proves uniqueness of
the normalization-nuisance profile, and derives its local Fisher curvature.
The Asimov convention follows Ref.~\cite{CowanCranmerGrossVitells2011}.
Coverage calibration remains separate because periodic phases, discrete
ordering, boundaries, degeneracies, and constrained nuisances can violate the
regularity conditions of a naive Wilks approximation
\cite{FeldmanCousins1998,NOvAProfiledFC2022}.

\subsection{Poisson likelihood with auxiliary constraints}
\label{app:dune-exact-likelihood}

Let $n_b$ be the observed count in reconstructed analysis bin $b$, and let
$\bm\lambda$ denote the oscillation parameters.  The nine AEDL normalization
coordinates identified in Appendix~\ref{app:dune-systematics} are written as
standardized pulls $\bm\xi$.  If $\bm c$ denotes the measured centers of
their unit-Gaussian auxiliary constraints, the rate model is
\begin{equation}
  \mu_b(\bm\lambda,\bm\xi)
  =
  \mu_b^0(\bm\lambda)
  +
  \sum_{k=1}^{9}
  \xi_k A_{kb}(\bm\lambda),
  \label{eq:appF-linear-rate-model}
\end{equation}
where $A_{kb}$ already includes the fractional AEDL error and the sum of all
channel occurrences carrying systematic name $k$.  Thus repeated AEDL names
remain one coherent nuisance coordinate.

Up to terms independent of the prediction, twice the negative joint
log-likelihood ratio is
\begin{align}
  \chi^2(\bm\lambda,\bm\xi;\bm n,\bm c)
  ={}&
  2\sum_b
  \left[
    \mu_b-n_b
    +
    n_b\ln\!\left(\frac{n_b}{\mu_b}\right)
  \right]
  \nonumber\\
  &+
  (\bm\xi-\bm c)^{\mathsf T}
  (\bm\xi-\bm c),
  \label{eq:appF-joint-deviance}
\end{align}
with $0\ln(0/\mu_b)\equiv0$ and the fit restricted to $\mu_b>0$ for every
bin.  The physics profile is
\begin{equation}
  \chi^2_{\mathrm{prof}}(\bm\lambda;\bm n,\bm c)
  =
  \min_{\bm\xi:\,\bm\mu>0}
  \chi^2(\bm\lambda,\bm\xi;\bm n,\bm c).
  \label{eq:appF-profile-definition}
\end{equation}

The conditional auxiliary-center ensemble fixes $\bm c=\bm0$.  In the joint
auxiliary-measurement ensemble, a new
$\bm c\sim\mathcal N(\bm0,\mathbb I_9)$ is drawn for each pseudo-experiment.
These are different repeated-sampling
prescriptions and need not have identical coverage.  The draw formally
includes the center of the inactive tau-background coordinate; as shown
below, that coordinate profiles to its drawn center and therefore has no
effect on the event-rate likelihood or the profile statistic.

\subsection{Strict convexity and the inactive tau coordinate}
\label{app:dune-nuisance-convexity}

Write $\mathsf A$ for the $9\times264$ matrix with entries $A_{kb}$.
Using componentwise division in the score, one half of the nuisance gradient
and Hessian are
\begin{align}
  \frac{1}{2}\nabla_{\bm\xi}\chi^2
  &=
  \mathsf A
  \left(
    \bm1-\frac{\bm n}{\bm\mu}
  \right)
  +
  \bm\xi-\bm c,
  \label{eq:appF-nuisance-score}
  \\
  \frac{1}{2}
  \nabla^2_{\bm\xi\bm\xi}\chi^2
  &=
  \mathbb I_9
  +
  \mathsf A
  \operatorname{diag}\!\left(
    \frac{n_b}{\mu_b^2}
  \right)
  \mathsf A^{\mathsf T}
  \succeq
  \mathbb I_9.
  \label{eq:appF-nuisance-hessian}
\end{align}
The positive-mean domain is the intersection of the affine half-spaces
$\mu_b(\bm\lambda,\bm\xi)>0$ and is therefore convex.  The Gaussian term
makes the Hessian strictly positive definite throughout that domain.
Consequently, any interior stationary solution is the unique global nuisance
minimum over the positive-mean domain.

We obtain this solution with a positivity-preserving damped Newton iteration.
The Newton step is shortened before any predicted mean can cross zero and is
then backtracked until the objective satisfies an Armijo decrease condition.

As established in Appendix~\ref{app:dune-systematics}, the row associated with
\texttt{err\_nutau\_bg} is identically zero for the supplied public tables.
The corresponding score equation reduces to
\begin{equation}
  \widehat\xi_{\mathrm{nutau}}
  =
  c_{\mathrm{nutau}}.
  \label{eq:appF-inactive-tau-pull}
\end{equation}
It contributes a unit eigenvalue to the nuisance Hessian but no event-rate
direction and no likelihood penalty at its optimum.  Therefore the exact
likelihood has nine declared nuisance coordinates and eight active event-rate
directions in both the fixed-center and fluctuated-center ensembles.

At the Asimov generating point,
$\bm n=\bm\mu^0(\bm\lambda_\star)$ and $\bm c=\bm0$.  The calculation gives
zero deviance and zero pulls; the nuisance-gradient infinity norm is
$1.02\times10^{-14}$.
The smallest eigenvalue of the nuisance Hessian in
Eq.~\eqref{eq:appF-nuisance-hessian} is unity to machine precision.

\subsection{Local equivalence to the profiled Fisher matrix}
\label{app:dune-profile-curvature}

Let
\begin{equation}
  \mathsf D_{ib}
  =
  \left.
  \frac{\partial\mu_b^0}{\partial\lambda_i}
  \right|_{\bm\lambda_\star},
  \qquad
  \mathsf W
  =
  \operatorname{diag}\!\left(
    \frac{1}{\mu_b^\star}
  \right).
  \label{eq:appF-rate-jacobian}
\end{equation}
The joint expected-information blocks at the Asimov point are
\begin{align}
  \mathsf F_{\lambda\lambda}
  &=
  \mathsf D\mathsf W\mathsf D^{\mathsf T},
  &
  \mathsf F_{\lambda\xi}
  &=
  \mathsf D\mathsf W\mathsf A^{\mathsf T},
  \\
  \mathsf F_{\xi\xi}
  &=
  \mathbb I_9
  +
  \mathsf A\mathsf W\mathsf A^{\mathsf T}.
  \label{eq:appF-joint-fisher-blocks}
\end{align}
Exact nuisance profiling gives the Schur complement
\begin{equation}
  \mathsf F_{\mathrm{prof}}
  =
  \mathsf F_{\lambda\lambda}
  -
  \mathsf F_{\lambda\xi}
  \mathsf F_{\xi\xi}^{-1}
  \mathsf F_{\xi\lambda},
  \label{eq:appF-profiled-schur}
\end{equation}
and the local likelihood identity
\begin{equation}
  \left.
  \frac{1}{2}
  \nabla^2_{\bm\lambda\bm\lambda}
  \chi^2_{\mathrm{prof}}
  \right|_{\bm\lambda_\star}
  =
  \mathsf F_{\mathrm{prof}}.
  \label{eq:appF-hessian-fisher-identity}
\end{equation}

We validate Eq.~\eqref{eq:appF-hessian-fisher-identity} twice.  First, define
the scale vector and its diagonal scale matrix by
\begin{equation}
  \bm s
  =
  \left(
    1,
    1,
    1,
    1,
    10^{-4}\,\mathrm{eV}^{2},
    10^{-3}\,\mathrm{eV}^{2}
  \right)^{\mathsf T},
  \qquad
  S
  =
  \operatorname{diag}(\bm s),
  \label{eq:appF-curvature-scale-vector}
\end{equation}
and introduce dimensionless local coordinates through
\begin{equation}
  \bm\lambda
  =
  \bm\lambda_\star+S\bm x.
  \label{eq:appF-curvature-scaled-coordinates}
\end{equation}
A 73-point central Hessian of the independently minimized profile is then
evaluated using the common dimensionless coordinate step
\begin{equation}
  h_{\mathrm{curv}}
  =
  3.0\times10^{-5}.
  \label{eq:appF-curvature-step}
\end{equation}
For six coordinates, the central construction uses the generating point, the
12 one-coordinate displacements, and the 60 mixed-coordinate displacements,
giving $1+12+60=73$ profile evaluations.  The resulting Hessian agrees with
the analytic Schur complement, transformed to the same scaled coordinates, to
relative $L_2$ error $7.75\times10^{-9}$.  All 73 nuisance fits converge,
and the largest nuisance-gradient norm is $6.39\times10^{-7}$.

Second, the separately compiled GLoBES 3.2.18
\texttt{chiMultiExp} profile gives relative error $3.26\times10^{-6}$.
The latter comparison independently checks the spectrum, the repeated-name
correlations, and the likelihood normalization.

\subsection{Scope of the nonlocal calculations}

The finite-bank statistic, paired pseudo-experiment ensemble, and expanded
four-coordinate profile are defined and reported once in
Sec.~\ref{sec:dune}.  The exact likelihood and convex nuisance profile above
establish the statistical objective used in those calculations; they do not
supply a calibrated global confidence construction.  Such a construction
would require continuously or adaptively converged profiles, a declared
external-constraint model, and Monte Carlo calibration at multiple regular
and degenerate truths in both orderings
\cite{NOvAProfiledFC2022}.

\bibliographystyle{apsrev4-2}
\bibliography{references}

\end{document}